\documentclass[12pt]{article}

\usepackage[T1]{fontenc}
\usepackage[utf8]{inputenc}
\usepackage{amsmath,amsfonts,amsthm,amssymb}
\usepackage{booktabs,tabularx}
\usepackage{float}
\usepackage{xcolor}
\usepackage{geometry}
\usepackage{enumerate}
\usepackage{graphicx,color}
\usepackage{hyperref}
\usepackage{mathtools}
\usepackage{setspace}
\usepackage{cite}

\makeatletter
\renewcommand*\env@matrix[1][\arraystretch]{%
  \edef\arraystretch{#1}%
  \hskip -\arraycolsep
  \let\@ifnextchar\new@ifnextchar
  \array{*\c@MaxMatrixCols c}}
\makeatother

\newtheorem{theorem}{Theorem}[section]
\newtheorem{lemma}[theorem]{Lemma}
\newtheorem{corollary}[theorem]{Corollary}
\newtheorem{proposition}[theorem]{Proposition}

\newtheorem{definition}{Definition}[section]
\newtheorem{remark}[definition]{Remark}

\newcommand{\etal}{\textrm{et al.}}
\newcommand{\ignore}[1]{}

\title{Homeostasis in Networks with Multiple Input Nodes and Robustness in Bacterial Chemotaxis}

\renewcommand*{\thefootnote}{\arabic{footnote}}

\author{
\renewcommand*{\thefootnote}{\arabic{footnote}}
Jo\~ao Luiz de Oliveira Madeira\footnotemark[1]
\and
\renewcommand*{\thefootnote}{\arabic{footnote}}
Fernando Antoneli\footnotemark[2]
}

\date{}

\begin{document}

\maketitle

\renewcommand*{\thefootnote}{\fnsymbol{footnote}}

\footnotetext{
*\parbox[t]{5in}{
Fernando Antoneli

\href{mailto:fernando.antoneli@unifesp.br}%
{fernando.antoneli@unifesp.br} 

\smallskip

Jo\~ao Luiz de Oliveira Madeira

\href{mailto:joaoluizoliveiramadeira@gmail.com}%
{jldom20@bath.ac.uk} 

\medskip}
}

\renewcommand*{\thefootnote}{\arabic{footnote}}

\footnotetext[1]
{University of Bath,
Department of Mathematical Sciences,
Bath, BA2 7AY, UK}

\footnotetext[2]{
Universidade Federal de S\~ao Paulo,
CEBIME -- EPM,
S\~ao Paulo, 04039-032, Brazil}

\begin{abstract} \footnotesize
%{\color{red}
A biological system achieve \emph{homeostasis} when there is a regulated quantity that is maintained within a narrow range of values. Here we consider homeostasis as a phenomenon of network dynamics. In this context, we improve a general theory for the analysis of homeostasis in network dynamical systems with distinguished input and output nodes, called `input-output networks'. The theory allows one to define `homeostasis types' of a given network in a `model independent' fashion, in the sense that the classification depends on the network topology rather than on the specific model equations. Each `homeostasis type' represents a possible mechanism for generating homeostasis and is associated with a suitable `subnetwork motif' of the original network. Our contribution is an extension of the theory to the case of networks with multiple input nodes. To showcase our theory, we apply it to bacterial chemotaxis, a paradigm for homeostasis in biochemical systems. By considering a representative model of \emph{Escherichia coli} chemotaxis, we verify that the corresponding abstract network has multiple input nodes. Thus showing that our extension of the theory allows for the inclusion of an important class of models that were previously out of reach. Moreover, from our abstract point of view, the occurrence of homeostasis in the studied model is caused by a new mechanism, called \emph{input counterweight homeostasis}. This new homeostasis mechanism was discovered in the course of our investigation and is generated by a balancing between the several input nodes of the network -- therefore, it requires the existence of at least two input nodes to occur. Finally, the framework developed here allows one to formalize a notion of `robustness' of homeostasis based on the concept of `genericity' from the theory dynamical systems. We discuss how this kind of robustness of homeostasis appears in the chemotaxis model.
%}

\smallskip


\end{abstract}

\newpage

\setcounter{tocdepth}{2}
\tableofcontents

\section{Introduction}

%{\color{red}
The idea that an organism should keep its internal parameters within a strict range despite changes in the external environment to guarantee its survival, was first introduced by Claude Bernard in the $19^{\textrm{th}}$ century \cite{modell15}. This property is present in many biological systems, and it was called \textit{homeostasis} by Walter Cannon in 1929 \cite{modell15}.

Although the notion of homeostasis has appeared in the study of physiology of multicelular organisms, it has a much broader scope today. From a mathematical point of view, homeostasis can be defined in several distinct contexts. For example, one may consider homeostasis with respect to genetic variants or polymorphisms (i.e., in systems such that the phenotype is robust to genetic variation) \cite{nr14}. This is related to the relationship between the enzymatic activity variation promoted by the polymorphism and the impact of it on the relevant phenotype. Likewise, one may consider homeostasis in discrete stochastic biochemical systems by observing the adaptability of the system over finite time intervals \cite{enciso2016transient}. Or one can study collective behavior, rather than the individual one, in response to an external parameter \cite{marquez2011stochastic}.

Consider a biological system model as a dynamical system with input parameter $I$ which varies over an open interval $]I_{1}, I_{2}[$. Consider an output observable variable such that for each $I \in ]I_{1}, I_{2}[$, there is a stable equilibrium where the value of the observable is $x_{o}(I)$. In this situation, it is reasonable to say that the system would exhibit \emph{homeostasis} if after changing the input variable $\mathcal{I}$, the value of the observable $x_o(\mathcal{I})$ at the equilibrium remains approximately constant. 

Another concept related to homeostasis is \emph{Adaptation}, which is widely studied in synthetic biology and control engineering (cf.~\cite{ma09,ang13,tang16,araujo18,aoki19}). Adaptation is defined as the ability of the system to reset $x_o(\mathcal{I})$ to its pre-stimulated output level (its {\em set point}) after responding to an external change on the stimulus $\mathcal{I}$. Hence, adaptation is essentially equivalent to homeostasis. There are two formulations usually considered in the research on adaptation: (1) the strict condition of \emph{perfect adaptation}, where the observable $x_o(\mathcal{I})$ is required to be constant over a range of external stimuli $\mathcal{I}$; (2) the more general condition of \emph{near-perfect adaptation}, where the observable $x_o(\mathcal{I})$ is required to be within a narrow interval of values over a range of external stimuli $\mathcal{I}$.

A similar formulation proposed by Golubitsky and Stewart \cite{gs17, gs18} -- motivated by previous work on biochemical networks \cite{reed17} (see also \cite{nrbu04,best09,nr14,nbr15,nbr18}) -- employs methods from singularity theory to define the notion of `infinitesimal homeostasis'. Their framework can also be applied to other biological systems, such as gene expression modeled in a deterministic way \cite{antoneli18}. According to this approach, a system exhibits \emph{infinitesimal homeostasis} if $\frac{dx_{o}}{d I}(I_{0}) = 0$ for some input value $I_{0}$, where $x_{o}$ is the function that associates to each input parameter a unique value of the observable at the equilibrium. Infinitesimal homeostasis generalizes the notion of perfect adaptation (or perfect homeostasis), since condition (1) is equivalent to $\frac{dx_{o}}{d I} \equiv 0$. On the other hand, many systems exhibiting near-perfect adaptation (or near-perfect homeostasis) do not have any $I_0$ where $\frac{dx_{o}}{d I}(I_{0}) = 0$ (e.g. product inhibition \cite{reed17}). We shall not pursue the study of near-perfect homeostasis here.

In this paper, we shall focus on the occurrence of infinitesimal homeostasis in networks of dynamical systems with distinguished input and output nodes. More precisely, Wang \etal~\cite{wang20} introduced the notion of `abstract input-output network' and obtained a method to classify infinitesimal homeostasis in networks with a single input node and a single input parameter affecting this input node. They introduced a notion of `infinitesimal homeostasis types' corresponding to the `mechanisms' that are responsible for the occurrence of infinitesimal homeostasis in a network. These `homeostasis types' were further divided in two `homeostasis classes', called \emph{appendage} and \emph{structural}, which correspond respectively to feedback and feed-forward mechanisms. Although the results of~\cite{wang20} completely covered the classification of infinitesimal homeostasis of networks with one input and one output nodes, they can not be applied to systems with more than one input node. This limitation excludes several important mathematical models. One such exclusion, that are known to exhibit perfect homeostasis, is the mathematical model of bacterial chemotaxis systems.

Generally speaking, \emph{bacterial chemotaxis} refers to the ability of bacteria to sense changes in their extracellular environment and to bias their motility towards favorable stimuli (attractants) and away from unfavorable stimuli (repellents). There exist a number of different but related bacterial chemosensory systems, of which the most widely studied is the \textit{Escherichia coli} chemotaxis \cite{B04}. Extensive mathematical modeling has described different aspects of the chemotaxis pathway and has mainly focused on explaining the initial response to addition of attractant, as well as robustness of perfect homeostasis. For instance, \cite{Alon99,Barkai97} showed that perfect homeostasis is robust (insensitive to parameter variations in the pathway), if the kinetics of receptor methylation depends only on the activity of receptors and not explicitly on the receptor methylation level or external chemical concentration. Their idea was later extended by others, providing conditions for perfect homeostasis \cite{mello03,Yi00}, as well as robustness to noise by the network architecture~\cite{kollmann05,hansen08}.

Our goal in this paper is to extend the theory of \cite{wang20} to input-output networks with multiple input nodes and a single input parameter affecting these input nodes. More precisely, in our extended version of the theory, the same structure of the classification is maintained, except that there is a new `homeostasis class'.  This new class, called \emph{input counterweight}, does not exist in the single input node case. The corresponding homeostasis mechanism is generated by a balancing between the several input nodes of the network. That explains why one needs at least two input nodes to detect it. Finally, we show that our generalization allows us to analyze a representative model for chemotaxis that has good agreement with experimental findings, including occurrence of perfect homeostasis \cite{clausznitzer10,tindall08a,tindall08b}.

The main idea in our approach is to reduce to the single input-output case. For each input node, we  consider the corresponding single input-output network (by `forgetting' the action of the input parameter on the other input nodes) and apply the results of~\cite{wang20} to each of these networks. The novelty consists in showing that the properties of the individual single input-output networks are compatible in a certain sense. This is sufficient to obtain a classification for the original multiple input-output network. The approach introduced here can be further extended to the case of input-output networks with multiple input nodes and multiple input parameters.
This will be pursed in another publication.
%}

\subsection*{Mathematical Modeling of Chemotaxis}

%{\color{red}
The mathematical modeling of chemotaxis can be roughly divided into two types: single cell models and bacterial population models~\cite{tindall08a,tindall08b}. Single cell models consider the activation of the flagellar motor by detection of attractants and repellents in the extracellular medium. The flagellar motor activity of bacteria is regulated by a signal transduction pathway, which integrates changes of environmental chemical concentrations into a behavioral response. Assuming mass-action kinetics, the reactions in the signal transduction pathway can be modeled mathematically by ODEs. The bacterial population models describe evolution of bacterial density by parabolic PDEs involving an anti-diffusion `chemotaxis term' proportional to the gradient of the chemoattractant, thus allowing movement up-the-gradient, the most prominent feature of chemotaxis. The most extensively studied bacterial population models are the Patlak-Keller-Segel (PKS) type models. 
We will only consider single cell models here.

Understanding the response of \textit{E. coli} cells to external attractants has been the subject of experimental work and mathematical models for nearly 40 years. In fact, many models of the chemotaxis have been formulated and developed to provide a comprehensive description of the cellular processes and include details of receptor methylation, ligand-receptor binding and its subsequent effect on the biochemical signaling cascade, along with a description of motor driving CheY/CheY-P levels, the main output of the chemotaxis system  (see~\cite{tindall08a} for a survey).
%}

However, including such detail has often led to very complex mathematical models consisting of tens of governing differential equations, making mathematical analysis of the underlying cellular response difficult, if not in many cases, impossible. A model proposed by Clausznitzer \etal~\cite{clausznitzer10} has sought to provide a comprehensive description of the \textit{E. coli} response, by coupling a simplified statistical mechanical description of receptor methylation and ligand binding, with the signaling cascade dynamics. The model consists of five nonlinear ordinary differential equations (ODEs) and is parameterized using data from the literature. The authors were able to show that the model is in good agreement with experimental findings. However, being a fifth-order nonlinear ODE model, it is difficult to treat analytically. More recently, Edgington and Tindall \cite{edgington18} undertook a comprehensive mathematical analysis of a number of simplified forms of the model of~\cite{clausznitzer10} and proposed a fourth-order reduction of this model that has been used previously in the theoretical literature~\cite{edgington15}.

In the following we shall consider the model proposed by \cite{edgington15,edgington18}. It has four variables for the concentrations of CheA/CheA-P ($a_p$), CheY/CheY-P ($y_p$), CheB/CheB-P ($b_p$) and the receptor methylation ($m$) and is given by the following system of ODEs (in non-dimensional form):
\begin{equation} \label{original_e_coli}
\begin{aligned}
& \frac{d m}{d t} = \gamma_{R}\, (1 - \phi(m,L)) - \gamma_{B} \, \phi(m,L) \, b_{p}^{2} \\
& \frac{d a_{p}}{d t} = \phi(m,L) \, k_{1} \, (1 - a_{p}) - k_{2} \, (1 - y_{p})\, a_{p} - k_{3} \, (1 - b_{p}) \, a_{p} \\
& \frac{d y_{p}}{d t} = \alpha_{1} \, k_{2} \, (1 - y_{p})a_{p} - k_{4} \, y_{p} \\
& \frac{d b_{p}}{d t} = \alpha_{2} \, k_{3} \, (1 - b_{p}) \, a_{p} - k_{5} \, b_{p}
\end{aligned}
\end{equation}
where $\gamma_B$, $\gamma_R$, $k_1,\ldots,k_5$ are non-dimensional parameters, the extracellular ligand concentration $L$ is the \emph{external parameter} and the function $\phi$ is determined by a  Monod–Wyman–Changeux (MWC) description of receptor clustering
\begin{equation} \label{definition_phi}
    \phi (m,L) = \frac{1}{1 + e^{F(m,L)}} 
    \qquad\text{with}\qquad 
    F(m,L) = N\left[ 1 - \frac{m}{2} + \log \left( \frac{1 + \frac{L}{K_{a}^{\textrm{off}}}}{1 + \frac{L}{K_{a}^{\textrm{on}}}}\right) \right]
\end{equation}

The key observation of~\cite{edgington18} is that system \eqref{original_e_coli} has a unique asymptotically stable equilibrium $X^*=(m^*,a_p^*,y_p^*,b_p^*)$, with $a_p^*$, $y_p^*$ and $b_p^*$ positive and $m^{*}$ is a real number, for the non-dimensional parameters obtained from the parameter values originally used in~\cite{clausznitzer10}. Furthermore, \cite{edgington18} were able to show that some pairs of parameters might yield oscillatory behavior, but in regions of parameter space outside that observed experimentally. This was done by carrying out the stability analysis for pair-wise parameter variations, whereby for each case the occurrence of at least two non-zero imaginary parts was recorded as indicating possible oscillatory dynamics. The steady-state $X^*$ can be easily found by numerical integration for parameter values that are experimentally valid, although some combinations of parameters produce a large stiffness coefficient~\cite{clausznitzer10}.

More importantly, the stability of $X^*$ persists as $L$ is varied in the range $(0,+\infty)$. By standard arguments (see subsection 2.1) this implies that there is a well-defined smooth mapping $L\mapsto X^*(L)=\big(m^*(L),a_p^*(L),y_p^*(L),b_p^*(L)\big)$. Since the values of $a_p^*$, $y_p^*$ and $b_p^*$ are independent of $L$ (see~\cite{edgington18}), it follows that the individual component functions $a_p^*(L)$, $y_p^*(L)$ and $b_p^*(L)$ are actually constant functions with respect to $L$.

\begin{figure}[!ht]
\centering
\includegraphics[width=\linewidth,trim=0cm 1cm 0cm 2.5cm,clip=true]{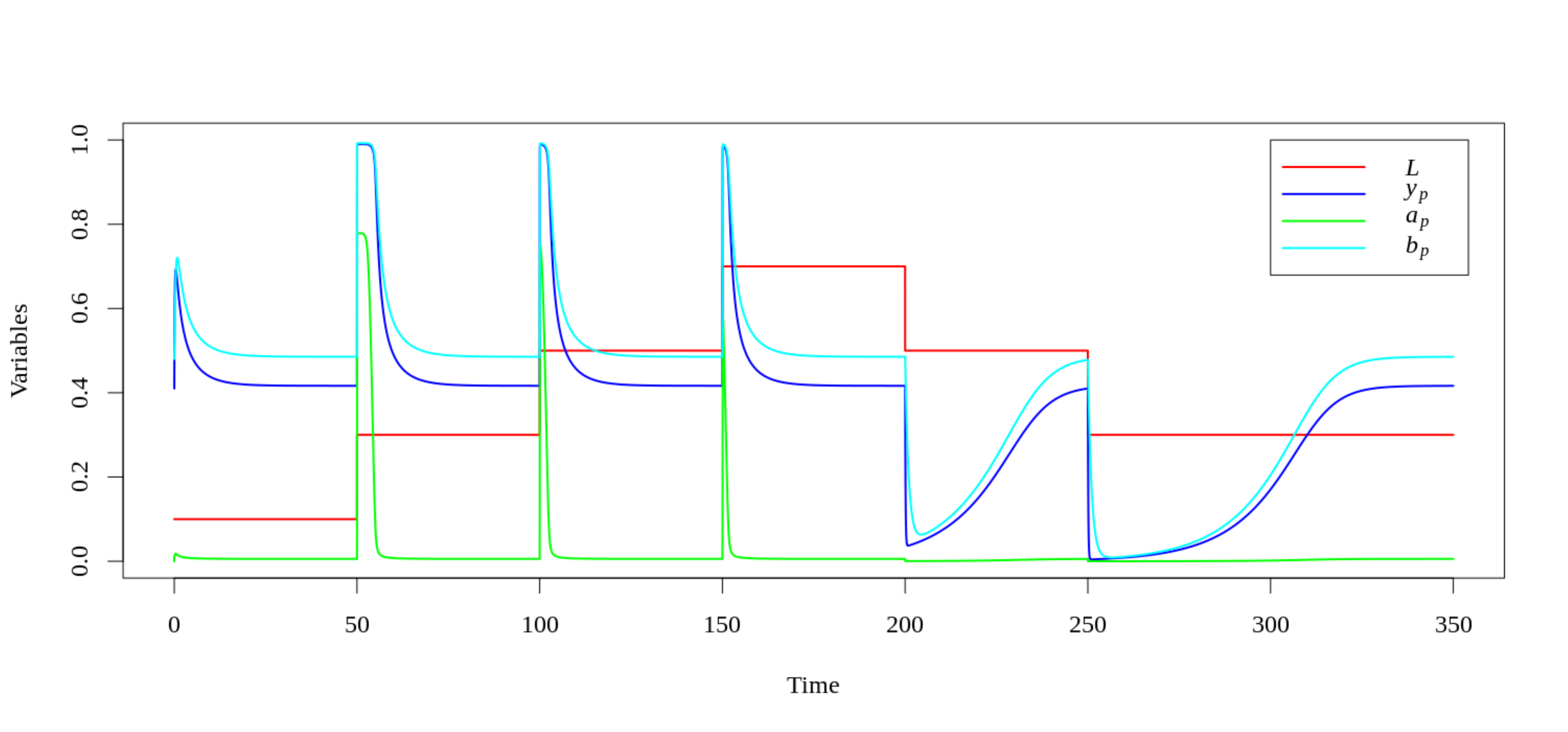}
\caption{\label{e_coli_plot1} Time series of the model \eqref{original_e_coli}, showing perfect homeostasis of the three variables $a_p$ (green), $b_p$ (cyan), $y_p$ (blue) at the non-dimensional equilibrium given by $a_p^* = 5.5 \times 10^{-3}$, $b_p^* = 0.48$, $y_p^* = 0.41$. Input parameter $L$ is given by a step function (red curve). The parameters were set to non-dimensional values of \cite[Table~2]{edgington18}.  Time series were computed using the software \textsc{XPPAut} \cite{bard02}.}
\end{figure}

In Figure~\ref{e_coli_plot1} we show the time series of the three variables $a_p$, $y_p$ and $b_p$ and how they are perturbed when the parameter $L$ is varied by sudden jumps. After a transient, which depends on the contraction rates at the equilibrium point, each variable returns to its corresponding steady-state value. Moreover, when the parameters $\gamma_B$, $\gamma_R$, $k_1,\ldots,k_5$ are changed the equilibrium $X^*$ also changes, but the new functions $a_p^*(L)$, $y_p^*(L)$ and $b_p^*(L)$ remain constant with respect to $L$, possibly with different values.  This is exactly the robust perfect homeostasis that was shown to occur in other models of chemotaxis \cite{Barkai97,Alon99,mello03,Yi00,kollmann05,hansen08}.

%{\color{red}
We shall consider this model in Section \ref{sec:e_coli}, in the light of the newly developed theory. We show that the network corresponding to the model equations \ref{original_e_coli} has two input nodes and one output node and thus cannot be analyzed using the methods of \cite{wang20}. In fact this can be seen directly from the equations \ref{original_e_coli}, since both the equation for $m$ and $a_p$ depend on the external parameter $L$.
%}

\subsection*{Singularity Approach to Homeostasis} 

\ignore{
The result shown in Figure \ref{e_coli_plot1} is expected, since the model \ref{original_e_coli} is closely related to the original models \cite{Alon99,Barkai97} where robust perfect homeostasis was discovered in the first place. Nevertheless, it raises an interesting question: \emph{How come that several distinct mathematical models, with varying degrees of detail and distinct parameters are capable of exhibiting the same behavior, as far as homeostasis (or adaptation) is concerned?} From the biological point of view it is a trivial question, since all these models are supposed to describe the same phenomenon, and hence, should yield the same predictions. Nevertheless, it is not a trivial \emph{mathematical} issue. One could argue that these models are related among themselves by certain mathematical procedures, for instance, `reduction of variables', `averaging', etc., but it is not easy to show \emph{a priori}, that any of these procedures preserve the property in question. Of course, after the procedure is applied, one can always check \emph{a posteriori} if the property in question was lost or not.
}

%{\color{red}
Undoubtedly, an important concern when analyzing a theoretical model given by a dynamical system is whether the properties of interest are preserved under perturbations (of a certain type). For instance, in systems biology and control engineering phenomena are usually modeled by a parametric family of ODEs and a relevant concept `robustness' means that a property is \emph{robust} if it is preserved by changes in the parameters of the model. This is exactly the notion of robustness that is employed in the literature on perfect homeostasis of chemotaxis \cite{Alon99,Barkai97,Yi00,mello03,clausznitzer10,edgington18}.

In singularity theory one is concerned with parametric families of maps or dynamical systems and the role of `robustness' is played by the notion of \emph{genericity} -- structure preservation by all small perturbations of the whole parametric family (in the appropriate function space). In this case the small perturbations are as general as possible, and not restricted to variation of parameters within a fixed family.

It is clear that perfect homeostasis is not a robust in the wider sense above.  Nevertheless, we show in Section \ref{sec:e_coli} that perfect homeostasis is robust in the model equations \eqref{original_e_coli} for a much larger set of perturbations than the obtained by changes in the parameters of the model. Moreover, our analysis suggests that weaker versions of it, namely `infinitesimal homeostasis' and `near perfect homeostasis', are generic for the model equations \eqref{original_e_coli}.
%}

\paragraph{Structure of the paper.}
The remainder of this paper is divided in three parts.  In Section \ref{sec:homeostasis} we present the theory of infinitesimal homeostasis for networks with multiple input nodes.  In Section \ref{mathematical_arguments_and_proofs} we provide detailed proofs of the theorems of Section \ref{sec:homeostasis}. In Section \ref{sec:e_coli} we apply the general theory to study homeostasis in the model equations \eqref{original_e_coli} for bacterial chemotaxis. Sections \ref{mathematical_arguments_and_proofs} and \ref{sec:e_coli} can be read independently from each other.

\section{Homeostasis in Coupled Systems}
\label{sec:homeostasis}

In this section we define the basic objects of the theory: multiple input nodes input-output networks, network admissible systems of differential equations, input-output functions and infinitesimal homeostasis points. Then we introduce the generalised homeostasis matrix and show how to use it to find infinitesimal homeostasis points and to classify homeostasis types. Finally, we relate the classification of homeostasis types with the topology of the network by associating `subnetwork motifs' to the irreducible factors of the determinant of the generalised homeostasis matrix and present an algorithm to find all these factors in terms of the `subnetwork motifs' constructed before.

\subsection{A Dynamical Systems Formalism for Homeostasis}

Golubitsky and Stewart proposed a mathematical method for the study of homeostasis based on dynamical systems theory~\cite{gs17,gs18} (see the review \cite{gsahy20}). 
In this framework, one considers a system of differential equations 
\begin{equation} \label{general_dynamics}
  \dot{X} = F(X, \mathcal{I})
\end{equation}
where $X = (x_{1}, \cdots, x_{k}) \in \mathbb{R}^{k}$ and parameter $\mathcal{I}\in\mathbb{R}$ represents the external input to the system.

Suppose that $(X^*, \mathcal{I}^*)$ is a linearly stable equilibrium of \eqref{general_dynamics}. By the implicit function theorem, there is a function $\tilde{X}(\mathcal{I})$ defined in a neighborhood of $\mathcal{I}^*$ such that $\tilde{X}(\mathcal{I}^*) = X^*$ and $F(\tilde{X}(\mathcal{I}), \mathcal{I}) \equiv 0$. The simplest case is when there is a variable, let's say $x_{k}$, whose output is of interest when $\mathcal{I}$ varies. Define the associated \emph{input-output function} as $z(\mathcal{I})=\tilde{x}_k(\mathcal{I})$.

The input-output function allows one to formulate two of the most used definitions that capture the notion of homeostasis \cite{ma09,ang13,tang16}.

\begin{definition} \normalfont
Let $z(\mathcal{I})$ be the input-output function associated to a system of differential equations \eqref{general_dynamics}. We say that the corresponding system \eqref{general_dynamics} exhibits
\begin{enumerate}[(a)]
\item \emph{Perfect Homeostasis (Adaptation)} on the interval $]\mathcal{I}_{1}, \mathcal{I}_{2}[$ if
\begin{equation} \label{definition_perfect_adaptation}
  \frac{d z}{d \mathcal{I}} (\mathcal{I}) = 0 
  \qquad\text{for all} \; \mathcal{I} \in ]\mathcal{I}_{1}, \mathcal{I}_{2}[
\end{equation}
That is, $z$ is constant on $]\mathcal{I}_{1}, \mathcal{I}_{2}[$.
\item \emph{Near-perfect Homeostasis (Adaptation)} relative to the point $\mathcal{I}_0$ on the interval $]\mathcal{I}_{1}, \mathcal{I}_{2}[$ if for a fixed $\delta$
\begin{equation} \label{definition_near_perfect_adaptation}
  | z(\mathcal{I}) - z(\mathcal{I}_0) | \leqslant \delta
  \qquad\text{for all} \; \mathcal{I} \in ]\mathcal{I}_{1}, \mathcal{I}_{2}[
\end{equation}
That is, $z$ stays within $z(\mathcal{I}_0)\pm\delta$ on $]\mathcal{I}_{1}, \mathcal{I}_{2}[$.
\end{enumerate}
\end{definition}

It is clear that perfect homeostasis implies near-perfect homeostasis, but the converse does not hold. Inspired by Reed \etal \cite{nijhout14,best09}, Golubitsky and Stewart \cite{gs17,gs18} introduced another definition of homeostasis that is essentially intermediate between perfect and near-perfect homeostasis. Moreover, this new definition allows the tools from singularity theory to bear on the study of homeostasis.

\begin{definition} \rm
Let $z(\mathcal{I})$ be the input-output function associated to a system of differential equations \eqref{general_dynamics}. We say that the corresponding system \eqref{general_dynamics} displays \emph{Infinitesimal Homeostasis} at the point $\mathcal{I}_0$ on the interval $]\mathcal{I}_{1}, \mathcal{I}_{2}[$ if
\begin{equation} \label{definition_homeostasis}
  \frac{d z}{d \mathcal{I}} (\mathcal{I}_{0}) = 0
\end{equation}
\end{definition}

It is obvious that perfect homeostasis implies infinitesimal homeostasis. Moreover, it follows from Taylor's theorem that infinitesimal homeostasis implies near-perfect homeostasis in a neighborhood of $\mathcal{I}_0$. It is easy to see that the converse to both implications is not generally valid (see \cite{reed17}).

When combined with coupled systems theory~\cite{gs06} the formalism of \cite{gs17,gs18,gsahy20} becomes very effective in the analysis of model equations.

\subsection{Multiple Input-Node Input-Output Networks}

A \emph{multiple input-node input-output network} is a network $\mathcal{G}$ with $n$ distinguished \emph{input nodes} $\iota=\{\iota_{1}, \iota_{2}, \ldots, \iota_{n}\}$, all of them associated to the same input parameter $\mathcal{I}$, one distinguished \emph{output node} $o$, and $N$ \emph{regulatory nodes} $\rho=\{\rho_1,\ldots,\rho_N\}$.
The associated network systems of differential equations have the form
\begin{equation} \label{admissible_systems_ODE_multiple_input_nodes}
\begin{aligned}
\dot{x}_{\iota} & = f_{\iota}(x_{\iota}, x_{\rho}, x_{o}, \mathcal{I}) \\
\dot{x}_{\rho} & = f_{\rho}(x_{\iota}, x_{\rho}, x_{o})\\
\dot{x}_{o} & = f_{o}(x_{\iota}, x_{\rho}, x_{o})
\end{aligned}
\end{equation}
where $\mathcal{I}\in\mathbb{R}$ is an \emph{external input parameter} and $X=(x_{\iota},x_{\rho},x_o)\in\mathbb{R}^n\times\mathbb{R}^N\times\mathbb{R}$ is the vector of state variables associated to the network nodes.

We write a vector field associated with the system \eqref{admissible_systems_ODE_multiple_input_nodes} as
\[
F(X,\mathcal{I})=(f_{\iota}(X,\mathcal{I}),f_\rho(X),f_o(X))
\]
and call it an \emph{admissible vector field} for the network $\mathcal{G}$.

Let $f_{j,x_\ell}$ denote the partial derivative of the $j^{th}$ node function $f_j$ with respect to the $\ell^{th}$ node variable $x_\ell$. We make the following assumptions about the vector field $F$ throughout:
\begin{enumerate}[(a)]
\item The vector field $F$ is smooth and has an asymptotically stable equilibrium at $(X^*,\mathcal{I}^*)$. Therefore, by the implicit function theorem, there is a function $\tilde{X}(\mathcal{I})$ defined in a neighborhood of $\mathcal{I}^*$ such that $\tilde{X}(\mathcal{I}^*) = X^*$ and $F(\tilde{X}(\mathcal{I}), \mathcal{I}) \equiv 0$. 
\item The partial derivative $f_{j,x_\ell}$ can be non-zero only if the network $\mathcal{G}$ has an arrow $\ell\to j$, otherwise $f_{j,x_\ell} \equiv 0$.
\item Only the input node coordinate functions $f_{\iota_m}$ depend on the external input parameter $\mathcal{I}$ and the partial derivative of $f_{\iota_m,\mathcal{I}}$ generically satisfies
\begin{equation} \label{e:f_iota_I}
 f_{\iota_m,\mathcal{I}} \neq 0.
\end{equation}
\end{enumerate}

\begin{definition} \rm \label{IOF}
Let $\mathcal{G}$ be an input-output network with $n$ input nodes and $F$ be a family admissible vector field with equilibrium point $\tilde{X}(\mathcal{I})=\big(x_{\iota}(\mathcal{I}),x_{\rho}(\mathcal{I}),x_o(\mathcal{I})\big)$. The mapping $\mathcal{I} \mapsto x_o(\mathcal{I})$ is called the \emph{input-output function} of the network $\mathcal{G}$, relative to the family of equilibria $\tilde{X}(\mathcal{I})$.
\end{definition}

\subsection{Infinitesimal Homeostasis by Cramer's Rule}

As noted previously \cite{gs17,gsahy20,reed17,wang20}, a straightforward application of Cramer's rule gives a formula for determining infinitesimal homeostasis points. This has a straightforward generalization to multiple input networks.

Let $J$ be the $(n+N+1)\times (n+N+1)$ Jacobian matrix of an admissible vector field $F=(f_{\iota},f_{\sigma},f_{o})$, that is,
\begin{equation} \label{jacobian}
J = \begin{pmatrix}
  f_{\iota, x_\iota}   &  f_{\iota, x_\rho} & f_{\iota, x_o} \\
  f_{\rho, x_\iota}   &  f_{\rho, x_\rho} & f_{\rho, x_o} \\
  f_{o, x_\iota} &  f_{o, x_\rho} & f_{o, x_o} 
\end{pmatrix}
\end{equation}
The $(n+N+1)\times (n+N+1)$ matrix $\langle H \rangle$ obtained from $J$ by replacing the last column by $(-f_{\iota,\mathcal{I}},0,0)^t$, is called \emph{generalized homeostasis matrix}:
\begin{equation} \label{weighted_homeostasis_matrix_definition}
\langle H \rangle = 
\begin{pmatrix}
f_{\iota, x_\iota} &  f_{\iota, x_\rho} & -f_{\iota, \mathcal{I}} \\
f_{\rho, x_\iota}&  f_{\rho, x_\rho} & 0 \\
f_{o, x_\iota} &  f_{o, x_\rho} & 0
\end{pmatrix}
\end{equation}
Here all partial derivatives $f_{\ell,x_j}$ are evaluated at $\big(\tilde{X}(\mathcal{I}),\mathcal{I}\big)$.

\begin{lemma} \label{cramer_rule}
The input-output function $x_o(\mathcal{I})$ satisfies
\begin{equation} \label{xo'}
x_o'(\mathcal{I}) = \frac{\det\!\big(\langle H \rangle\big) }{\det(J)}
\end{equation}
Here $\det(J)$ and $\det\!\big(\langle H \rangle\big)$ are evaluated at $\big(\tilde{X}(\mathcal{I}),\mathcal{I}\big)$. Hence, $\mathcal{I}_0$ is a point of infinitesimal homeostasis if and only if
\begin{equation} \label{xo'_reduced}
\det\!\big(\langle H \rangle\big) = 0
\end{equation}
at the equilibrium $\big(\tilde{X}(\mathcal{I}_0),\mathcal{I}_0\big)$.
\end{lemma}

\begin{proof}
Implicit differentiation of the equation $f(\tilde{X}(\mathcal{I}),\mathcal{I})=0$ with respect to $\mathcal{I}$ yields the linear system
\begin{equation} \label{imp_diff}
J \begin{pmatrix} x_i' \\ x_\rho' \\ x_o'\end{pmatrix} =
-\begin{pmatrix}f_{\iota, \mathcal{I}} \\ 0 \\ 0 \end{pmatrix}
\end{equation}
Since $\tilde{X}(\mathcal{I})$ is assumed to be a linearly stable equilibrium, it follows that $\det(J)\neq 0$. On applying Cramer's rule to \eqref{imp_diff} we can solve for $x_o'(\mathcal{I})$ obtaining \eqref{xo'}.
\end{proof}

By expanding $\det(\langle H \rangle)$ with respect to the last column and each $\iota_k$ (input) row one obtains 
\begin{equation} \label{xo'_reduced_expand}
\det\!\big(\langle H \rangle\big) = \sum_{m=1}^n \pm f_{\iota_m,\mathcal{I}} \det(H_{\iota_m})
\end{equation}
Note that when there is a single input node, i.e. $n=1$, Lemma \ref{cramer_rule} gives the corresponding result obtained in~\cite{wang20}. In this case, there is only one matrix $H_{\iota_m}=H$, called the \emph{homeostasis matrix}, that played a fundamental role in the theory developed in \cite{wang20}. Hence, it is expected that the matrices $H_{\iota_m}$ should play a similar role in the generalization of \cite{wang20} to the multiple input node case.

\begin{definition} \label{definition_text_parcels_homeostasis_matrix} \rm
Let $\mathcal{G}$ be an input-output network with $n$ input nodes and
$f$ be an admissible vector field, with a family of equilibrium points $\tilde{X}(\mathcal{I})$. The \emph{partial homeostasis matrix} $H_{\iota_m}$ of $f$ is obtained from the Jacobian matrix $J$ of $F$ by dropping the last column and the $\iota_m$ row (see formula \eqref{definition_parcels_homeostasis_matrix}).
\end{definition}

\subsection{Classes and Types of Homeostasis}

The classification of homeostasis types proceeds as in \cite{wang20}. The first step is to apply Frobenius-K\"onig theory \cite{schneider77,br91} to the generalized homeostasis matrix $\langle H \rangle$. More precisely, Frobenius-K\"onig theory implies that there exist (constant) permutation matrices $P$ and $Q$ such that
\begin{equation} \label{weighted_homeostasis_normal_form}
    P \langle H \rangle Q = \begin{pmatrix} B_{1} & * & \cdots & * & * \\
    0 & B_{2} & \cdots & * & * \\
    \vdots & \vdots & \ddots & \vdots & \vdots \\
    0 & 0 & \cdots & B_{s} & * \\
    0 & 0 & \cdots & 0 & C
    \end{pmatrix}
\end{equation}
where each diagonal block $B_{1}$, \ldots, $B_{s}$ and $C$ is fully indecomposable (in the sense of \cite{br91}), that is, $\det(B_{1})$, \ldots, $\det(B_{s})$ and $\det(C)$ are irreducible polynomials. As $P$ and $Q$ are constant permutation matrices, we have that
\begin{equation} \label{FK_factorization}
    \det\!\big(\langle H \rangle\big) = \pm \det(B_{1}) \cdots \det(B_{s}) \cdot \det(C)
\end{equation}
In order to simplify nomenclature, we will call $B_{1}$, \ldots, $B_{s}$ and $C$ \emph{irreducible homeostasis blocks}, although in the literature the term irreducible matrix may have a different meaning (see \cite{schneider77}).

A direct comparison of factorization \eqref{FK_factorization} with expansion 
\eqref{xo'_reduced_expand} suggests that the irreducible factors $\det(B_j)$ are the common factors of $\det(H_{\iota_m})$ and $\det(C)$ is a weighted alternating sum of $f_{\iota_{1},\mathcal{I}}, \ldots, f_{\iota_{n},\mathcal{I}}$. Indeed, as we show in Section \ref{mathematical_arguments_and_proofs}, the matrix $C$ in \eqref{weighted_homeostasis_normal_form} contains all the functions $f_{\iota_{\ell}, \mathcal{I}}$ as entries, that is, it is a homogeneous polynomial of degree $1$ on $f_{\iota_{1}, \mathcal{I}}, \ldots, f_{\iota_{n}, \mathcal{I}}$, whereas the matrices $B_{1}, \ldots, B_{s}$ do not contain any of them.

The next step is to classify the irreducible homeostasis blocks of $\langle H \rangle$ according to their number number of self-couplings. Indeed, we show that each block $B_{j}$ of order $k_{j}$ has exactly $k_{j}$ or $k_{j}-1$ self-couplings (see Section \ref{mathematical_arguments_and_proofs}). But, unlike \cite{wang20}, in the multiple input nodes case we find \emph{three} classes of irreducible homeostasis blocks that may occur in core networks with multiple input nodes.

\begin{definition} \label{definition_appendage_and_structural} \rm
Let $B_{j}$ be an irreducible homeostasis block of order $k_{j}$ which does not contain any partial derivatives $f_{\iota_{m}, \mathcal{I}}$ with respect to $\mathcal{I}$. We say that the homeostasis class of $B_{j}$ is \emph{appendage} if $B_{j}$ has $k_{j}$ self-couplings and \emph{structural} if $B_{j}$ has $k_{j} - 1$ self-couplings.
\end{definition}
 
\begin{definition} \label{types_homeostasis} \rm
We say $\mathcal{G}$ exhibits \emph{appendage homeostasis} if there is an appendage irreducible homeostasis block $B_{j}$ such that $\det(B_{j}) = 0$. In an analogous way, we say $\mathcal{G}$ exhibits \emph{structural homeostasis} if there is an structural irreducible homeostasis block $B_{j}$ such that $\det(B_{j}) = 0$.
\end{definition}

As shown in Wang \etal \cite{wang20}, appendage and structural homeostasis occur in core networks with one input node. Nevertheless, networks with multiple input nodes also exhibit a new class of homeostasis that is not found in networks with only one input node.

\begin{definition} \rm
\label{definition_input_counterweight_homeostasis_block}
Let $C$ be an irreducible homeostasis block whose determinant $\det(C)$ is a homogeneous polynomial of degree $1$ on the variables $f_{\iota_{1}, \mathcal{I}}, \ldots, f_{\iota_{n}, \mathcal{I}}$. We say that the homeostasis class of $C$ is \emph{input counterweight}. Moreover, we say that $\mathcal{G}$ exhibits \emph{input counterweight homeostasis} when $\det(C) = 0$.
\end{definition}

The final step in our theory is to associate a subnetwork motif' of $\mathcal{G}$ to each homeostasis block of $\langle H \rangle$ in such a way that each class of homeostasis corresponds to a distinguished class of subnetworks. Because of the appearance of a third homeostasis class, the extension of the results of \cite{wang20} to the multiple input nodes require several new ideas. 

\subsection{Network Topology and Homeostasis} \label{combinatorial_characterization_homeostasis}

We start with the basic combinatorial definitions needed to understand the construction of `subnetwork motifs' associated with the homeostasis blocks in networks with multiple input nodes.
We notice that an example of how to apply the definitions and results of this subsection to an abstract network can be found in figure \ref{figure_definitions}.

Recall that	a node $\rho$ is \emph{downstream} from a node $\tau$ if there is a directed path from $\tau$ to $\rho$ and \emph{upstream} if there is a directed path from $\rho$ to $\tau$. We always assume that every node is downstream and upstream from itself.

\begin{remark} \rm \label{rmk:nontriviality}
From now on we assume that all networks satisfy the following condition: the output node is downstream from all input nodes. This seemingly innocuous assumption, that is implicit in every study about input-output networks, is absolutely necessary to ensure that all results are true and non-trivial (see Appendix~\ref{ap:AppB}).
\end{remark}

\begin{definition} \label{defining_core_networks} \rm
Let $\mathcal{G}$ be a network with input nodes $\iota_{1}, \ldots, \iota_{n}$ and output node $o$. We call $\mathcal{G}$ a \emph{core network} if every node in $\mathcal{G}$ is upstream from $o$ and downstream from at least one input node. Analogously, we define the core subnetwork $\mathcal{G}_{m}$ between $\iota_{m}$ and $o$ as the subnetwork composed by nodes downstream $\iota_{m}$ and upstream $o$.
\end{definition}

The main result about core networks (extending~\cite[Thm 2.4]{wang20}) is that infinitesimal homeostasis in $\mathcal{G}_c$ is `the same' as in $\mathcal{G}$ (see Theorem~\ref{theorem_homeostasis_points_are_the_same_core_and_original_networks} for details). Therefore, without loss of generality, one can consider only core networks.

We classify the nodes in a core network $\mathcal{G}$ according to their role in the topology of the network (see figure \ref{figure_definitions}).

\begin{definition} \label{combinatorial_definitions_in_network_multiple_input_nodes} \rm
Let $\mathcal{G}$ be a core network with input nodes $\iota_{1}, \iota_{2}, \ldots, \iota_{n}$ and output node $o$.
\begin{enumerate}[(a)]
\item A directed path connecting nodes $\rho$ and $\tau$ is called a \emph{simple path} if it visits each node on the path at most once.
\item An \emph{$\iota_{m}o$-simple path} is a simple path connecting the input node $\iota_{m}$ to the output node $o$.
\item A node is \emph{$\iota_{m}$-simple} if it lies on an $\iota_{m}o$-simple path.
\item A node is \emph{$\iota_{m}$-appendage} if it is downstream from $\iota_{m}$ and it is not an $\iota_{m}$-simple node.
\item A node is \emph{absolutely simple} if it is an $\iota_{m}$-simple node, for every $m = 1, \ldots, n$.
\item A node is \emph{absolutely appendage} if it is an $\iota_{m}$-appendage node, for every $m = 1, \ldots, n$.
\end{enumerate}
\end{definition}

\begin{figure}[!htp]
\centering
\begin{tabular}{@{}ll}
(a) & (b) \\
\includegraphics[width=0.475\linewidth,trim=5.5cm 3cm 5.5cm 3cm,clip=true]%
{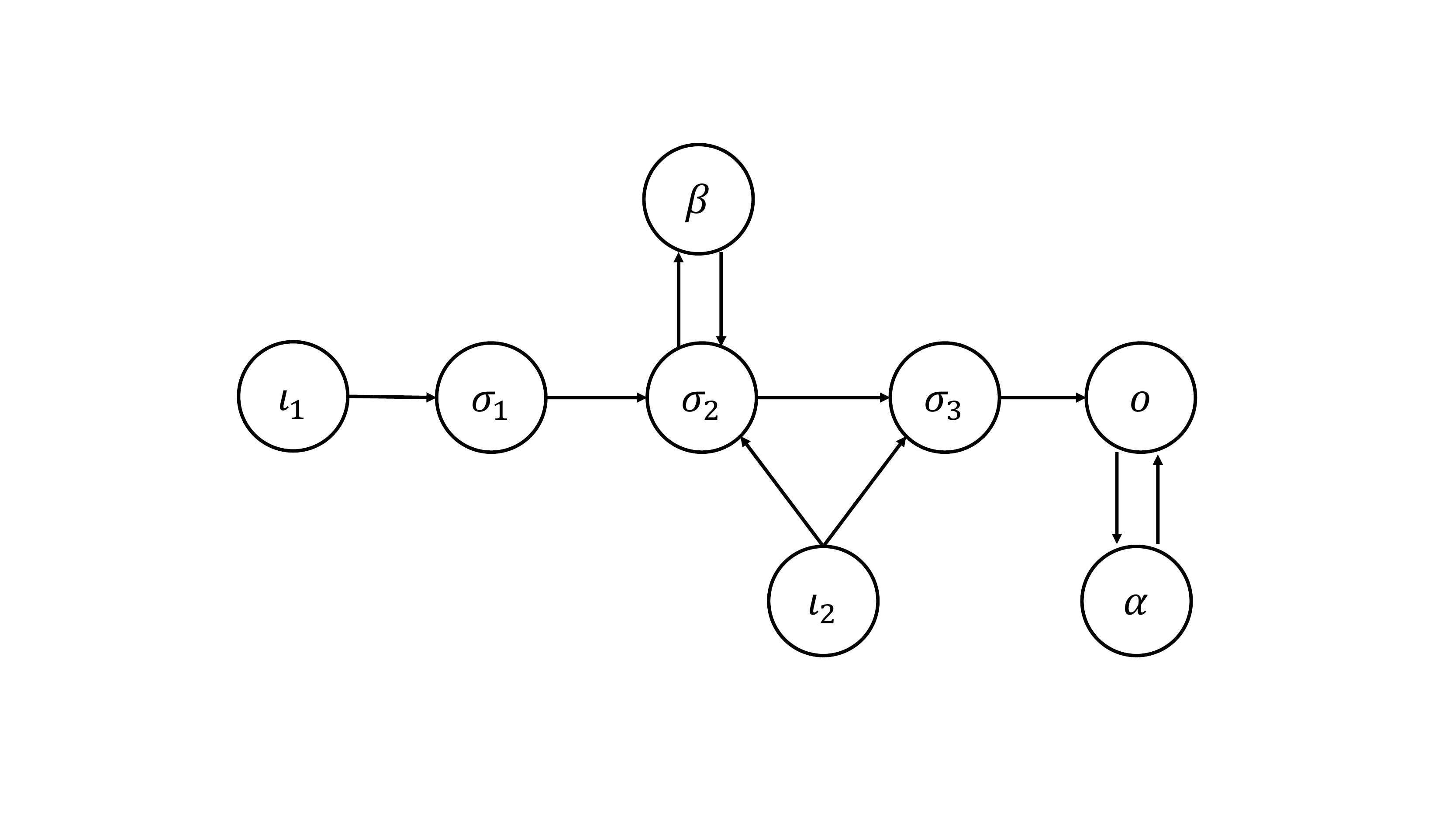} &
\includegraphics[width=0.475\linewidth,trim=4.5cm 2cm 5cm 1.7cm,clip=true]%
{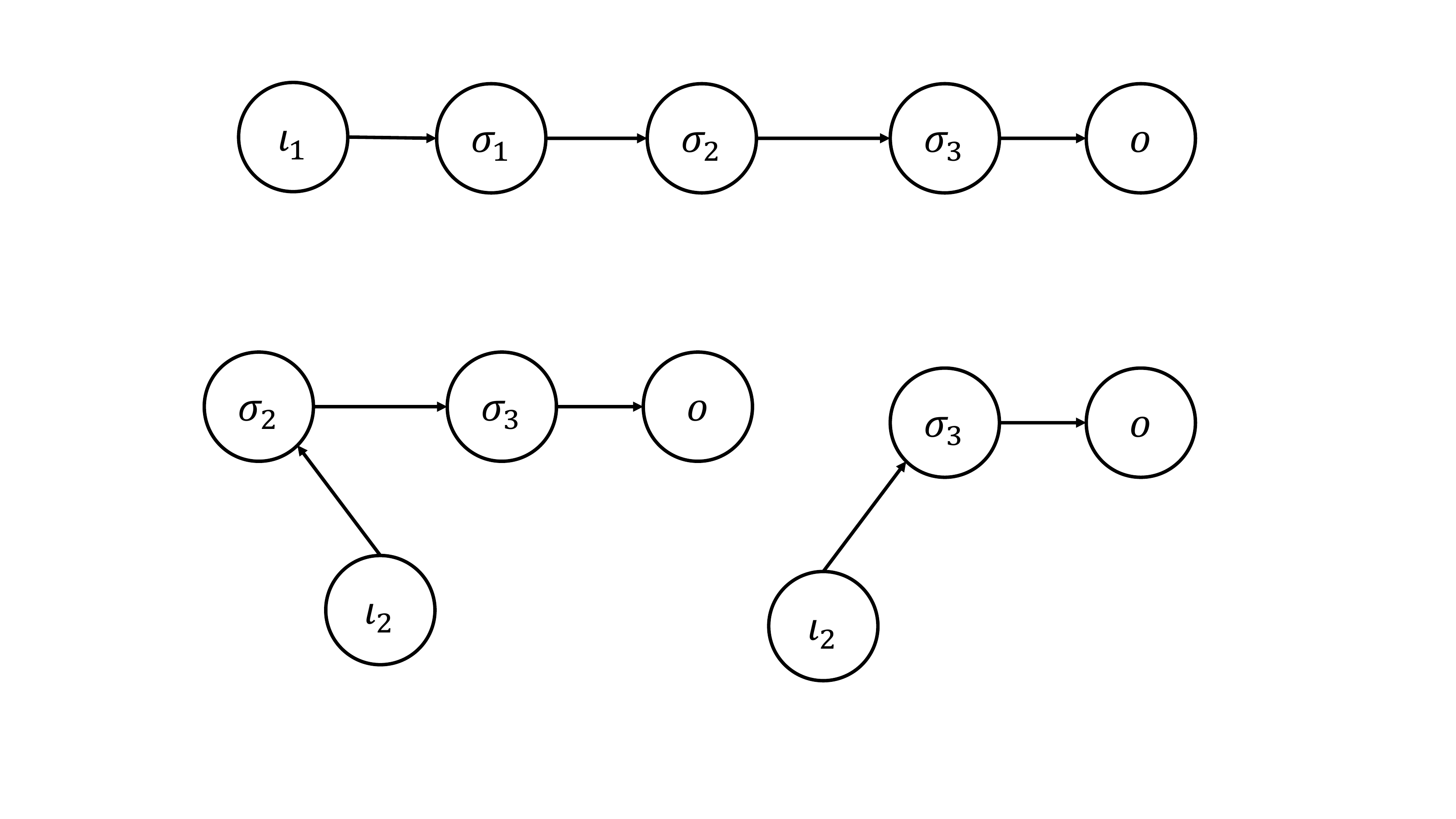} \\
(c) & (d) \\
\includegraphics[width=0.475\linewidth,trim=5.5cm 3cm 5.5cm 3cm,clip=true]%
{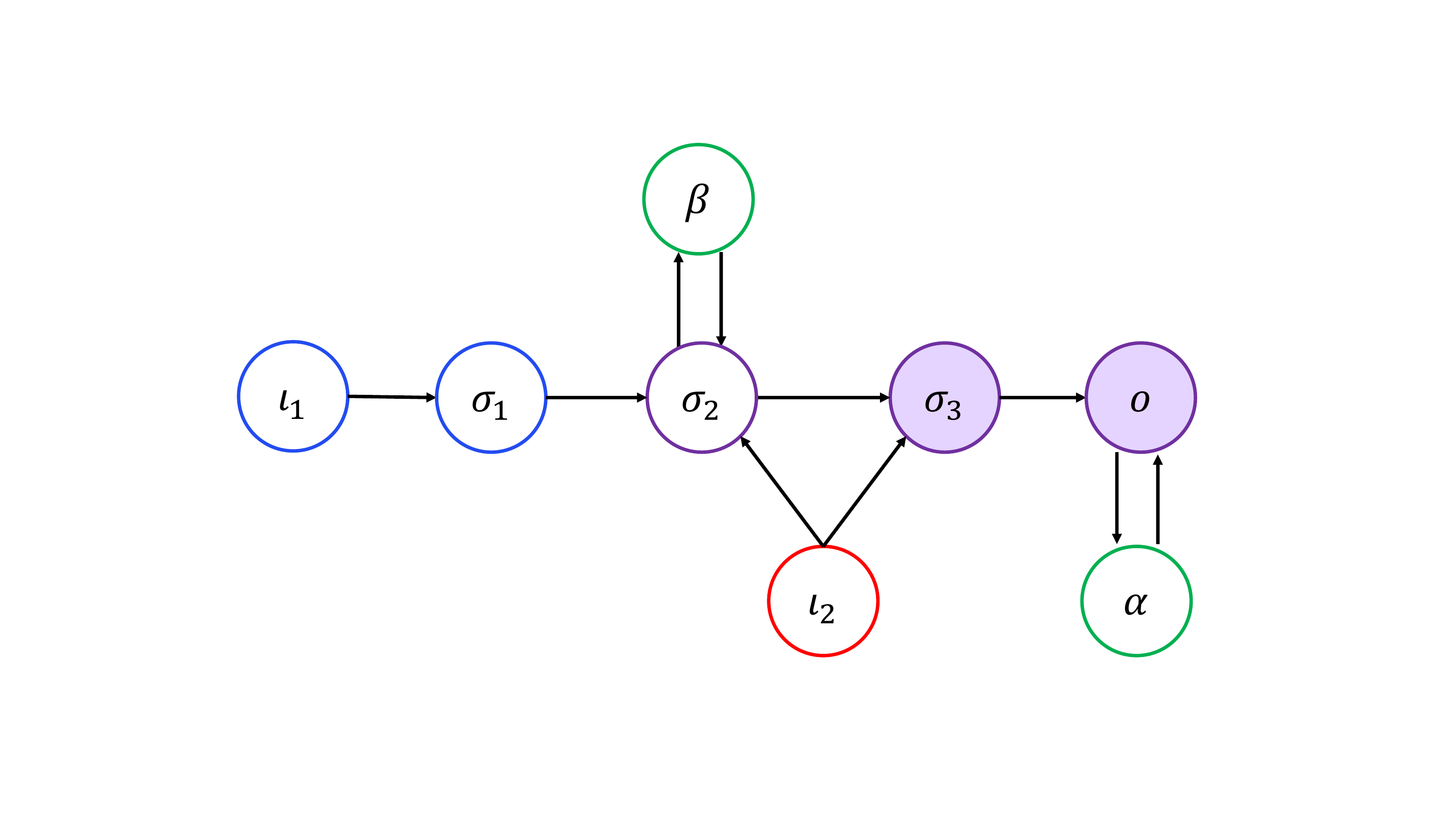} &
\includegraphics[width=0.475\linewidth,trim=4.5cm 2cm 5cm 1.7cm,clip=true]%
{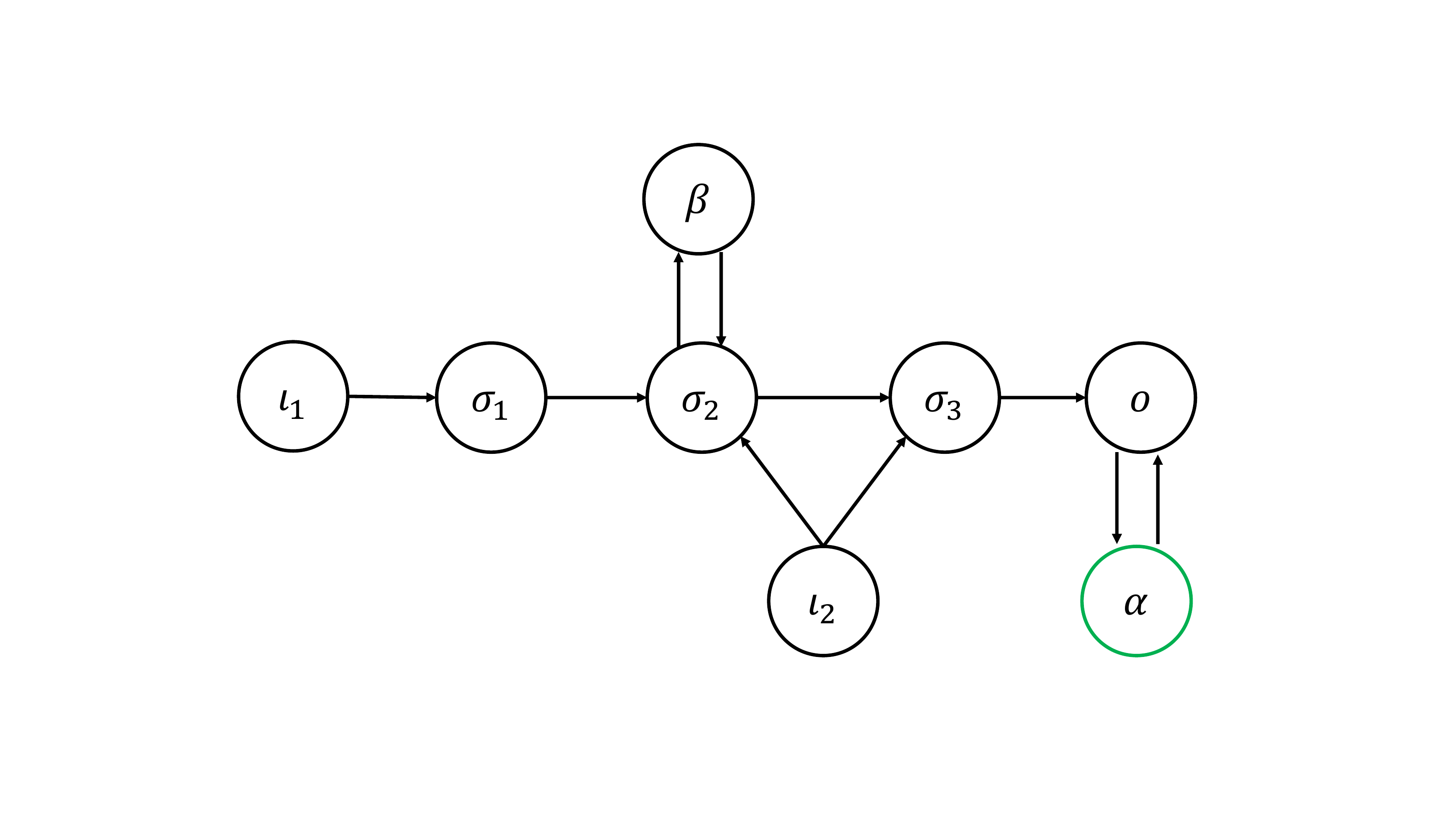} \\
(e) & (f) \\
\includegraphics[width=0.475\linewidth,trim=5.5cm 3cm 5.5cm 3cm,clip=true]%
{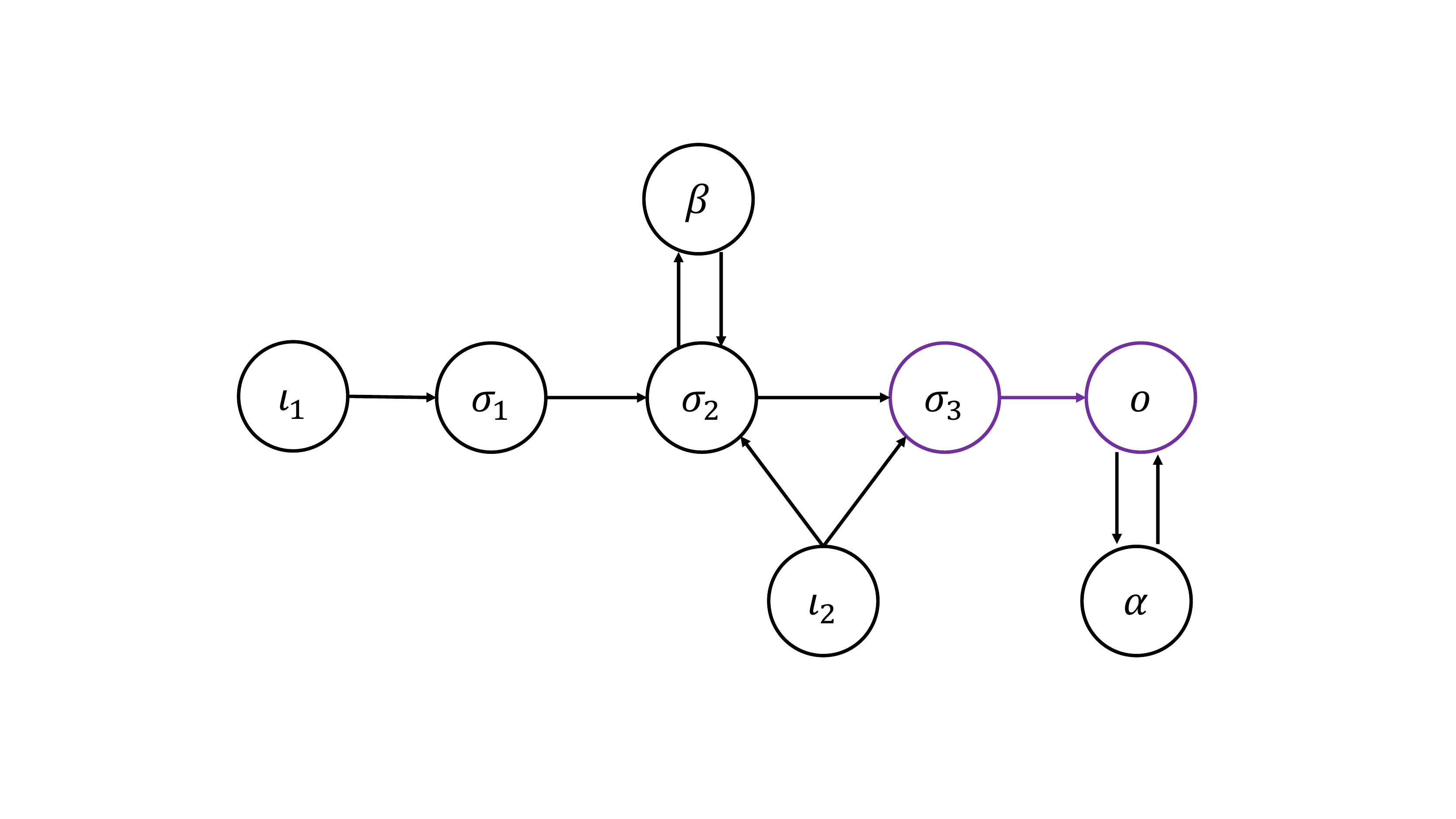} &
\includegraphics[width=0.475\linewidth,trim=4.5cm 2cm 5cm 1.7cm,clip=true]%
{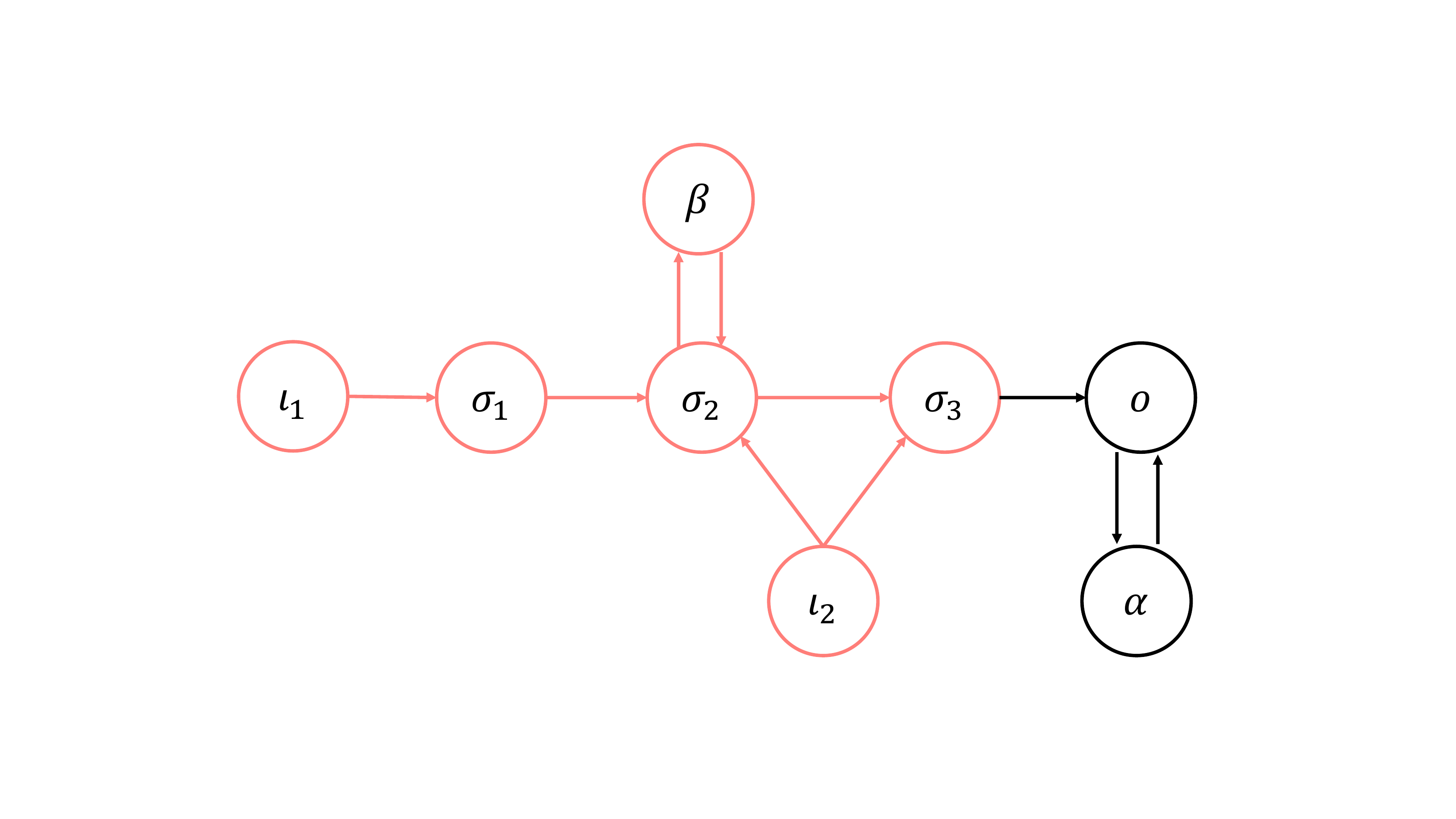}
\end{tabular}
\renewcommand{\figurename}{Figure}
\caption{(a) A core network $\mathcal{G}$ with input nodes $\iota_{1}$ and $\iota_{2}$, and output node $o$ (see definition \ref{defining_core_networks}). (b) Simple paths between the input nodes and the output node. (c) Nodes $\iota_{1}$ and $\sigma_{1}$ (both in blue) are $\iota_{1}$-simple, but they are not downstream $\iota_{2}$; node $\iota_{2}$ is $\iota_{2}$-simple, but it is not downstream $\iota_{2}$. Nodes highlighted in purple are either absolutely simple ($\sigma_{2}$) or absolutely super-simple ($\sigma_{3}$ and $o$). Both $\alpha$ and $\beta$ (in green) are absolutely appendage nodes (definitions \ref{combinatorial_definitions_in_network_multiple_input_nodes} and \ref{useful_definitions_structural_multiple_input_nodes}). The algorithm of subsection \ref{algorithm_core_network_multiple_input_nodes}, implies that $\mathcal{G}$ supports three classes of homeostasis: appendage (d); structural (e); and input counterweight (f).}
\label{figure_definitions}
\end{figure}

Wang \etal \cite{wang20} introduced the concept of path equivalent classes in appendage subnetworks of networks with only one input node. As we need this definition in other contexts, we generalize it to every subnetwork of $\mathcal{G}$.

\begin{definition} \rm \label{definition_path_component}
Let $\mathcal{K}$ be a nonempty subnetwork of $\mathcal{G}$. We say that nodes $\rho_{i}, \rho_{j}$ of $\mathcal{K}$ are \textit{path equivalent in $\mathcal{K}$} if there are paths in $\mathcal{K}$ from $\rho_{i}$ to $\rho_{j}$ and from $\rho_{j}$ to $\rho_{i}$. A $\mathcal{K}$\textit{-path component} is a path equivalence class in $\mathcal{K}$.
\end{definition}

\begin{definition} \rm \label{complement_subnetwork_to_simple_paths} Let $\mathcal{G}$ be a core subnetwork with multiple input nodes $\iota_{1}, \ldots, \iota_{n}$ and output node $o$ and let $\mathcal{G}_{m}$ be the core subnetwork between $\iota_{m}$ and $o$.
\begin{enumerate} [(a)]
\item The $\mathcal{G}_{m}$\textit{-complementary subnetwork} of an $\iota_{m}o$-simple path $S$ is the subnetwork $C_{m}S$ consisting of all nodes of $\mathcal{G}_{m}$ not on $S$ and all arrows in $\mathcal{G}_{m}$ connecting those nodes.
\item The $\mathcal{G}$\textit{-complementary subnetwork} of an $\iota_{m}o$-simple path $S$ is the subnetwork $CS$ consisting of all nodes of $\mathcal{G}$ not on $S$ and all arrows in $\mathcal{G}$ connecting those nodes.
\end{enumerate}
\end{definition}

We start with the `subnetwork motifs' associated with appendage homeostasis.

\begin{definition} \label{useful_definitions_appendage_multiple_input_nodes} \rm
Let $\mathcal{G}$ be a core subnetwork with multiple input nodes $\iota_{1}, \ldots, \iota_{n}$ and output node $o$ and let $\mathcal{G}_{m}$ be the core subnetwork between $\iota_{m}$ and $o$.
\begin{enumerate}[(a)]
\item For every $m = 1, \ldots, n$, we define the $\iota_{m}$-appendage subnetwork $\mathcal{A}_{\mathcal{G}_{m}}$ as the subnetwork of $\mathcal{G}$ composed by all $\iota_{m}$-appendage nodes and all arrows in $\mathcal{G}$ connecting $\iota_{m}$-appendage nodes.
\item The appendage subnetwork $\mathcal{A}_{\mathcal{G}}$ is the subnetwork of $\mathcal{G}$ composed by all absolutely appendage nodes and all arrows in $\mathcal{G}$ connecting absolutely appendage nodes, i.e.,
\begin{equation*}
    \mathcal{A}_{\mathcal{G}} = \mathcal{A}_{\mathcal{G}_{1}} \cap \cdots \cap \mathcal{A}_{\mathcal{G}_{n}}
\end{equation*}
\end{enumerate}
\end{definition}
 
By Definition \ref{definition_path_component}, each path component of a network is a path equivalence class of this network. 
Therefore, we can partition $\mathcal{A}_{\mathcal{G}}$ in different $\mathcal{A}_{\mathcal{G}}$-path components. 
We still need another concept to associate a component of this partition with the appendage homeostasis blocks.

\begin{definition}  \label{no_cycle} \rm
Let $\mathcal{A}_{i}$ be an $\mathcal{A}_{\mathcal{G}}$-path component.
We say that $\mathcal{A}_{i}$ satisfies the \emph{generalized no cycle condition} if the following holds: for every $m = 1, \ldots, n$, for every $\iota_{m}o$-simple path $S_{m}$, nodes in $\mathcal{A}_{i}$ are not $CS_{m}$-path equivalent to any node in $CS_{m} \setminus \mathcal{A}_{i}$. 
\end{definition}

The condition in Definition \ref{no_cycle} is the correct generalization of the `no cycle condition' of Wang \etal \cite{wang20}. Finally, it is shown in Section \ref{mathematical_arguments_and_proofs}, that each appendage homeostasis block corresponds exactly to an $\mathcal{A}_{\mathcal{G}}$-path component $\mathcal{A}_{i}$ satisfying the generalized no cycle condition is an irreducible appendage homeostasis block (see Theorems \ref{characterization_appendage_homeostasis_general_G}  and \ref{reverse_theorem_appendage_block_general_G}). Moreover, this is equivalent to the assertion that each appendage homeostasis block is an appendage homeostasis block of each core subnetwork $\mathcal{G}_{m}$ (see Theorems \ref{characterization_appendage_homeostasis_each_G_m} and \ref{reverse_theorem_appendage_block_each_G_m}).

The topological characterization of appendage homeostasis in networks with multiple input nodes is similar to the topological characterization of appendage homeostasis in single input node networks. This is not the case for the other homeostasis classes. 
Indeed, in single input node  networks there are only two classes of homeostasis, appendage and structural, while in networks with multiple input nodes there is also the input counterweight homeostasis. 
Moreover, single input node networks always support structural homeostasis, which is not always the case with networks with multiple input nodes (see Section \ref{mathematical_arguments_and_proofs}). 

Now we consider the `subnetwork motifs' associated with structural homeostasis.

\begin{definition}
\label{useful_definitions_structural_multiple_input_nodes} \rm
Let $\mathcal{G}$ be a core subnetwork with multiple input nodes $\iota_{1}, \ldots, \iota_{n}$ and output node $o$.
\begin{enumerate}[(a)]
\item An $\iota_{m}$-\textit{super-simple node} is an $\iota_{m}$-simple node that lies on every $\iota_{m}o$-simple path.
\item An \textit{absolutely super-simple node} is an absolutely simple node that lies on every $\iota_{m}o$-simple path, for every $m = 1, \ldots, n$. In particular, an absolutely super simple-node is an $\iota_{m}$-super-simple node, for every $m = 1, \ldots, n$.
\end{enumerate}
\end{definition}

It is straightforward that every core network $\mathcal{G}$ with multiple input nodes $\iota_{1}, \ldots, \iota_{n}$ and output node $o$ has at least one absolutely super-simple node: the output node $o$. However, in contrast to core networks with only one input node where both the input the output nodes are super-simple, core networks with multiple input nodes may exhibit the output node as their only super-simple node (see Figure \ref{example_network_only_one_super_simple_node}). In fact, as we will see later, this is the reason why abstract core networks with multiple input nodes do not necessarily support structural homeostasis.

\begin{figure}[!ht]
\centering
\includegraphics[width=\linewidth,trim=0cm 5cm 0cm 5cm,clip=true]%
{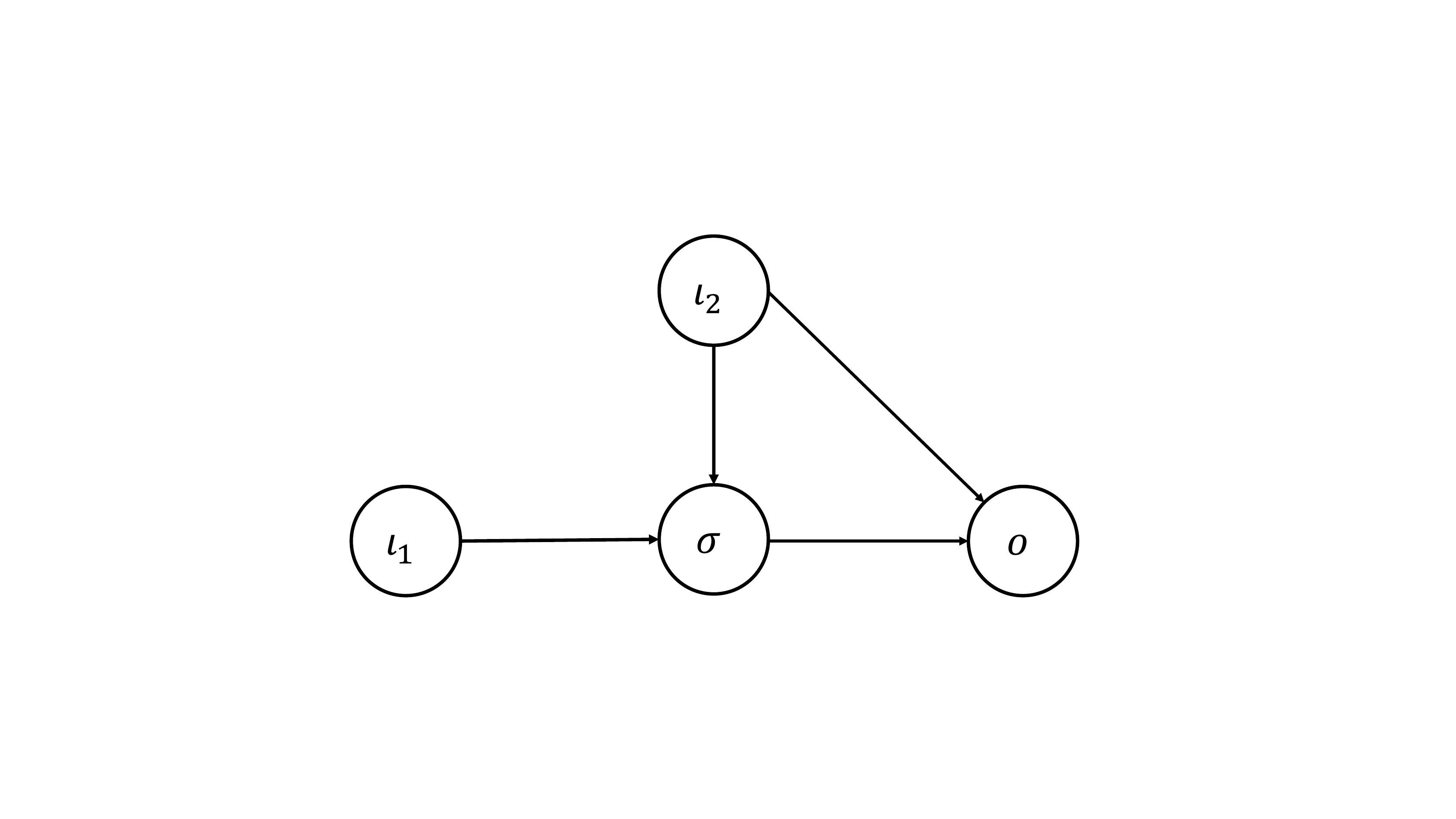}
\renewcommand{\figurename}{Figure}
\caption{A core network with input nodes $\iota_{1}$ and $\iota_{2}$ and output node $o$. The only absolutely super-simple node is $o$. Node $\sigma$ is an absolutely simple node, but it is not between two absolutely super-simple nodes, as it would be expected for core networks with only one input node (see Lemma \ref{relation_i_m_simple_node_and_i_m_super_simple_node}).}
\label{example_network_only_one_super_simple_node}
\end{figure}

Similarly to what happens in networks with single input node there is a natural way of ordering the absolutely super-simple. Indeed, the absolutely super-simple nodes can be uniquely ordered by $\rho_{1} > \rho_{2} > \cdots > \rho_{p} > o$, where $a > b$ when $b$ is downstream $a$ by all $\iota_{m}o$-simple paths (see Lemma \ref{order_absolutely_super_simple_nodes}). Through this ordering, we say that two absolutely super-simple nodes $\rho_{k}, \rho_{k+1}$ are \emph{adjacent} when $\rho_{k+1}$ is the first absolutely super-simple node which appears after $\rho_{k}$.

\begin{definition} \rm
\label{definition_adjacent_absolutely_super_simple_}
Let $\rho_{k}>\rho_{k+1}$ be adjacent absolutely super-simple nodes.
\begin{enumerate}[(a)]
\item An absolutely simple node $\rho$ is \emph{between} $\rho_{k}$ and $\rho_{k+1}$ if there exists an $\iota_{m}o$-simple path that includes $\rho_{k}$ to $\rho$ to $\rho_{k+1}$ in that order, for some $m\in\{1,\ldots,n\}$.
\item The \emph{absolutely super-simple subnetwork}, denoted $\mathcal{L}(\rho_{k}, \rho_{k+1})$, is the subnetwork whose nodes are absolutely simple nodes between $\rho_{k}$ and $\rho_{k+1}$ and whose arrows are arrows of $\mathcal{G}$ connecting nodes in $\mathcal{L}(\rho_{k}, \rho_{k+1})$.
\end{enumerate}
\end{definition}

In addition to the $\mathcal{A}_\mathcal{G}$-path components that satisfy the generalized no cycle condition there may be $\mathcal{A}_{\mathcal{G}}$-path components $\mathcal{B}_{i}$ that do not satisfy this property. More precisely, for every $m = 1, \ldots, n$, there is an $\iota_{m}o$-simple path $S_{m}$ such that nodes in $\mathcal{B}_{i}$ are $CS_{m}$-path equivalent to an absolutely simple node in $CS_{m} \setminus \mathcal{B}_{i}$ which belongs to an absolutely super-simple subnetwork $\mathcal{L}(\rho_{k}, \rho_{k+1})$, where $\rho_{k}, \rho_{k+1}$ are adjacent absolutely super-simple nodes. There is an unique correspondence between each $\mathcal{A}_{\mathcal{G}}$-path component $\mathcal{B}_{i}$ and the absolutely super-simple subnetwork to which $\mathcal{B}_{i}$ is $CS_{m}$-path equivalent, for some $\iota_{m}o$-simple path $S_{m}$ (see Lemma \ref{unique_correspondence_between_B_i_and_absolutely_super_simple_subnetwork}). The union of the $\mathcal{A}_{\mathcal{G}}$-path components $\mathcal{B}_{i}$ and the corresponding absolutely super-simple subnetworks generate the primary subnetworks associated to structural homeostasis.

\begin{definition} \label{absolutely_super_simple_structural_subnetwork} \rm
Let $\rho_{k}$ and $\rho_{k+1}$ be adjacent absolutely super-simple nodes in $\mathcal{G}$. The \textit{absolutely super-simple structural subnetwork} $\mathcal{L}'(\rho_{k}, \rho_{k+1})$ is the input-output subnetwork consisting of nodes in $\mathcal{L}(\rho_{k}, \rho_{k+1}) \cup \mathcal{B}$, where $\mathcal{B}$ consists of all absolutely appendage nodes that are $CS_{m}$-path equivalent to nodes in $\mathcal{L}(\rho_{k}, \rho_{k+1})$ for some $\iota_{m}o$-simple path $S_{m}$, for some $m \in \{1,\ldots,n\}$, i.e., $\mathcal{B}$ consists of all $\mathcal{A}_{\mathcal{G}}$-path components $\mathcal{B}_{i}$ that are $CS_{m}$-path equivalent to nodes in $\mathcal{L}(\rho_{k}, \rho_{k+1})$ for some $S_{m}$, for some $m \in \{1,\ldots,n\}$. Arrows of $\mathcal{L}'(\rho_{k}, \rho_{k+1})$ are arrows of $\mathcal{G}$ that connect nodes in $\mathcal{L}'(\rho_{k}, \rho_{k+1})$. Note that $\rho_{k}$ is the input node and that $\rho_{k+1}$ is the output node of $\mathcal{L}'(\rho_{k}, \rho_{k+1})$.
\end{definition}

Each absolutely super-simple structural subnetwork $\mathcal{L}'(\rho_{k}, \rho_{k+1})$ is a single node input-output network with $\rho_{k}$ as the input node and $\rho_{k+1}$ as the output node. Therefore, the homeostasis matrix $H(\mathcal{L}'(\rho_{k}, \rho_{k+1}))$ is well defined. Indeed, we show that the homeostasis matrix of each absolutely super-simple structural subnetwork corresponds to an irreducible structural homeostasis block and, conversely, each irreducible structural homeostasis block is given by the homeostasis matrix of an absolutely super-simple structural subnetwork (see Theorems \ref{structural_homeostasis_G} and \ref{structural_homeostasis_G_reverse}).

Finally, we define the `subnetwork motif' associated with input counterweight homeostasis.

\begin{definition} \rm \label{definition_input_counterweight_subnetwork}
Let the absolutely super-simple nodes of $\mathcal{G}$ be $\rho_{1} > \cdots > \rho_{s} > o$. The \emph{input counterweight subnetwork} $\mathcal{W}_{\mathcal{G}}$ of $\mathcal{G}$ is the subnetwork composed by: (1) the input nodes $\iota_{1}, \ldots, \iota_{n}$, (2) the absolutely super-simple node $\rho_{1}$, (3) nodes $\tau$ for which there exists an $m \in \{1,\ldots,n\}$ such that there is an $\iota_{m}o$-simple path that passes at $\iota_{m}$, $\tau$ and $\rho_{1}$ in that order, (4) the nodes that are not absolutely appendage nor absolutely simple, and (5) nodes in $\mathcal{C}$, where $\mathcal{C}$ consists of all absolutely appendage nodes that are $CS_{m}$-path equivalent to nodes that are not absolutely appendage and that are not between two absolutely super-simple nodes, for some $\iota_{m}o$-simple path $S_{m}$ ($m \in \{1,\ldots,n\}$). Arrows of $\mathcal{W}_{\mathcal{G}}$ are the arrows of $\mathcal{G}$ that connect nodes of $\mathcal{W}_{\mathcal{G}}$.
\end{definition}

\subsection{Enumerating Homeostasis Subnetworks} 
\label{algorithm_core_network_multiple_input_nodes}

The classification of homeostasis types obtained in this paper allows us to we write down an algorithm for enumerating subnetworks corresponding to the $r = p + q + 1$ homeostasis blocks {\color{red}(see figure \ref{figure_definitions})}.

\noindent
\textbf{Step 1:} Determine the $\mathcal{A}_\mathcal{G}$-path components $\mathcal{A}_{1}, \ldots, \mathcal{A}_{p}$ satisfying the ge\-ne\-ra\-li\-zed no cycle condition. By Theorems \ref{characterization_appendage_homeostasis_general_G} and \ref{reverse_theorem_appendage_block_general_G}, these are the appendage homeostasis subnetworks of $\mathcal{G}$, and their corresponding Jacobian matrix $J_{\mathcal{A}_{i}}$ is an irreducible appendage homeostasis block that appears in the normal form of $\langle H \rangle$. Moreover, there are $p$ independent defining conditions for appendage homeostasis based on the determinants $\det (J_{\mathcal{A}_{i}}) = 0$, for $i = 1, \ldots, p$.

\medskip

\noindent
\textbf{Step 2:} Determine the absolutely super-simple nodes of  $\mathcal{G}$. If the only absolutely super-simple node of $\mathcal{G}$ is $o$, then $\mathcal{G}$ does not support structural homeostasis. On the other hand, if there is more than one absolutely super-simple node, consider their natural order $\rho_{1} > \cdots > \rho_{q} > \rho_{q+1} = o$ and determine the corresponding absolutely super-simple structural subnetwork $\mathcal{L}'(\rho_{k}, \rho_{k+1})$. By Theorems \ref{structural_homeostasis_G} and \ref{structural_homeostasis_G_reverse}, the corresponding homeostasis matrix $H(\mathcal{L}'(\rho_{k}, \rho_{k+1}))$ is an irreducible structural homeostasis block that appears in the normal form of $\langle H \rangle$. Moreover, there are $q$ independent defining conditions for structural homeostasis based on the determinants $\det\!\big(H(\mathcal{L}'(\rho_{k}, \rho_{k+1}))\big) = 0$, for $k = 1, \ldots, q$.

\medskip

\noindent
\textbf{Step 3:} Determine the input counterweight subnetwork $\mathcal{W}_{\mathcal{G}}$ of $\mathcal{G}$. Then, the generalized homeostasis matrix of $\langle H \rangle(\mathcal{W}_{\mathcal{G}})$ is, up to permutation of rows or columns, the input counterweight homeostasis block $C$ that appears in the normal form of $\langle H \rangle$. Furthermore, there is one defining condition for input counterweight homeostasis based on the determinant $\det\!\big( \langle H \rangle (\mathcal{W}_{\mathcal{G}})\big) = 0$.

\section{Structure of Infinitesimal Homeostasis} \label{mathematical_arguments_and_proofs}

In this section we provide the proofs of all results behind the classification of homeostasis types and the algorithm in subsection \ref{algorithm_core_network_multiple_input_nodes}. We follow Wang \etal \cite{wang20} and provide the appropriate generalizations of each corresponding result. Nevertheless, it is important to remark that there are new difficulties that arise in the multiple input node context that do not have a single input node counterpart. Completely new arguments were required to overcome these difficulties.

\subsection{Core Networks}

We extended the results of Wang \etal \cite{wang20} to a core network associated to a multiple input nodes input-output network.
Recall Definition~\ref{defining_core_networks} of a core network.

The linearly stable equilibrium $(X_{0}, \mathcal{I}_{0})$ of \eqref{admissible_systems_ODE_multiple_input_nodes} satisfies the system of equations that can be explicitly written as
\begin{equation} \label{admissible_systems_ODE-equilibrium}
\begin{aligned}
    & f_{\iota_{1}}(x_{\iota_{1}}, x_{\iota_{2}}, \ldots, x_{\iota_{n}}, x_{\rho}, x_{o}, \mathcal{I}) = 0 \\
    & f_{\iota_{2}}(x_{\iota_{1}}, x_{\iota_{2}}, \ldots, x_{\iota_{n}}, x_{\rho}, x_{o}, \mathcal{I}) = 0 \\
    & \vdots \\
    & f_{\iota_{n}}(x_{\iota_{1}}, x_{\iota_{2}}, \ldots, x_{\iota_{n}}, x_{\rho}, x_{o}, \mathcal{I}) = 0 \\
    & f_{\rho}(x_{\iota_{1}}, x_{\iota_{2}}, \ldots, x_{\iota_{n}}, x_{\rho}, x_{o}) = 0 \\
    & f_{o}(x_{\iota_{1}}, x_{\iota_{2}}, \ldots, x_{\iota_{n}}, x_{\rho}, x_{o}) = 0
\end{aligned}
\end{equation}

Following Wang \etal\cite{wang20}, we partition the regulatory nodes $\rho$ in three classes depending if they are upstream from the output node or/and downstream from at least one input node. 

More precisely, consider the partition of the nodes in $\mathcal{G}$ in three classes as follows:  
\begin{enumerate}[(1)]
\item those nodes $\sigma$ that are both upstream from $o$ and downstream from at least one input node $\iota_{m}$,
\item those nodes $d$ that are not downstream from any input node $\iota_{m}$,
\item those nodes $u$ which are downstream from at least one input node $\iota_{m}$, but not upstream from $o$.
\end{enumerate}
Figure \ref{networks_multiple_input_nodes} shows the connections which can be found in $\mathcal{G}$.

We can now rewrite equations in \eqref{admissible_systems_ODE-equilibrium} as
\begin{equation} \label{admissible_systems_ODE-equilibrium_rewritten}
\begin{aligned}
    & f_{\iota_{1}}(x_{\iota_{1}}, x_{\iota_{2}}, \ldots, x_{\iota_{n}}, x_{\sigma}, x_{u}, x_{d}, x_{o}, \mathcal{I}) = 0 \\
    & f_{\iota_{2}}(x_{\iota_{1}}, x_{\iota_{2}}, \ldots, x_{\iota_{n}}, x_{\sigma}, x_{u}, x_{d}, x_{o}, \mathcal{I}) = 0 \\
    & \vdots \\
    & f_{\iota_{n}}(x_{\iota_{1}}, x_{\iota_{2}}, \ldots, x_{\iota_{n}}, x_{\sigma}, x_{u}, x_{d}, x_{o}, \mathcal{I}) = 0 \\
    & f_{\sigma}(x_{\iota_{1}}, x_{\iota_{2}}, \ldots, x_{\iota_{n}}, x_{\sigma}, x_{u}, x_{d}, x_{o}) = 0 \\
    & f_{u}(x_{\iota_{1}}, x_{\iota_{2}}, \ldots, x_{\iota_{n}}, x_{\sigma}, x_{u}, x_{d}, x_{o}) = 0 \\
    & f_{d}(x_{\iota_{1}}, x_{\iota_{2}}, \ldots, x_{\iota_{n}}, x_{\sigma}, x_{u}, x_{d}, x_{o}) = 0 \\
    & f_{o}(x_{\iota_{1}}, x_{\iota_{2}}, \ldots, x_{\iota_{n}}, x_{\sigma}, x_{u}, x_{d}, x_{o}) = 0
\end{aligned}
\end{equation}
As already proved by Wang \etal \cite[Lem 2.1]{wang20}, \eqref{admissible_systems_ODE-equilibrium_rewritten} may be simplified to
\begin{equation} \label{admissible_systems_ODE-equilibrium_simplified}
\begin{aligned}
    & f_{\iota_{1}}(x_{\iota_{1}}, x_{\iota_{2}}, \ldots, x_{\iota_{n}}, x_{\sigma}, x_{d}, x_{o}, \mathcal{I}) = 0 \\
    & f_{\iota_{2}}(x_{\iota_{1}}, x_{\iota_{2}}, \ldots, x_{\iota_{n}}, x_{\sigma}, x_{d}, x_{o}, \mathcal{I}) = 0 \\
    & \vdots \\
    & f_{\iota_{n}}(x_{\iota_{1}}, x_{\iota_{2}}, \ldots, x_{\iota_{n}}, x_{\sigma}, x_{d}, x_{o}, \mathcal{I}) = 0 \\
    & f_{\sigma}(x_{\iota_{1}}, x_{\iota_{2}}, \ldots, x_{\iota_{n}}, x_{\sigma}, x_{d}, x_{o}) = 0 \\
    & f_{u}(x_{\iota_{1}}, x_{\iota_{2}}, \ldots, x_{\iota_{n}}, x_{\sigma}, x_{u}, x_{d}, x_{o}) = 0 \\
    & f_{d}(x_{d}) = 0 \\
    & f_{o}(x_{\iota_{1}}, x_{\iota_{2}}, \ldots, x_{\iota_{n}}, x_{\sigma}, x_{d}, x_{o}) = 0
\end{aligned}
\end{equation}
Note that if we fix $x_{d}$ at some value, it is trivial to obtain an admissible system to $\mathcal{G}_{c}$ from \eqref{admissible_systems_ODE-equilibrium_simplified}.

\begin{figure}[!ht]
\centering
\includegraphics[width=\linewidth,trim=0cm 1.5cm 0cm 1.5cm,clip=true]{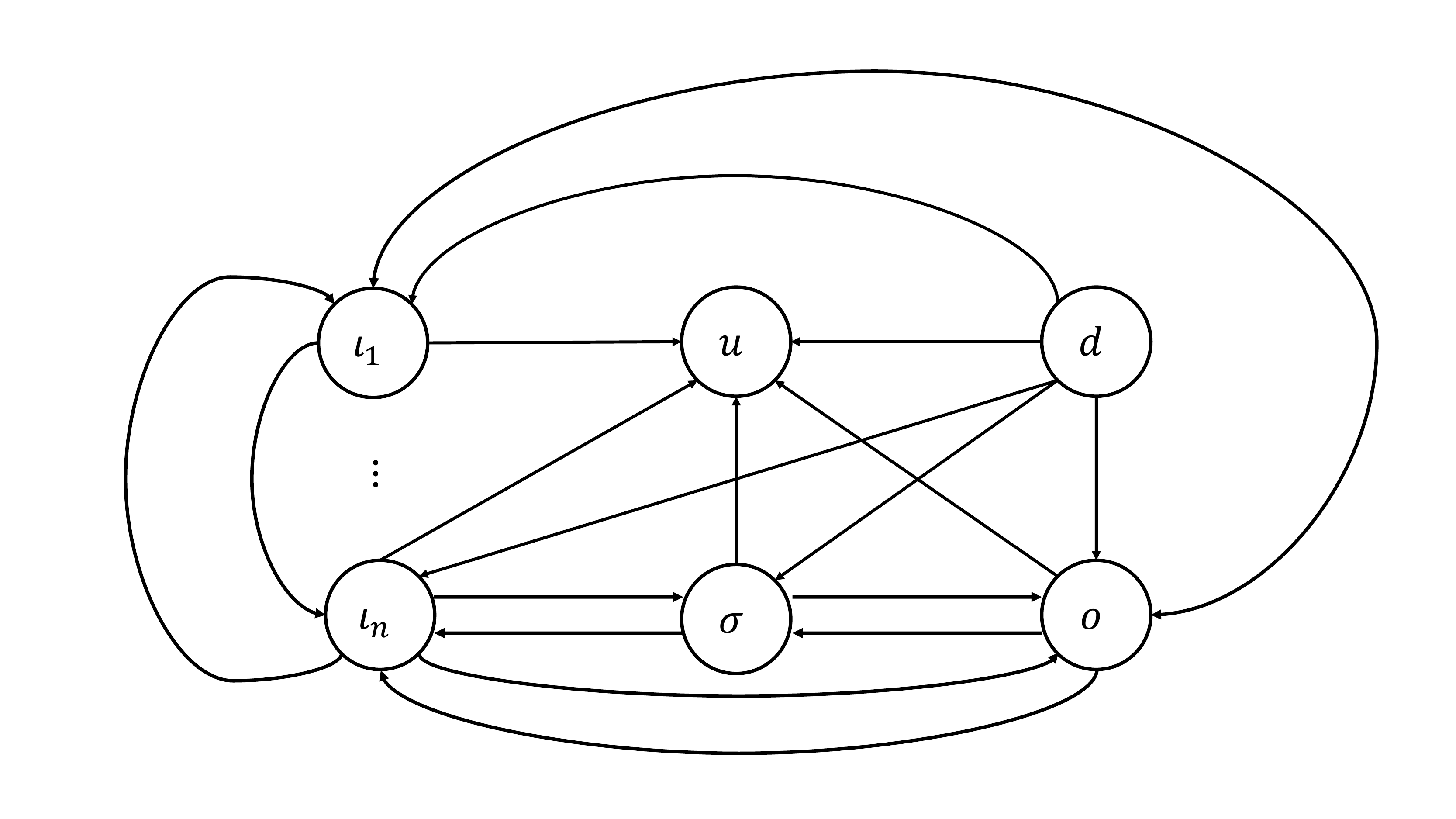}
\renewcommand{\figurename}{Figure}
\caption{The possible connections in $\mathcal{G}$. The corresponding core network $\mathcal{G}_{c}$ is composed by the input nodes $\iota_{1}, \ldots, \iota_{n}$, the nodes $\sigma$ that are upstream from the output node and downstream of at least one input node, and the output node $o$.}
\label{networks_multiple_input_nodes}
\end{figure}

\begin{lemma} \label{if_original_equilibrium_stable_so_is_core}
Suppose $X_{0} = (x_{\iota_{1}}^{*}, x_{\iota_{2}}^{*}, \ldots, x_{\iota_{n}}^{*}, x_{\sigma}^{*}, x_{u}^{*}, x_{d}^{*}, x_{o}^{*})$ is a linear stable equilibrium of \eqref{admissible_systems_ODE-equilibrium_simplified}. Then the core admissible system (obtained by freezing $x_{d}$ at $x_{d}^{*}$)
\begin{equation} \label{core_admissible_system}
\begin{aligned}
    & \dot{x}_{\iota_{1}} = f_{\iota_{1}}(x_{\iota_{1}}, x_{\iota_{2}}, \ldots, x_{\iota_{n}}, x_{\sigma}, x_{d}^{*}, x_{o}, \mathcal{I})\\
    & \dot{x}_{\iota_{2}} = f_{\iota_{2}}(x_{\iota_{1}}, x_{\iota_{2}}, \ldots, x_{\iota_{n}}, x_{\sigma}, x_{d}^{*}, x_{o}, \mathcal{I})\\
    & \vdots \\
    & \dot{x}_{\iota_{n}} = f_{\iota_{n}}(x_{\iota_{1}}, x_{\iota_{2}}, \ldots, x_{\iota_{n}}, x_{\sigma}, x_{d}^{*}, x_{o}, \mathcal{I})\\
    & \dot{x}_{\sigma} = f_{\sigma}(x_{\iota_{1}}, x_{\iota_{2}}, \ldots, x_{\iota_{n}}, x_{\sigma}, x_{d}^{*}, x_{o})\\
    & \dot{x}_{o} = f_{o}(x_{\iota_{1}}, x_{\iota_{2}}, \ldots, x_{\iota_{n}}, x_{\sigma}, x_{d}^{*}, x_{o})
\end{aligned}
\end{equation}
has a linear stable equilibrium in $Y_{0} = (x_{\iota_{1}}^{*}, x_{\iota_{2}}^{*}, \ldots, x_{\iota_{n}}^{*}, x_{\sigma}^{*}, x_{o}^{*})$.
\end{lemma}

\begin{proof}
It is trivial that $Y_{0}$ is an equilibrium of \eqref{core_admissible_system}. We start by verifying that $Y_{0}$ is linearly stable. Indeed, the Jacobian matrix $J$ of \eqref{admissible_systems_ODE-equilibrium_simplified} evaluated at $X_{0}$ is

\begin{equation} \label{jacobian_original_system}
    J = \begin{pmatrix} f_{\iota_{1}, x_{\iota_{1}}} & \cdots & f_{\iota_{1}, x_{\iota_{n}}} & f_{\iota_{1}, x_{\sigma}} & f_{\iota_{1}, x_{d}} & 0 & f_{\iota_{1}, x_{o}} \\
    \vdots & \ddots & \vdots & \vdots & \vdots & \vdots & \vdots \\
    f_{\iota_{n}, x_{\iota_{1}}} & \cdots & f_{\iota_{n}, x_{\iota_{n}}} & f_{\iota_{n}, x_{\sigma}} & f_{\iota_{n}, x_{d}} & 0 & f_{\iota_{n}, x_{o}} \\
    f_{\sigma, x_{\iota_{1}}} & \cdots & f_{\sigma, x_{\iota_{n}}} & f_{\sigma, x_{\sigma}} & f_{\sigma, x_{d}} & 0 & f_{\sigma, x_{o}} \\
    0 & \cdots & 0 & 0 & f_{d, x_{d}} & 0 & 0 \\
    f_{u, x_{\iota_{1}}} & \cdots & f_{u, x_{\iota_{n}}} & f_{u, x_{\sigma}} & f_{u, x_{d}} & f_{u, x_{u}} & f_{u, x_{o}} \\
    f_{o, x_{\iota_{1}}} & \cdots & f_{o, x_{\iota_{n}}} & f_{o, x_{\sigma}} & f_{o, x_{d}} & 0 & f_{o, x_{o}}
    \end{pmatrix}
\end{equation}
and therefore the eigenvalues of $J$ are the same eigenvalues of $f_{d, x_{d}}$, $f_{u, x_{u}}$ and of the matrix $J_{c}$, where
\begin{equation} \label{jacobian_core_system}
    J_{c} = \begin{pmatrix} f_{\iota_{1}, x_{\iota_{1}}} & \cdots & f_{\iota_{1}, x_{\iota_{n}}} & f_{\iota_{1}, x_{\sigma}} & f_{\iota_{1}, x_{o}} \\
    \vdots & \ddots & \vdots & \vdots & \vdots \\
    f_{\iota_{n}, x_{\iota_{1}}} & \cdots & f_{\iota_{n}, x_{\iota_{n}}} & f_{\iota_{n}, x_{\sigma}} & f_{\iota_{n}, x_{o}} \\
    f_{\sigma, x_{\iota_{1}}} & \cdots & f_{\sigma, x_{\iota_{n}}} & f_{\sigma, x_{\sigma}} & f_{\sigma, x_{o}} \\
    f_{o, x_{\iota_{1}}} & \cdots & f_{o, x_{\iota_{n}}} & f_{o, x_{\sigma}} & f_{o, x_{o}}
    \end{pmatrix}
\end{equation}
Notice now that $J_{c}$ is the Jacobian matrix of \eqref{core_admissible_system} calculated at $Y_{0}$, and therefore, if $X_{0}$ is a linearly stable equilibrium, then so it is $Y_{0}$.
\end{proof}

\begin{theorem} \label{theorem_homeostasis_points_are_the_same_core_and_original_networks}
Let $x_{o}(\mathcal{I})$ be the input-output function of the admissible system \eqref{admissible_systems_ODE_multiple_input_nodes} and let $x^{c}_{o}(\mathcal{I})$ be the input-output function of the associated core admissible system \eqref{core_admissible_system}. Then $x^{c}_{o}$ has a point of infinitesimal homeostasis at $\mathcal{I}_{0}$ if and only if $x_{o}$ has a point of infinitesimal homeostasis at $\mathcal{I}_{0}$.
\end{theorem}

\begin{proof}
By Lemma \ref{cramer_rule} the input-output function of the  admissible system \eqref{admissible_systems_ODE_multiple_input_nodes} is given by
\begin{equation} \label{homeostasis_equation_part_ii}
    \frac{d x_o}{d \mathcal{I}} = \frac{\det\!\big(\langle H \rangle\big)}{\det(J)}
\end{equation}
where
\begin{equation} \label{definition_homeostasis_matrix_original_network}
    \langle H \rangle = \begin{pmatrix} f_{\iota_{1}, x_{\iota_{1}}} & \cdots & f_{\iota_{1}, x_{\iota_{n}}} & f_{\iota_{1}, x_{\sigma}} & f_{\iota_{1}, x_{d}} & 0 & - f_{\iota_{1}, \mathcal{I}} \\
    \vdots & \ddots & \vdots & \vdots & \vdots & \vdots & \vdots \\
    f_{\iota_{n}, x_{\iota_{1}}} & \cdots & f_{\iota_{n}, x_{\iota_{n}}} & f_{\iota_{n}, x_{\sigma}} & f_{\iota_{n}, x_{d}} & 0 & - f_{\iota_{n}, \mathcal{I}} \\
    f_{\sigma, x_{\iota_{1}}} & \cdots & f_{\sigma, x_{\iota_{n}}} & f_{\sigma, x_{\sigma}} & f_{\sigma, x_{d}} & 0 & 0 \\
    0 & \cdots & 0 & 0 & f_{d, x_{d}} & 0 & 0 \\
    f_{u, x_{\iota_{1}}} & \cdots & f_{u, x_{\iota_{n}}} & f_{u, x_{\sigma}} & f_{u, x_{d}} & f_{u, x_{u}} & 0 \\
    f_{o, x_{\iota_{1}}} & \cdots & f_{o, x_{\iota_{n}}} & f_{o, x_{\sigma}} & f_{o, x_{d}} & 0 & 0
    \end{pmatrix}
\end{equation}
Likewise, by Lemma \ref{cramer_rule} the input-output function of the core admissible system \eqref{core_admissible_system} is given by
\begin{equation} \label{core_homeostasis_equation_part_ii}
    \frac{d x_{o}^{c}}{d \mathcal{I}} = \frac{\det\!\big(\langle H_{c} \rangle\big)}{\det(J_{c})}
\end{equation}
where
\begin{equation} \label{definition_homeostasis_matrix_core_network}
    \langle H_{c} \rangle = \begin{pmatrix} f_{\iota_{1}, x_{\iota_{1}}} & \cdots & f_{\iota_{1}, x_{\iota_{n}}} & f_{\iota_{1}, x_{\sigma}} & - f_{\iota_{1}, \mathcal{I}} \\
    \vdots & \ddots & \vdots & \vdots & \vdots \\
    f_{\iota_{n}, x_{\iota_{1}}} & \cdots & f_{\iota_{n}, x_{\iota_{n}}} & f_{\iota_{n}, x_{\sigma}} & - f_{\iota_{n}, \mathcal{I}} \\
    f_{\sigma, x_{\iota_{1}}} & \cdots & f_{\sigma, x_{\iota_{n}}} & f_{\sigma, x_{\sigma}} & 0 \\
    f_{o, x_{\iota_{1}}} & \cdots & f_{o, x_{\iota_{n}}} & f_{o, x_{\sigma}} & 0
    \end{pmatrix}
\end{equation}
From Lemma \ref{if_original_equilibrium_stable_so_is_core}, we have
\begin{equation} \label{determinant_jacobian_original_and_core_network}
    \det(J) = \det(f_{d, x_{d}}) \cdot \det(f_{u, x_{u}}) \cdot \det(J_{c})
\end{equation}
From \eqref{definition_homeostasis_matrix_original_network} and
\eqref{definition_homeostasis_matrix_core_network} we get
\begin{equation} \label{relation_homeostasis_matrix_original_core}
    \det\!\big(\langle H \rangle\big) = \det(f_{d, x_{d}}) \cdot \det(f_{u, x_{u}}) \cdot \det\!\big(\langle H_{c} \rangle\big)
\end{equation}
Hence
\begin{equation}
    \frac{d x_o}{d \mathcal{I}} = \frac{\det\!\big(\langle H \rangle\big)}{\det(J)} = \frac{\det(f_{d, x_{d}}) \cdot \det(f_{u, x_{u}}) \cdot \det\!\big(\langle H_{c} \rangle\big)}{ \det(f_{d, x_{d}}) \cdot \det (f_{u, x_{u}}) \cdot \det(J_{c})} = \frac{\det\!\big(\langle H_{c} \rangle\big)}{\det(J_{c})} = \frac{d x_{o}^{c}}{d \mathcal{I}}
\end{equation}
and so the theorem is proved.
\end{proof}

Theorem \ref{theorem_homeostasis_points_are_the_same_core_and_original_networks} allows us to analyze the core subnetwork in the search for infinitesimal homeostasis points. Therefore, given a network $\mathcal{G}$ with output node $o$ and input nodes $\iota_{1}, \iota_{2}, \ldots, \iota_{n}$, all of them associated to the same input parameter $\mathcal{I}$, we can detach from $\mathcal{G}$ nodes $d$ that are not downstream from any input node and nodes $u$ which are not upstream from $o$ to obtain the core subnetwork $\mathcal{G}_{c}$. We can also analyze $\mathcal{G}_{c}$ from a different point of view.

Given a network $\mathcal{G}$ described above, denote by $V$ the set of nodes of $\mathcal{G}$ and by $E$ the set of arrows of $\mathcal{G}$. Consider the core subnetwork $\mathcal{G}_{c}$ and $V_{c}$ and $E_{c}$ the sets of nodes and of arrows of $\mathcal{G}_{c}$, respectively. For every $m =1,\ldots,n$, consider the subnetwork $\mathcal{G}_{m}$ generated by $\iota_{m}, o$ and all nodes downstream from $\iota_{m}$ and upstream from $o$.

\begin{proposition} \label{core_subnetwork_as_union_individual_ core_networks}
Given an input-output network $\mathcal{G}$ with multiple input nodes \linebreak $\iota_{1}, \iota_{2}, \ldots, \iota_{n}$, all of them associated to the same input parameter $\mathcal{I}$ and output node $o$, then its core subnetwork $\mathcal{G}_{c}$ is the union of all core networks $\mathcal{G}_{m}$ from $\iota_{m}$ to $o$
\begin{equation}
     \mathcal{G}_{c} = \mathcal{G}_{1} \cup \cdots \cup \mathcal{G}_{n}
\end{equation}
i.e., considering networks as directed graphs, we define $V_{m}$ and $E_{m}$ as the sets of nodes and arrows of $\mathcal{G}_{m}$, respectively, for every $m=1,\ldots,n$, and $V_{c}$ and $E_{c}$ as the sets of nodes and arrows of $\mathcal{G}_{c}$, respectively. Then:
\begin{equation}
    V_{c} = V_{1} \cup \cdots \cup V_{n} \textrm{ and } E_{c} = E_{1} \cup \cdots \cup E_{n}
\end{equation}
\end{proposition}

\begin{proof}
The proof is straightforward as every node in $\mathcal{G}_{c}$ is upstream $o$ and downstream an input node $\iota_{m}$.
\end{proof}

\begin{figure}[!ht]
\centering
\includegraphics[width=\linewidth]{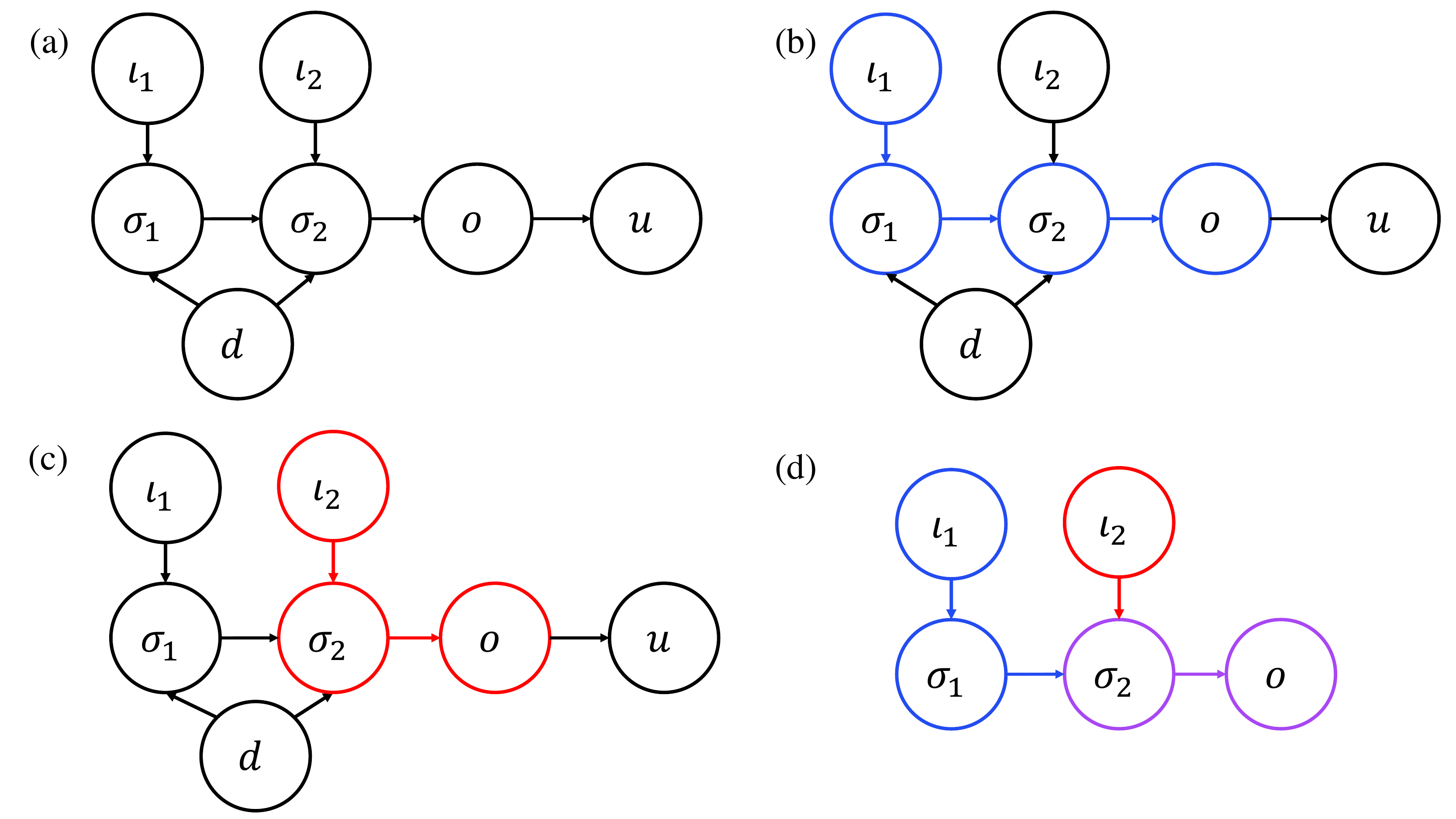}
\renewcommand{\figurename}{Figure}
\caption{(a) An abstract network $\mathcal{G}$ with output node $o$ and input nodes $\iota_{1}$ and $\iota_{2}$. (b) Nodes and arrows that belong to the core subnetwork $\mathcal{G}_{1}$ between $\iota_{1}$ and $o$ are highlighted in blue. (c) Nodes and arrows that belong to the core subnetwork $\mathcal{G}_{2}$ between $\iota_{2}$ and $o$ are highlighted in red. (d) The core network $\mathcal{G}_{c}$ is obtained by the union between $\mathcal{G}_{1}$ and $\mathcal{G}_{2}$. Nodes and arrows of $\mathcal{G}_{c}$ that belong to both $\mathcal{G}_{1}$ and $\mathcal{G}_{2}$ are highlighted in purple. Nodes and arrows which belong to $\mathcal{G}_{1}$ but not to $\mathcal{G}_{2}$ are highlighted in blue and nodes and arrows which belong to $\mathcal{G}_{2}$ but not to $\mathcal{G}_{1}$ are highlighted in red.}
\label{example_networks_multiple_input_nodes}
\end{figure}

For example, consider the abstract network $\mathcal{G}$ in Figure \ref{example_networks_multiple_input_nodes}(a). In this example, the subnetwork $\mathcal{G}_{1}$ composed by nodes downstream $\iota_{1}$ and upstream $o$ is
\begin{equation}
    \iota_{1} \rightarrow \sigma_{1} \rightarrow \sigma_{2} \rightarrow o
\end{equation}
On the other hand, the subnetwork $\mathcal{G}_{2}$ is
\begin{equation}
    \iota_{2} \rightarrow \sigma_{2} \rightarrow o
\end{equation}
By Proposition \ref{core_subnetwork_as_union_individual_ core_networks}, we conclude that the corresponding core network is the one shown in Figure \ref{example_networks_multiple_input_nodes}(d).

\ignore{
\begin{figure}[!ht]
\centering
\includegraphics[scale = 0.35,trim=0cm 2.5cm 0cm 2.5cm,clip=true]%
{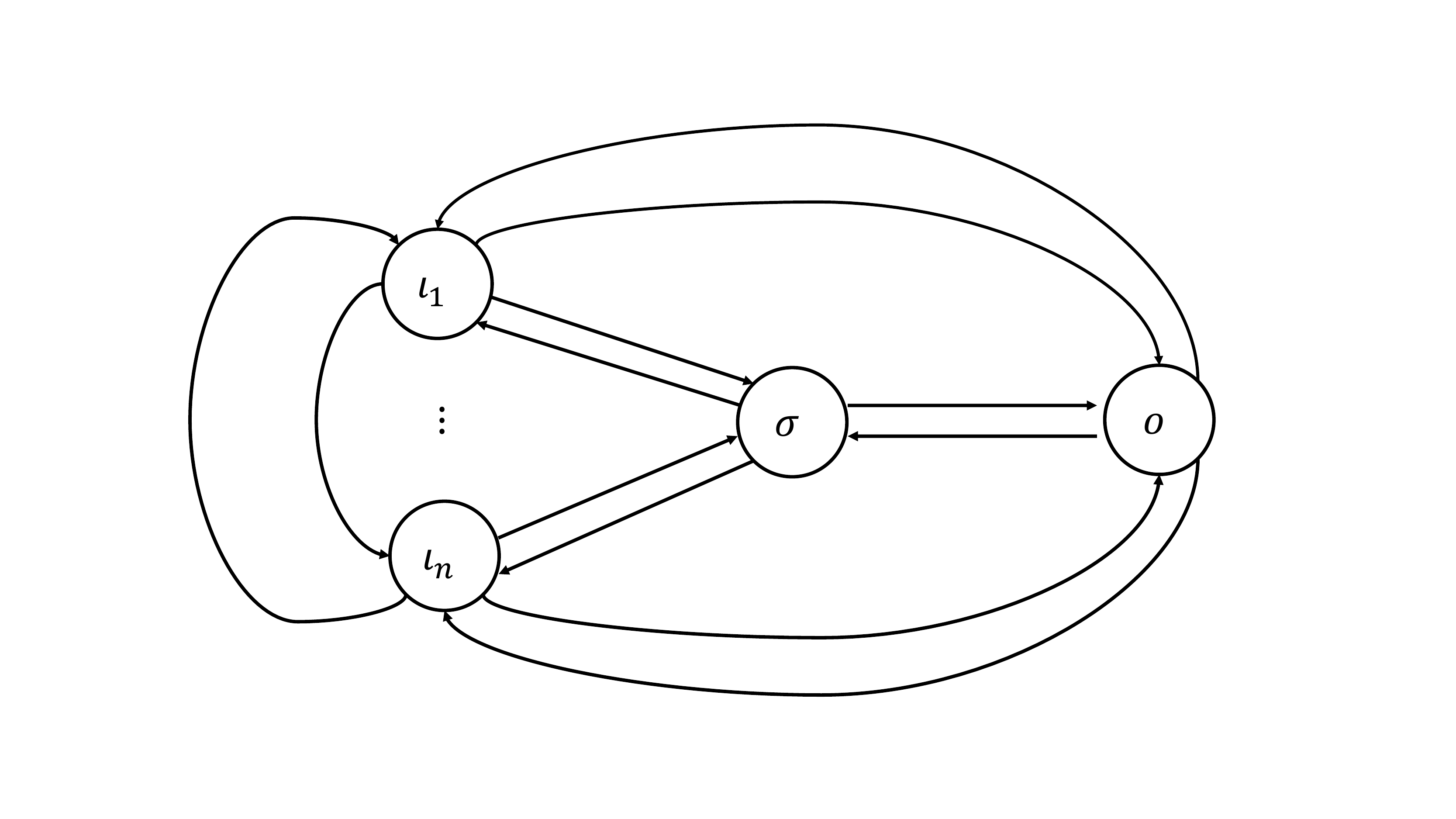}
\renewcommand{\figurename}{Figure}
\caption{A general core network with input nodes $\iota_{1}, \iota_{2}, \cdots, \iota_{n}$, output node $o$ and regulatory nodes $\sigma$.}
\label{abstract_core_network_multiple_input_nodes}
\end{figure}
}

\subsection{Generalized Homeostasis Matrix}

Consider a core network $\mathcal{G}$ with input nodes $\iota_{1}, \iota_{2}, \cdots, \iota_{n}$, output node $o$ and regulatory nodes $\sigma$ which are upstream $o$ and downstream at least one of the input nodes. The generalized homeostasis matrix of $\mathcal{G}$ is given by
\begin{equation} \label{homeostasis_matrix_core_network_multiple_inputs}
    \langle H \rangle = \begin{pmatrix} f_{\iota_{1}, x_{\iota_{1}}} & \cdots & f_{\iota_{1}, x_{\iota_{n}}} & f_{\iota_{1}, x_{\sigma}} & - f_{\iota_{1}, \mathcal{I}} \\
    \vdots & \ddots & \vdots & \vdots & \vdots \\
    f_{\iota_{n}, x_{\iota_{1}}} & \cdots & f_{\iota_{n}, x_{\iota_{n}}} & f_{\iota_{n}, x_{\sigma}} & - f_{\iota_{n}, \mathcal{I}} \\
    f_{\sigma, x_{\iota_{1}}} & \cdots & f_{\sigma, x_{\iota_{n}}} & f_{\sigma, x_{\sigma}} & 0 \\
    f_{o, x_{\iota_{1}}} & \cdots & f_{o, x_{\iota_{n}}} & f_{o, x_{\sigma}} & 0
    \end{pmatrix}
\end{equation}
Expanding the determinant $\det\!\big(\langle H \rangle\big)$ with respect to the last column gives
\begin{equation} \label{homeostasis_matrix_core_network_multiple_inputs_as_a_sum}
    \det\!\big(\langle H \rangle\big) = \sum_{m = 1}^{n} \pm f_{\iota_{m}, \mathcal{I}}\det(H_{\iota_{m}})
\end{equation}
where, according to Definition \ref{definition_text_parcels_homeostasis_matrix},
\begin{equation} \label{definition_parcels_homeostasis_matrix}
    H_{\iota_{m}} = \begin{pmatrix} f_{\iota_{1}, x_{\iota_{1}}} & \cdots & f_{\iota_{1}, x_{\iota_{n}}} & f_{\iota_{1}, x_{\sigma}}\\
    \vdots & \ddots & \vdots & \vdots \\ f_{\iota_{m-1}, x_{\iota_{1}}} & \cdots & f_{\iota_{m-1}, x_{\iota_{n}}} & f_{\iota_{m-1}, x_{\sigma}} \\ f_{\iota_{m+1}, x_{\iota_{1}}} & \cdots & f_{\iota_{m+1}, x_{\iota_{n}}} & f_{\iota_{m+1}, x_{\sigma}} \\  \vdots & \ddots & \vdots & \vdots \\
    f_{\iota_{n}, x_{\iota_{1}}} & \cdots & f_{\iota_{n}, x_{\iota_{n}}} & f_{\iota_{n}, x_{\sigma}} \\
    f_{\sigma, x_{\iota_{1}}} & \cdots & f_{\sigma, x_{\iota_{n}}} & f_{\sigma, x_{\sigma}} \\
    f_{o, x_{\iota_{1}}} & \cdots & f_{o, x_{\iota_{n}}} & f_{o, x_{\sigma}} \end{pmatrix}
\end{equation}

\begin{remark} \rm
Note that when the network $\mathcal{G}$ has only one input node $\iota_1=\iota$, \eqref{definition_parcels_homeostasis_matrix} is exactly the homeostasis matrix defined in Wang \etal \cite{wang20}.
\end{remark}

Let $\mathcal{G}_{m}$ be the subnetwork of $\mathcal{G}$ consisting of $\iota_m$, $o$ and the nodes downstream $\iota_m$ and upstream $o$, i.e., the core subnetwork between $\iota_{m}$ and $o$. 
Denote by $H_{\iota_{m}}^{c}$ the homeostasis matrix of $\mathcal{G}_{m}$, considered as an input-output network, i.e.,
\begin{equation} \label{definition_core_parcels_homeostasis_matrix_reordering}
    \det(H_{\iota_{m}}^{c}) = \begin{vmatrix} f_{\rho, x_{\iota_{m}}} & f_{\rho, x_{\sigma}} \\
    f_{o, x_{\iota_{m}}} & f_{o, x_{\rho}} \end{vmatrix}
\end{equation}
By the definition of $\mathcal{G}$, there may be nodes in $\mathcal{G}$ that are not downstream $\iota_{m}$ (but downstream other input nodes).

The \emph{vestigial subnetwork (with respect to $\iota_m$)} $\mathcal{D}_{m}$ of $\mathcal{G}$ consists of nodes $d_{m}$ that are not downstream $\iota_{m}$ and the arrows that connect these nodes, that is, $\mathcal{D}_{m}$ = $\mathcal{G} \setminus \mathcal{G}_{m}$. As nodes in $\mathcal{D}_{m}$ must not be downstream from $\iota_{m}$ and, consequently not downstream from any node $\rho$, we have, by \eqref{definition_parcels_homeostasis_matrix}, after a permutation of rows and columns that
\begin{equation} \label{definition_parcels_homeostasis_matrix_reordering}
    \det(H_{\iota_{m}}) = \pm \begin{vmatrix} f_{\rho, x_{\iota_{m}}} & f_{\rho, x_{\rho}} & f_{\rho, x_{\mathcal{D}_{m}}} \\
    f_{o, x_{\iota_{m}}} & f_{o, x_{\rho}} & f_{o, x_{\mathcal{D}_{m}}} \\
    0 & 0 & f_{\mathcal{D}_{m}, x_{\mathcal{D}_{m}}}
     \end{vmatrix}
\end{equation}
The Jacobian $J_{\mathcal{D}_{m}}$ of the subnetwork $\mathcal{D}_{m}$ is defined as $\, J_{\mathcal{D}_{m}} \equiv \left[f_{\mathcal{D}_{m}, x_{\mathcal{D}_{m}}}\right]$. Therefore, by \eqref{definition_core_parcels_homeostasis_matrix_reordering} and \eqref{definition_parcels_homeostasis_matrix_reordering} we have
\begin{equation} \label{factorization_det_H_m}
    \det(H_{\iota_m}) = \pm f_{\iota_{m}, \mathcal{I}} \det (J_{\mathcal{D}_{m}}) \det(H_{\iota_m}^{c})
\end{equation}

In order to simplify notation, for $m=1,\ldots,n$ such that $\mathcal{G}$ and $\mathcal{G}_{m}$ have the same nodes, i.e., $\mathcal{G} = \mathcal{G}_{m} \Rightarrow \mathcal{D}_{m} = \mathcal{G} \setminus \mathcal{G}_{m} = \varnothing$, define $J_{\mathcal{D}_{m}} \equiv [1]$. With this convention, it follows from Eqs. \eqref{homeostasis_matrix_core_network_multiple_inputs_as_a_sum} and \eqref{factorization_det_H_m} that
\begin{equation} \label{simplification_of_homeostasis_matrix_core_network_multiple_inputs_as_a_sum}
    \det\!\big(\langle H \rangle\big) = \sum_{m = 1}^{n} \pm f_{\iota_{m}, \mathcal{I}} \det (J_{\mathcal{D}_{m}}) \det(H_{\iota_m}^{c})
\end{equation}

\subsection{Structure of the Homeostasis Blocks}

Before proceeding with the topological characterization of homeostasis types we need to show that, for every $m=1,\ldots,n$, $\det(J_{\mathcal{D}_{m}})$ does not share common factors with $\det\!\big(\langle H \rangle\big)$. We use Frobenius-K\"onig theory to factorize $\det(J_{\mathcal{D}_{m}})$ and then we show that none of these factors of $\det(J_{\mathcal{D}_{m}})$ is a factor of $\det\!\big(\langle H \rangle\big)$.

By Frobenius-K\"onig theory, for every $m=1,\ldots,n$, there are permutation matrices $P_{m}$ and $Q_{m}$ such that
\begin{equation}
    P_{m}J_{\mathcal{D}_{m}}Q_{m} = \begin{pmatrix} D_{m, 1} & * & * & \cdots & * \\
    0 & D_{m, 2} & * & \cdots & * \\
    0 & 0 & D_{m, 3} & \cdots & * \\
    \vdots & \vdots & \vdots & \ddots & \vdots \\
    0 & 0 & 0 & \cdots & D_{m, n_{m}}
    \end{pmatrix}
\end{equation}
where $\det D_{m, j}$ is an irreducible polynomial, for every $1 \leq j \leq n_{m}$. We study the number of self-couplings in each matrix $D_{m, j}$. For this purpose, consider that each matrix $D_{m, j}$ is a square matrix of order $k_{j}$ and suppose that $J_{\mathcal{D}_{m}}$ is a square matrix of order $N$, i.e., that the subnetwork $\mathcal{D}_{m}$ has $N$ nodes. It is straightforward that
\begin{equation}
    N = \sum_{j = 1}^{n_{m}} k_{j}
\end{equation}

\begin{lemma} \label{number_self_couplings}
Each matrix $D_{m, j}$ has exactly $k_{j}$ self-couplings.
\end{lemma}

\begin{proof}
Suppose that $\mathcal{D}_{m}$ is composed by nodes $\rho_{1}, \rho_{2}, \ldots,\rho_{N}$. Therefore, as every node is self-coupled, then one of the summands of $\det J_{\mathcal{D}_{m}}$ is
\begin{equation} \label{summand_with_all_self_couplings}
    f_{\rho_{1}, x_{\rho_{1}}} \cdot f_{\rho_{2}, x_{\rho_{2}}} \cdot \, \cdots \, \cdot f_{\rho_{N}, x_{\rho_{N}}}
\end{equation}
As proved in Wang \etal \cite[Lem 4.6]{wang20}, by Frobenius-König theory, \eqref{summand_with_all_self_couplings} is the product of a summand of each of the determinants $\det D_{m, j}$, which means that in all of these matrices two different self-couplings cannot share the same line or column. By the pigeon hole principle, each matrix $D_{m, j}$ has exactly $k_{j}$ self-couplings.
\end{proof}

Note that as the self-couplings must have different rows and columns from each other, we may choose permutation matrices $P_{m}$ and $Q_{m}$ such that the product $P_{m}J_{\mathcal{D}_{m}}Q_{m}$ have all the self-couplings in the main diagonal, i.e, we may assume that for every $j=1,\ldots,n_{m}$, $D_{m, j}$ has the form:
\begin{equation} \label{jacobian_subnetwork_equal_number_self-couplings}
    D_{m, j} = \begin{pmatrix} f_{\tau_{1}, x_{\tau_{1}}} & f_{\tau_{1}, x_{\tau_{2}}} & \cdots & f_{\tau_{1}, x_{\tau_{k_{j}}}} \\
    f_{\tau_{2}, x_{\tau_{1}}} & f_{\tau_{2}, x_{\tau_{2}}} & \cdots & f_{\tau_{2}, x_{\tau_{k_{j}}}} \\
    \vdots & \vdots & \ddots & \vdots \\
    f_{\tau_{k_{j}}, x_{\tau_{1}}} & f_{\tau_{k_{j}}, x_{\tau_{2}}} & \cdots & f_{\tau_{k_{j}}, x_{\tau_{k_{j}}}}
    \end{pmatrix}
\end{equation}
where $\tau_{1}, \tau_{2}, \ldots, \tau_{k_{j}}$ are nodes of $\mathcal{D}_{m}$, i.e., $D_{m, j}$ is the Jacobian of the subnetwork $\mathcal{D}_{m, j}$ of $\mathcal{D}_{m}$ composed by nodes $\tau_{1}, \tau_{2}, \ldots, \tau_{k_{j}}$.

\begin{lemma} \label{aux_lemma_block_lower_triangular}
Let $\mathcal{K}$ be a proper subnetwork of $\mathcal{D}_{m}$. If nodes in $\mathcal{K}$ are not $\mathcal{D}_{m}$-path equivalent to any node in $\mathcal{D}_{m} \setminus \mathcal{K}$, then upon relabeling nodes, $J_{\mathcal{D}_{m}}$ is block lower triangular.
\end{lemma}

\begin{proof}
The proof is exactly the same as in Wang \etal \cite[Lem 5.3]{wang20}.
\end{proof}

The following theorem fully characterizes the subnetworks $\mathcal{D}_{m, j}$.

\begin{theorem} \label{characterization_dmj}
Let $\mathcal{D}_{m, j}$ be the subnetwork of $\mathcal{D}_{m}$ associated to the matrix $D_{m, j}$. Then:
\begin{enumerate}[(a)]
    \item Nodes in $\mathcal{D}_{m, j}$ are not $\mathcal{D}_{m}$-path equivalent to any node in $\mathcal{D}_{m} \setminus \mathcal{D}_{m, j}$
    \item $\mathcal{D}_{m, j}$ is a path component of $\mathcal{D}_{m}$
\end{enumerate}
\end{theorem}

\begin{proof}
$(a)$ Suppose that there are nodes in $\mathcal{D}_{m, j}$ which are $\mathcal{D}_{m}$-path equivalent to nodes $\mathcal{D}_{m} \setminus \mathcal{D}_{m, j}$. Let $\mathcal{B} \subset \mathcal{D}_{m} \setminus \mathcal{D}_{m, j}$ be the non-empty set of nodes that are $\mathcal{D}_{m}$-path equivalent to nodes in $\mathcal{D}_{m, j}$. Notice that $\mathcal{D}_{m}$-path equivalence is an equivalence relation, and therefore, nodes in $\mathcal{D}_{m, j}$ and $\mathcal{B}$ are not $\mathcal{D}_{m}$-path equivalent to nodes in $(\mathcal{D}_{m} \setminus \mathcal{D}_{m, j}) \setminus \mathcal{B} = \mathcal{D}_{m} \setminus (\mathcal{D}_{m, j} \cup \mathcal{B})$. We now have two possibilities. If $\mathcal{D}_{m} \setminus (\mathcal{D}_{m, j} \cup \mathcal{B}) = \emptyset$, then $\mathcal{D}_{m} = \mathcal{D}_{m, j} \cup \mathcal{B}$. However, $\det(J_{\mathcal{D}_{m, j}}) = \det(D_{m, j})$ is not a factor of $\det(J_{\mathcal{D}_{m, j} \cup \mathcal{B}})$, as proved by Wang \etal \cite[Thm 5.4]{wang20}, which is a contradiction. On the other hand, if $\mathcal{D}_{m} \setminus (\mathcal{D}_{m, j} \cup \mathcal{B}) \neq \emptyset$, then, by Lemma \ref{aux_lemma_block_lower_triangular}, $J_{\mathcal{D}_{m}}$ is a block lower triangular matrix of the form:
\begin{equation}
    J_{\mathcal{D}_{m}} = \begin{pmatrix} U & 0 & 0 \\
    * & J_{\mathcal{D}_{m, j} \cup \mathcal{B}} & 0 \\
    * & * & V
    \end{pmatrix}
\end{equation}
where
\begin{equation}
    J_{\mathcal{D}_{m, j} \cup \mathcal{B}} = \begin{pmatrix} f_{\mathcal{D}_{m, j}, x_{\mathcal{D}_{m, j}}} & f_{\mathcal{D}_{m, j}, x_{\mathcal{B}}} \\ f_{\mathcal{B}, x_{\mathcal{D}_{m, j}}} & f_{\mathcal{B}, x_{\mathcal{B}}}
    \end{pmatrix}
\end{equation}
Again, $\det(J_{\mathcal{D}_{m, j}}) = \det(D_{m, j})$ is not a factor of $\det(J_{\mathcal{D}_{m, j} \cup \mathcal{B}})$, which is a contradiction. Therefore, we conclude that nodes in $\mathcal{D}_{m, j}$ are not $\mathcal{D}_{m}$-path equivalent to nodes in $\mathcal{D}_{m} \setminus \mathcal{D}_{m, j}$. \\
$(b)$ It is enough to show that $\mathcal{D}_{m, j}$ is path connected, as, by item $(a)$ of this lemma, if a $\mathcal{D}_{m}$-path component contains nodes in $\mathcal{D}_{m} \setminus \mathcal{D}_{m, j}$, then this path component does not contain any node in $\mathcal{D}_{m, j}$. Suppose that $\mathcal{D}_{m, j}$ is not path connected. In that case, we may split in two subnetworks $\mathcal{A}$ and $\mathcal{B}$ such that
\begin{equation}
    J_{\mathcal{D}_{m, j}} = \begin{pmatrix} J_{\mathcal{A}} & 0 \\
    * & J_{\mathcal{B}}
    \end{pmatrix}
\end{equation}
This is however a contradiction as by hypothesis $\det D_{m, j}$ is irreducible. Therefore $\mathcal{D}_{m, j}$ is a $\mathcal{D}_{m}$-path component.
\end{proof}

\begin{proposition} \label{every_block_D_m_j_has_simple_node}
For every $j = 1, \ldots, n_{m}$, $\mathcal{D}_{m, j}$ contains an $\iota_{k}$-simple node, for some $k = 1, 2, \ldots, n$, $k \neq m$.
\end{proposition}

\begin{proof}
Consider a subnetwork $\mathcal{D}_{m, j}$. Nodes in $\mathcal{D}_{m, j}$ are not downstream from $\iota_{m}$, and therefore, as $\mathcal{G}$ is a core network, for every node $\rho$ in $\mathcal{D}_{m, j}$, there is $k \in\{1,\ldots,n\}$, $k \neq m$, such that $\rho$ is downstream from $\iota_{k}$ and $\iota_{k}$ is not downstream from $\iota_{m}$. If $\rho$ is an $\iota_{k}$-simple node, the corollary is proved. On the other hand, suppose that $\rho$ is an $\iota_{k}$-appendage node. By definition of $\iota_{k}$-appendage node, there is a path $\iota_{k} \rightarrow \sigma_{1} \rightarrow \cdots \rightarrow \sigma_{p} \rightarrow \rho \rightarrow \sigma_{p+1} \rightarrow \cdots \rightarrow o$ such that at least one node $\tau$ in this path appears before and after $\rho$. Moreover, nodes in the path between $\tau$ and $\rho$ (and vice-versa) must not be downstream from $\iota_{m}$, and therefore, they are paths of $\mathcal{D}_{m}$. If $\iota_{k}$ satisfies this condition, then $\rho$ is $\mathcal{D}_{m}$-path equivalent to $\iota_{k}$, which is an $\iota_{k}$-simple node. If $\iota_{k}$ does not satisfy, consider the smallest $r$ such that $\sigma_{r}$ satisfies the condition. Then $\sigma_{r}$ is an $\iota_{k}$-simple node and $\rho$ is $\mathcal{D}_{m}$-path equivalent to $\sigma_{r}$. In both cases, there is an $\iota_{k}$-simple node which belongs to the same $\mathcal{D}_{m}$-path equivalence class of $\rho$. By theorem $\ref{characterization_dmj}$, as $\mathcal{D}_{m, j}$ is a path component of $\mathcal{D}_{m}$, this $\iota_{k}$-simple node belongs to $\mathcal{D}_{m, j}$.
\end{proof}

\begin{proposition} \label{det_J_Dm_does_not_share_common_factors_det_H}
For every $m = 1, \ldots, n$ such that $\mathcal{D}_{m}$ is not empty, factors of $\det(J_{\mathcal{D}_{m}})$ are not factors of $\det\!\big(\langle H \rangle\big)$.
\end{proposition}

\begin{proof}
We prove this statement by contradiction. Suppose that there is $m$ such that $\mathcal{D}_{m}$ is not empty and there is an irreducible factor $\det(D_{m, j})$ of $\det(J_{\mathcal{D}_{m}})$ which is also a factor of $\det \!\big(\langle H \rangle\big)$. By equation \eqref{simplification_of_homeostasis_matrix_core_network_multiple_inputs_as_a_sum}, $\det\!\big(\langle H \rangle\big)$ can be seen as an homogeneous polynomial with variables $f_{\iota_{p}, \mathcal{I}}$ and respective coefficients $\pm \det(J_{\mathcal{D}_{p}}) \det(H_{\iota_p}^{c})$, for every $p = 1, \ldots, n$. By Lemma \ref{lemma_commum_factors_det_H}, $\det(D_{m,j})$ must be a factor of $\det(J_{\mathcal{D}_{p}}) \det(H_{\iota_p}^{c})$, for every $p = 1, \ldots, n$. Consider now a node $\rho$ in $\mathcal{D}_{m, j}$. As $\mathcal{G}$ is a core network, there is $k\in\{1,\ldots,n\}$, $k \neq m$, such that $\rho$ is downstream from $\iota_{k}$ and $\iota_{k}$ is not downstream from $\iota_{m}$. As $\det(D_{m, j})$ is irreducible, we conclude that $\det(D_{m, j})$ must be a factor of $\det(J_{\mathcal{D}_{k}})$ or of  $\det(H_{\iota_k}^{c})$. As $\rho$ is in the core network between $\iota_{k}$ and $o$, $\det(D_{m, j})$ must be a factor of  $\det(H_{\iota_k}^{c})$. We have already proved in Lemma \ref{number_self_couplings} that the number of self-couplings in $D_{m, j}$ is equal to the order of $D_{m, j}$. As shown in Wang \etal \cite[Lem 5.2]{wang20}, this means that $\det(D_{m, j})$ is an appendage block and therefore all nodes in $\mathcal{D}_{m, j}$ should be $\iota_{k}$-appendage. However, by Proposition \ref{every_block_D_m_j_has_simple_node}, $\mathcal{D}_{m, j}$ has an $\iota_{k}$-simple node, which is a contradiction.
\end{proof}

\subsection{Topological Characterization of Homeostasis Blocks} \label{topological_characterization}

Our aim now is to associate the factorization of $\det\!\big(\langle H \rangle\big)$ with the network topology. In order to do this, recall that, by equation \eqref{simplification_of_homeostasis_matrix_core_network_multiple_inputs_as_a_sum}, we have
\begin{equation*}
     \det\!\big(\langle H \rangle\big) = \sum_{m = 1}^{n} \pm f_{\iota_{m}, \mathcal{I}} \det(J_{\mathcal{D}_{m}}) \det(H_{\iota_m}^{c})
\end{equation*}
and hence we can consider the expression of $\det\!\big(\langle H \rangle\big)$ as an homogeneous polynomial of degree $1$ on variables $f_{\iota_{m}, \mathcal{I}}$ and respective coefficients $\pm \det(J_{\mathcal{D}_{m}}) \det(H_{\iota_m}^{c})$ for all $m = 1, \ldots, n$.

By Lemma \ref{lemma_commum_factors_det_H}, to factorize $\det\!\big(\langle H \rangle\big)$, we must search for common factors of the coefficients $\pm \det(J_{\mathcal{D}_{m}}) \det(H_{\iota_{m}}^{c})$. On the other hand, Proposition \ref{det_J_Dm_does_not_share_common_factors_det_H} implies that, for all $m=1,\ldots,n$, $\det(J_{\mathcal{D}_{m}})$ does not share common factors with any term $\det(J_{\mathcal{D}_{p}}) \det(H_{\iota_{p}}^{c})$, for all $p\in\{1,\ldots,n\}$, $p \neq m$. Therefore, in order to look for common factors between the coefficients, we must search for common factors of $\det(H_{\iota_{m}}^{c})$.

Bearing all the facts above in mind and applying Frobenius-König theory to $\det(H_{\iota_{m}}^{c})$, there are square matrices $B_{1}, B_{2}, \ldots, B_{s}$ such that $\det(B_{j})$ is irreducible for $j = 1, \ldots, \eta$ and for every $m = 1, \ldots, n$, there exists a square matrix $C_{\iota_{m}}$ such that
\begin{equation} \label{common_factor_H_iota_m}
    \det(H_{\iota_{m}}^{c}) = \pm \left(\det(B_{1}) \cdot \det(B_{2}) \cdot \, \cdots \, \cdot \det(B_{s})\right) \cdot \det(C_{\iota_{m}})
\end{equation}
where $\det(C_{\iota_{1}})$,\ldots, $\det(C_{\iota_{n}})$ do not share common factors. In case $\det(H_{\iota_{1}}^{c})$, \ldots, $\det(H_{\iota_{n}}^{c})$ do not share common factors, we can consider that $H_{\iota_{m}}^{c} = C_{\iota_{m}}$.

From Eqs. \eqref{simplification_of_homeostasis_matrix_core_network_multiple_inputs_as_a_sum} and \eqref{common_factor_H_iota_m} it follows that
\begin{equation} \label{factoring_det_weighted_H}
    \det\!\big(\langle H \rangle\big) = \left(\det(B_{1}) \cdot \det(B_{2})  \cdots  \det(B_{s})\right) \left( \sum_{m = 1}^{n} \pm f_{\iota_{m}, \mathcal{I}} \det(J_{\mathcal{D}_{m}}) \det(C_{\iota_{m}}) \right)
\end{equation}
By Corollary \ref{corollary_irreductible_coefficients_irreductible_polynomial}, the expression
\begin{equation} \label{term_input_counter_weight_homeostasis}
    \sum_{m = 1}^{n} \pm f_{\iota_{m}, \mathcal{I}} \det(J_{\mathcal{D}_{m}}) \det(C_{\iota_{m}}) \equiv \det(C)
\end{equation}
is irreducible. Substituting \eqref{term_input_counter_weight_homeostasis} into \eqref{factoring_det_weighted_H} gives
\begin{equation}
    \det\!\big(\langle H \rangle\big) = \det(B_{1}) \cdot \det(B_{2})  \cdots \det(B_{s}) \cdot \det(C)
\end{equation}
where each of the matrices $B_{1}, \ldots, B_{s}$ and $C$ is an irreducible homeostasis block. As explained in Subsection \ref{combinatorial_characterization_homeostasis}, $C$ is the \textit{input counterweight homeostasis block}. Observe that the difference between the matrices $B_{1}, B_{2}, \ldots, B_{s}$ and $C$ is that the terms $f_{\iota_{m}, \mathcal{I}}$ do not appear in any of the matrices $B_{j}$.

\begin{corollary}
Every core network $\mathcal{G}$ with multiple input nodes supports input counterweight homeostasis.
\end{corollary}

\begin{proof}
The proof is straightforward as $\det\!\big(\langle H \rangle\big)$ is always a multiple of an irreducible homogeneous polynomial of degree $1$ on variables $f_{\iota_{1}, \mathcal{I}}, \ldots, f_{\iota_{n}, \mathcal{I}}$
\end{proof}

\begin{remark} \rm
Although input counterweight homeostasis does not occur in networks with only one input node, it is interesting to note that in these networks there is a corresponding matrix $C$. In fact, from equation \eqref{weighted_homeostasis_matrix_definition} and considering networks with only one input node $\iota$, we have
\begin{equation}
    C = [- f_{\iota, \mathcal{I}}] \qquad\Rightarrow\qquad 
    \det(C) = - f_{\iota, \mathcal{I}} \neq 0
\end{equation}
As, by hypothesis, $f_{\iota, \mathcal{I}} \neq 0$, the input counterweight homeostasis is never present in networks with only one input node.
\end{remark}

In Wang \etal \cite{wang20} the classification of the irreducible homeostasis blocks was based on the number of self-couplings. The same arguments can be used in our context. Initially, as $H_{\iota_{m}}^{c}$ is the homeostasis matrix of $\mathcal{G}_{m}$, we conclude that $B_{1}, \ldots, B_{s}$ are irreducible homeostasis blocks of each of the core networks $\mathcal{G}_{m}$ with one input node. Thus, we can conclude that, for every $j=1,\ldots,\eta$, $B_{j}$ has exactly $k_{j}$ self-couplings or $k_{j} - 1$ self-couplings, where $k_{j}$ is the order of $B_{j}$. In order to maintain the same terminology employed in \cite{wang20}, we call $B_{j}$ an \textit{appendage homeostasis block} when $B_{j}$ has exactly $k_{j}$ self-couplings, and a \textit{structural homeostasis block} otherwise (see Definition \ref{definition_appendage_and_structural}). 
 
\subsubsection{Appendage Homeostasis}

Recall that each appendage homeostasis block $B_{j}$ of order $k_{j}$ has exactly $k_{j}$ self-couplings. In an analogous way of equation \eqref{jacobian_subnetwork_equal_number_self-couplings}, $B_{j}$ must be the Jacobian matrix of a subnetwork $\mathcal{K}_j$.

\begin{theorem} \label{characterization_appendage_homeostasis_each_G_m}
Let $\mathcal{K}_j$ be a subnetwork of $\mathcal{G}$ associated with an appendage homeostasis block
$B_{j}$. Then the following statements are valid:
\begin{enumerate}[(a)]
    \item Each node in $\mathcal{K}_j$ is an $\iota_m$-appendage node, for every $m = 1, \ldots, n$.
    \item For every $\iota_{m}o$-simple path $S$, nodes in $\mathcal{K}_j$ are not $C_{m}S$-path equivalent to any node in $C_{m}S \setminus \mathcal{K}_j$, for all $m = 1, \ldots, n$.
    \item $\mathcal{K}_{j}$ is a path component of  $\mathcal{A}_{\mathcal{G}_{m}}$, for all $m = 1, \ldots, n$.
\end{enumerate}
\end{theorem}

\begin{proof}
The diagonal block $B_{j}$ must be an appendage homeostasis block of each core subnetwork $\mathcal{G}_{m}$. By \cite[Lem 5.2 and Thm 5.4]{wang20}, the theorem follows.
\end{proof}

Now we can characterize $\mathcal{K}_j$ with respect to the whole core network $\mathcal{G}$.

\begin{theorem} \label{characterization_appendage_homeostasis_general_G}
Let $\mathcal{K}_j$ be a subnetwork of $\mathcal{G}$ associated with an appendage homeostasis block
$B_{j}$. Then the following statements are valid:
\begin{enumerate}[(a)]
\item Each node in $\mathcal{K}_j$ is an absolutely appendage node.
\item For every $\iota_{m}o$-simple path $S$, nodes in $\mathcal{K}_j$ are not $CS$-path equivalent to any node in $CS \setminus \mathcal{K}_j$, for all $m = 1, \ldots, n$.
\item $\mathcal{K}_{j}$ is a path component of  $\mathcal{A}_{\mathcal{G}}$.
\end{enumerate}
\end{theorem}

\begin{proof}
$(a)$ Each node in $\mathcal{K}_{j}$ is $\iota_{m}$-appendage, for all $m = 1, \ldots, n$, which means that each node in $\mathcal{K}_{j}$ is absolutely appendage (see definition \ref{combinatorial_definitions_in_network_multiple_input_nodes}). \\
$(b)$ Suppose that for some $m=1,\ldots,n$, there is an $\iota_{m}o$-simple path such that there is a node $\rho$ in $\mathcal{K}_{j}$ that is $CS$-path equivalent to a node $\tau$ in $CS \setminus \mathcal{K}_{j}$, which means that there are paths in $CS$ such that $\rho \rightarrow \tau$ and $\tau \rightarrow \rho$. By item $(a)$ above, $\rho$ is $\iota_{p}$-appendage for all $p = 1, \ldots, n$, and, in particular, $\rho$ is $\iota_{m}$-appendage $\Rightarrow$ every node downstream from $\rho$ is downstream from $\iota_{m}$ and belongs to $\mathcal{G}_{m} \Rightarrow$ there are paths $\rho \rightarrow \tau$ and $\tau \rightarrow \rho$ in $C_{m}S$, and therefore $\rho$ is $C_{m}S$-path equivalent to $\tau \in C_{m}S \setminus \mathcal{K}_j$, which is a contradiction by Theorem \ref{characterization_appendage_homeostasis_each_G_m}. \\
$(c)$ By statement $(a)$ above, $\mathcal{K}_{j} \subseteq \mathcal{A}_{\mathcal{G}_{m}}$, for all $m = 1, \ldots, n$, which means that all nodes in $\mathcal{K}_{j}$ belong to the same $\mathcal{A}_{\mathcal{G}}$-path equivalence class. Consider the path component $\mathcal{T}$ of $\mathcal{A}_{\mathcal{G}}$ such that $\mathcal{K}_{j} \subseteq \mathcal{T}$. We have $\mathcal{T} \subseteq \mathcal{A}_{\mathcal{G}} \Rightarrow \mathcal{T} \subseteq \mathcal{A}_{\mathcal{G}_{m}}$, for all $m = 1, \ldots, n$. As $\mathcal{K}_{j}$ is a path component of $\mathcal{A}_{\mathcal{G}_{m}}$, for all $m = 1, \ldots, n$, we conclude that $\mathcal{K}_{j} = \mathcal{T}$.
\end{proof}

Wang \etal \cite{wang20} proved that the conditions of Theorem \ref{characterization_appendage_homeostasis_each_G_m} are not only necessary, but also sufficient to determine subnetworks associated to appendage homeostasis in networks with only one input node. Theorem \ref{reverse_theorem_appendage_block_general_G} generalizes this result to networks with multiple input nodes.

\begin{proposition} \label{reverse_theorem_appendage_block_each_G_m}
Suppose $\mathcal{K}_{j}$ is a subnetwork of $\mathcal{G}$ such that
\begin{enumerate}[(a)]
    \item $\mathcal{K}_{j}$ is an $\mathcal{A}_{\mathcal{G}_{m}}$-path component, for all $m = 1, \ldots, n$.
    \item For every $\iota_{m}o$-simple path $S$, nodes in $\mathcal{K}_j$ are not $C_{m}S$-path equivalent to any node in $C_{m}S \setminus \mathcal{K}_j$, for all $m = 1, \ldots, n$.
\end{enumerate}
then $\det(J_{\mathcal{K}_{j}})$ is an irreducible factor of $\det \langle H \rangle$.
\end{proposition}

\begin{proof}
By \cite[Thm 7.1]{wang20}, statements $(a)$ and $(b)$ mean that $\det(J_{\mathcal{K}_{j}})$ is an irreducible factor of $\det H_{\iota_{m}}^{c}$, for every $m = 1, \ldots, n$. Therefore, by \eqref{simplification_of_homeostasis_matrix_core_network_multiple_inputs_as_a_sum}, $\det(J_{\mathcal{K}_{j}})$ is an irreducible factor of $\det\!\big(\langle H \rangle\big)$.
\end{proof}

\begin{theorem} \label{reverse_theorem_appendage_block_general_G}
Suppose $\mathcal{K}_{j}$ is a subnetwork of $\mathcal{G}$ such that
\begin{enumerate}[(a)]
    \item $\mathcal{K}_{j}$ is an $\mathcal{A}_{\mathcal{G}}$-path component.
    \item For every $\iota_{m}o$-simple path $S$, nodes in $\mathcal{K}_j$ are not $CS$-path equivalent to any node in $CS \setminus \mathcal{K}_j$, for all $m = 1, \ldots, n$.
\end{enumerate}
then $\det(J_{\mathcal{K}_{j}})$ is an irreducible factor of $\det\!\big(\langle H \rangle\big)$.
\end{theorem}

\begin{proof}
We begin proving that assertion $(b)$ above implies assertion $(b)$ of Proposition \ref{reverse_theorem_appendage_block_each_G_m}. In fact, for every $m = 1, \ldots, n$ and every $\iota_{m}o$-simple path $S$, $C_{m}S \subseteq CS$, and therefore \ref{reverse_theorem_appendage_block_general_G}$(b) \Rightarrow$ \ref{reverse_theorem_appendage_block_each_G_m}$(b)$. On the other hand, if $\mathcal{K}_{j}$ is an $\mathcal{A}_{\mathcal{G}}$-path component, then, as $\mathcal{A}_{\mathcal{G}} = \mathcal{A}_{\mathcal{G}_{1}} \cap \cdots \cap \mathcal{A}_{\mathcal{G}_{n}}$, the nodes in $\mathcal{K}_{j}$ belong to the same $\mathcal{A}_{\mathcal{G}_{m}}$-path equivalence class, for all $m = 1, \ldots, n$. Consider the $\mathcal{A}_{\mathcal{G}_{m}}$-path component $\mathcal{T}_{m}$ such that $\mathcal{K}_{j} \subseteq \mathcal{T}_{m}$. If there is a $m$ such that $\mathcal{K}_{j} \neq \mathcal{T}_{m}$, then for this $m$, for every $\iota_{m}o$-simple path $S$, as $\mathcal{T}_{m} \subseteq C_{m}S$, nodes in $\mathcal{K}_{j}$ are $C_{m}S$-path equivalent to nodes in $C_{m}S \setminus \mathcal{K}_{j}$, which is a contradiction. Therefore, $\mathcal{K}_{j}$ is a path component of $\mathcal{A}_{\mathcal{G}_{m}}$, for all $m = 1, \ldots, n$. As both assumptions of Proposition \ref{reverse_theorem_appendage_block_each_G_m} are satisfied, we conclude that $\det(J_{\mathcal{K}_{j}})$ is an irreducible factor of $\det\!\big( \langle H \rangle\big)$.
\end{proof}

\subsubsection{Structural Homeostasis}

In order to characterize structural homeostasis in networks with multiple input nodes we start with the absolutely super-simple nodes. In networks with one input node, the super-simple nodes may be ordered by simple paths. We can then apply this result to the subnetworks $\mathcal{G}_{m}$.

\begin{lemma} \label{order_i_m_super_simple_nodes}
For every $m = 1, \ldots, n$, the $\iota_{m}$-super simple nodes in $\mathcal{G}_{m}$ can be uniquely ordered by $\iota_{m} > \rho_{m, 1} > \rho_{m, 2} > \cdots > \rho_{m, p_{m}} > o$, where $a > b$ when $b$ is downstream from $a$ by all $\iota_{m}o$-simple paths.
\end{lemma}

\begin{proof}
For every $m = 1, \ldots, n$, $\mathcal{G}_{m}$ is a core subnetwork with only one input node $\iota_{m}$. Hence, we can apply the result about ordering obtained by Wang \etal \cite[Lem 6.1]{wang20} to conclude that the $\iota_{m}$-super-simple nodes may be uniquely ordered by all $\iota_{m}o$-simple paths.
\end{proof}

We now extend Lemma \ref{order_i_m_super_simple_nodes} to the absolutely super-simple nodes of $\mathcal{G}$.

\begin{lemma} \label{order_absolutely_super_simple_nodes}
The absolutely super-simple nodes in $\mathcal{G}$ can be uniquely ordered by $\rho_{1} > \rho_{2} > \cdots > \rho_{p} > o$, where $a > b$ when $b$ is downstream from $a$ by all $\iota_{m}o$-simple paths.
\end{lemma}

\begin{proof}
By Lemma \ref{order_i_m_super_simple_nodes}, for each $m=1,\ldots,n$, we can order the absolutely super-simple according to the $\iota_{m}o$-simple paths. Suppose there are $m_{1}, m_{2} \in\{1,\ldots,n\}$, $m_{1} \neq m_{2}$, such that there are absolutely super-simple nodes $\rho_{i}, \rho_{j}$, $\rho_{i} \neq \rho_{j}$, such that $\rho_{i} > \rho_{j}$ according to $\iota_{m_{1}}o$-simple paths and $\rho_{j} > \rho_{i}$ according to $\iota_{m_{2}}o$-simple paths. This means that there are an $\iota_{m_{1}}o$-simple path $S_{1}$: $\iota_{m_{1}} \rightarrow \cdots \rightarrow \rho_{i} \rightarrow \cdots \rightarrow \rho_{j} \rightarrow \cdots \rightarrow o$ and an $\iota_{m_{2}}o$-simple path $S_{2}$: $\iota_{m_{2}} \rightarrow \cdots \rightarrow \rho_{j} \rightarrow \cdots \rightarrow \rho_{i} \rightarrow \cdots \rightarrow o$. By definition of $S_{1}$, there is an $\iota_{1}\rho_{i}$-simple path which does not pass by $\rho_{j}$, and by definition of $S_{2}$ there is an $\rho_{i}o$-simple path which does not pass by $\rho_{j}$. Therefore, we conclude that there must be an $\iota_{1}o$-simple path which does not pass by $\rho_{j}$, contradicting the fact that $\rho_{j}$ is an absolutely super-simple node.
\end{proof}

\begin{corollary}
If $\mathcal{G}$ is a core network with input nodes $\iota_{1}, \ldots, \iota_{n}$ and output node $o$, then at most one input node of $\mathcal{G}$ is an absolutely super-simple node.
\end{corollary}

\begin{proof}
Suppose that $\mathcal{G}$ has two input nodes $\iota_{i}$ and $\iota_{j}$ which are absolutely super-simple nodes. Notice that if we order the super-simple nodes according to the $\iota_{i}$-simple paths, we have $\iota_{i} > \iota_{j}$, and if we order them according to the $\iota_{j}$-simple paths, we have $\iota_{j} > \iota_{i}$. The different orderings contradict Lemma \ref{order_absolutely_super_simple_nodes}.
\end{proof}

Let $\rho_{k}>\rho_{k+1}$ be adjacent $\iota_{m}$-super-simple nodes for some $m \in\{1,\ldots,n\}$.
An $\iota_{m}$-simple node $\rho$ is \textit{between} $\rho_{k}$ and $\rho_{k+1}$ if there exists an $\iota_{m}o$-simple path that includes $\rho_{k}$ to $\rho$ to $\rho_{k+1}$ in that order.

\begin{lemma} \label{relation_i_m_simple_node_and_i_m_super_simple_node}
Every $\iota_{m}$-simple node, which is not $\iota_{m}$-super-simple, lies uniquely between two adjacent $\iota_{m}$-super-simple nodes.
\end{lemma}

\begin{proof}
Since $\mathcal{G}_{m}$ can be seen as the core subnetwork between $\iota_{m}$ and $o$, i.e., $\mathcal{G}_{m}$ is a core subnetwork with one input node, the result follows from \cite[Lem 6.2]{wang20}.
\end{proof}

Observe that each network $\mathcal{A}_{\mathcal{G}_{m}}$ can be partitioned in $(p_{m} + q_{m})$ $\mathcal{A}_{\mathcal{G}_{m}}$-path components $\mathcal{A}_{m, 1}, \ldots, \mathcal{A}_{m, p_{m}}, \mathcal{B}_{m, 1}, \ldots, \mathcal{B}_{m, q_{m}}$, where nodes in components $\mathcal{A}_{m, i}$ are not $C_{m}S_{m}$-path equivalent to any node in $C_{m}S_{m} \setminus \mathcal{A}_{m, i}$, for every $\iota_{m}o$-simple path $S_{m}$. Whereas this is not the case for nodes in the $\mathcal{A}_{\mathcal{G}_{m}}$-path components $\mathcal{B}_{m, i}$, i.e., for every $\mathcal{B}_{m, i}$, there is an $\iota_{m}o$-simple path $S_{m}$ such that nodes in $\mathcal{B}_{m, i}$ are $C_{m}S_{m}$-path equivalent to nodes in $C_{m}S_{m} \setminus \mathcal{B}_{m, i}$. Clearly, for every $m=1,\ldots,n$ one has that
\begin{equation} \label{partition_A_G_m_two_types_path_omponents}
    \mathcal{A}_{\mathcal{G}_{m}} = \left( \mathcal{A}_{m, 1} \dot{\cup} \cdots \dot{\cup} \mathcal{A}_{m, p_{m}}\right) \dot{\cup} \left( \mathcal{B}_{m, 1} \dot{\cup} \cdots \dot{\cup} \mathcal{B}_{m, q_{m}}\right)
\end{equation}

\begin{lemma} \label{A_G_m_path_component_and_i_m_structural_subnetwork}
Consider the $\mathcal{A}_{\mathcal{G}_{m}}$-path component $\mathcal{B}_{m, i}$ and suppose there is an $\iota_{m}o$-simple path $S_{m}$ such that nodes in $\mathcal{B}_{m, i}$ are $C_{m}S_{m}$-path equivalent to nodes in $C_{m}S_{m} \setminus \mathcal{B}_{m, i}$. Then there is at least an $\iota_{m}$-simple node $\rho$ in $C_{m}S_{m} \setminus \mathcal{B}_{m, i}$ such that nodes in $\mathcal{B}_{m, i}$ are $C_{m}S_{m}$-path equivalent to $\rho$. Moreover, the $\iota_{m}$-simple nodes in $C_{m}S_{m} \setminus \mathcal{B}_{m, i}$ that are $C_{m}S_{m}$-path equivalent to $\mathcal{B}_{m, i}$, including $\rho$, are not $\iota_{m}$-super-simple and are contained in a unique $\iota_{m}$-super-simple subnetwork.
\end{lemma}

\begin{proof}
This is proved in \cite[Lem 6.3]{wang20}.
\end{proof}

Let $\rho_{k}> \rho_{k+1}$ be adjacent $\iota_{m}$-super-simple nodes, for some $m \in\{1,\ldots,n\}$.
The $\iota_{m}$\textit{-super-simple subnetwork}, denoted $\mathcal{L}_{m}(\rho_{k}, \rho_{k+1})$, is the subnetwork whose nodes are $\iota_{m}$-simple nodes between $\rho_{k}$ and $\rho_{k+1}$ and whose arrows are arrows of $\mathcal{G}_{m}$ connecting nodes in $\mathcal{L}_{m}(\rho_{k}, \rho_{k+1})$.

Let $\rho_{k}$ and $\rho_{k+1}$ be adjacent $\iota_{m}$-super-simple nodes in $\mathcal{G}_{m}$. The $\iota_{m}$\textit{-super-simple structural subnetwork} $\mathcal{L}'_{m}(\rho_{k}, \rho_{k+1})$ is the input-output subnetwork consisting of nodes in $\mathcal{L}_{m}(\rho_{k}, \rho_{k+1}) \cup \mathcal{B}_{m}$, where $\mathcal{B}_{m}$ consists of all appendage nodes that are $C_{m}S_{m}$-path equivalent to nodes in $\mathcal{L}_{m}(\rho_{k}, \rho_{k+1})$ for some $\iota_{m}o$-simple path $S_{m}$. Arrows of $\mathcal{L}'_{m}(\rho_{k}, \rho_{k+1})$ are arrows of $\mathcal{G}_{m}$ that connect nodes in $\mathcal{L}'_{m}(\rho_{k}, \rho_{k+1})$.

As $\mathcal{G}_{m}$ is a core subnetwork with only one input node, the homeostasis matrix of each $\iota_{m}$-super-simple structural subnetwork $\mathcal{L}'_{m}(\rho_{k}, \rho_{k+1})$ ($H(\mathcal{L}'_{m}(\rho_{k}, \rho_{k+1}))$) is an irreducible structural homeostasis block of the homeostasis matrix of $\mathcal{G}_{m}$ ($H_{\iota_{m}}^{c}$) (see \cite[Thm 6.11]{wang20}). We have already proved in Proposition \ref{det_J_Dm_does_not_share_common_factors_det_H} that the irreducible structural homeostasis blocks of $\langle H \rangle$ are also structural homeostasis blocks of each of the matrices $H_{\iota_{m}}^{c}$. Therefore, we now study under which conditions there is an super-simple structural subnetwork shared by all the core subnetworks $\mathcal{G}_{m}$.

\begin{lemma} \label{basic_properties_super_simple_subnetwork}
Let $\rho_{k} > \rho_{k+1}$ be two adjacent absolutely super-simple nodes. Then the following properties are valid:
\begin{enumerate} [(a)]
\item $\rho_{k}$ and  $\rho_{k+1}$ are adjacent $\iota_{m}$-super-simple nodes, for every $m=1,\ldots,n$.
\item Every $\iota_{m}$-simple node in $\mathcal{L}_{m}(\rho_{k}, \rho_{k+1})$ is an absolutely simple node, for every $m=1,\ldots,n$.
\item $\mathcal{L}_{m}(\rho_{k}, \rho_{k+1}) = \mathcal{L}(\rho_{k}, \rho_{k+1})$, for every $m=1,\ldots,n$.
\end{enumerate}
\end{lemma}

\begin{proof}
$(a)$ Suppose there is $m \in\{1,\ldots,n\}$ such that there is an $\iota_{m}$-super-simple node $\rho$ which satisfies $\rho_{k} > \rho > \rho_{k+1}$. Then, every $\iota_{m}o$-simple path passes by $\rho_{k}$, $\rho$ and $\rho_{k+1}$, in that order. Suppose now there is $j \in \{1,\ldots,n\}$, $j \neq m$, such that there is an $\iota_{j}o$-simple path which does not pass by $\rho$. That means that there is a path in $\mathcal{G}$ between $\rho_{k}$ and $\rho_{k+1}$ which does not pass by $\rho$. As there are an $\iota_{m}\rho_{k}$-simple path, a $\rho_{k}\rho_{k+1}$-simple path and a $\rho_{k+1}o$-simple path such that neither of them passes by $\rho$, we can obtain an $\iota_{m}o$-simple path which does not pass by $\rho$, which is a contradiction. \\
$(b)$ By item $(a)$ above $\rho_{k}, \rho_{k+1}$ must be adjacent $\iota_{m}$-super-simple nodes, and so $\mathcal{L}_{m}(\rho_{k}, \rho_{k+1})$ is well defined. For some $m \in \{1,\ldots,n\}$, consider an $\iota_{m}$-simple node $\rho$ which is between $\rho_{k}$ and $\rho_{k+1}$. By definition, there is an $\iota_{m}o$-simple path $S_{m}$ which passes by 
$$
\iota_{m} \rightarrow \cdots \rightarrow \rho_{k} \rightarrow \cdots \rightarrow \rho \rightarrow \cdots \rightarrow \rho_{k+1} \rightarrow \cdots \rightarrow o
$$ 
in that order. Let's take $j \in \{1,\ldots,n\}$. Consider any $\iota_{j}o$-simple path $S_{j}$. $S_{j}$ and $S_{m}$ have some nodes in common (at least $\rho_{k}, \rho_{k+1}$ and $o$), and these nodes must appear in the same order. Therefore, we can build a path $S_{j}^{*}$ taking the $S_{j}$ stretch from $\iota_{j}$ to $\rho_{k}$ and the $S_{m}$ stretch from $\rho_{k}$ to $o$ passing by $\rho$ (and $\rho_{k+1}$). By the argument above, $S_{j}^{*}$ is an $\iota_{j}o$-simple path, and therefore $\rho$ is an $\iota_{j}$-simple node. As this process may be done to any $j \in \{1,\ldots,n\}$, then $\rho$ is an absolutely simple node. \\
$(c)$ For every $m=1,\ldots,n$, all absolutely simple nodes between $\rho_{k}$ and $\rho_{k+1}$ are $\iota_{m}$-simple nodes, and therefore $\mathcal{L}(\rho_{k}, \rho_{k+1}) \subseteq \mathcal{L}_{m}(\rho_{k}, \rho_{k+1})$. On the other hand, by statement $b)$ of this lemma, every $\iota_{m}$-simple node between $\rho_{k}$ and $\rho_{k+1}$ is an absolutely simple node, which means that $\mathcal{L}_{m}(\rho_{k}, \rho_{k+1}) \subseteq \mathcal{L}(\rho_{k}, \rho_{k+1}) \Rightarrow \mathcal{L}_{m}(\rho_{k}, \rho_{k+1}) = \mathcal{L}(\rho_{k}, \rho_{k+1})$.
\end{proof}

\ignore{
\begin{corollary}
Let $\rho_{1}, \rho_{2}$ be two adjacent absolutely super-simple nodes. Then, every $\iota_{m}$-appendage node which is downstream $\rho_{1}$ and upstream $\rho_{2}$ is an absolutely appendage node.
\end{corollary}

\begin{proof}
Consider an $\iota_{m}$-appendage node $\rho$ which is downstream $\rho_{1}$ and upstream $\rho_{2}$. As $\rho_{1}$ is an absolutely super-simple node, $\rho$ belongs to $\mathcal{G}_{j}$, for all $j \in \mathbb{N} \cap [1, n]$. By statement $b)$ of lemma \ref{basic_properties_super_simple_subnetwork}, if there is $j \in \mathbb{N} \cap [1, n]$ such that $\rho$ is an $\iota_{j}$-simple node, then $\rho$ would be an absolutely simple node, which is a contradiction. Therefore, for all $j \in \mathbb{N} \cap [1, n]$, $\rho$ must be $\iota_{j}$-appendage, i.e., $\rho$ is absolutely appendage. 
\end{proof}
}

To verify that the absolutely super-simple structural subnetworks are well defined, we must partition the appendage subnetwork $\mathcal{A}_{\mathcal{G}}$ in a similar way that we have done to $\mathcal{A}_{\mathcal{G}_{m}}$. Therefore, we partition $\mathcal{A}_{\mathcal{G}}$ in $(p + q + r)$ $\mathcal{A}_{\mathcal{G}}$-path components $\mathcal{A}_{1}, \ldots, \mathcal{A}_{p}, \mathcal{B}_{1}, \ldots, \mathcal{B}_{q}, \mathcal{C}_{1}, \ldots, \mathcal{C}_{r}$, where
\begin{enumerate}[(1)]
\item $\mathcal{A}_{\mathcal{G}}$-path components $\mathcal{A}_{i}$ satisfy the following condition: for all $m = 1, \ldots, n$, for every $\iota_{m}o$-simple path $S_{m}$, nodes in $\mathcal{A}_{i}$ are not $CS_{m}$-path equivalent to any node in $CS_{m} \setminus \mathcal{A}_{i}$.
\item $\mathcal{A}_{\mathcal{G}}$-path components $\mathcal{B}_{i}$ satisfy the following condition: for all $m = 1, \ldots, n$, there is an $\iota_{m}o$-simple path $S_{m}$ such that nodes in $\mathcal{B}_{i}$ are $CS_{m}$-path equivalent to an absolutely simple node in $CS_{m} \setminus \mathcal{B}_{i}$ which belongs to an absolutely super-simple subnetwork $\mathcal{L}(\rho_{k}, \rho_{k+1})$, where $\rho_{k}, \rho_{k+1}$ are adjacent absolutely super-simple nodes.
\item $\mathcal{A}_{\mathcal{G}}$-path components $\mathcal{C}_{i}$ do not satisfy neither of the conditions (a) and (b) above.
\end{enumerate}
Again, we have
\begin{equation}
    \mathcal{A}_{\mathcal{G}} = \left( \mathcal{A}_{1} \dot{\cup} \cdots \dot{\cup} \mathcal{A}_{p}\right) \dot{\cup} \left( \mathcal{B}_{1} \dot{\cup} \cdots \dot{\cup} \mathcal{B}_{q}\right) \dot{\cup} \left( \mathcal{C}_{1} \dot{\cup} \cdots \dot{\cup} \mathcal{C}_{r}\right)
\end{equation}

Figure \ref{partition_A_G_path_components} exemplifies the three types of $\mathcal{A}_{\mathcal{G}}$-path components in a core network $\mathcal{G}$ with multiple input nodes.

\begin{figure}[!ht]
\centering
\includegraphics[width=\linewidth,trim=0cm 1.5cm 0cm 0cm,clip=true]%
{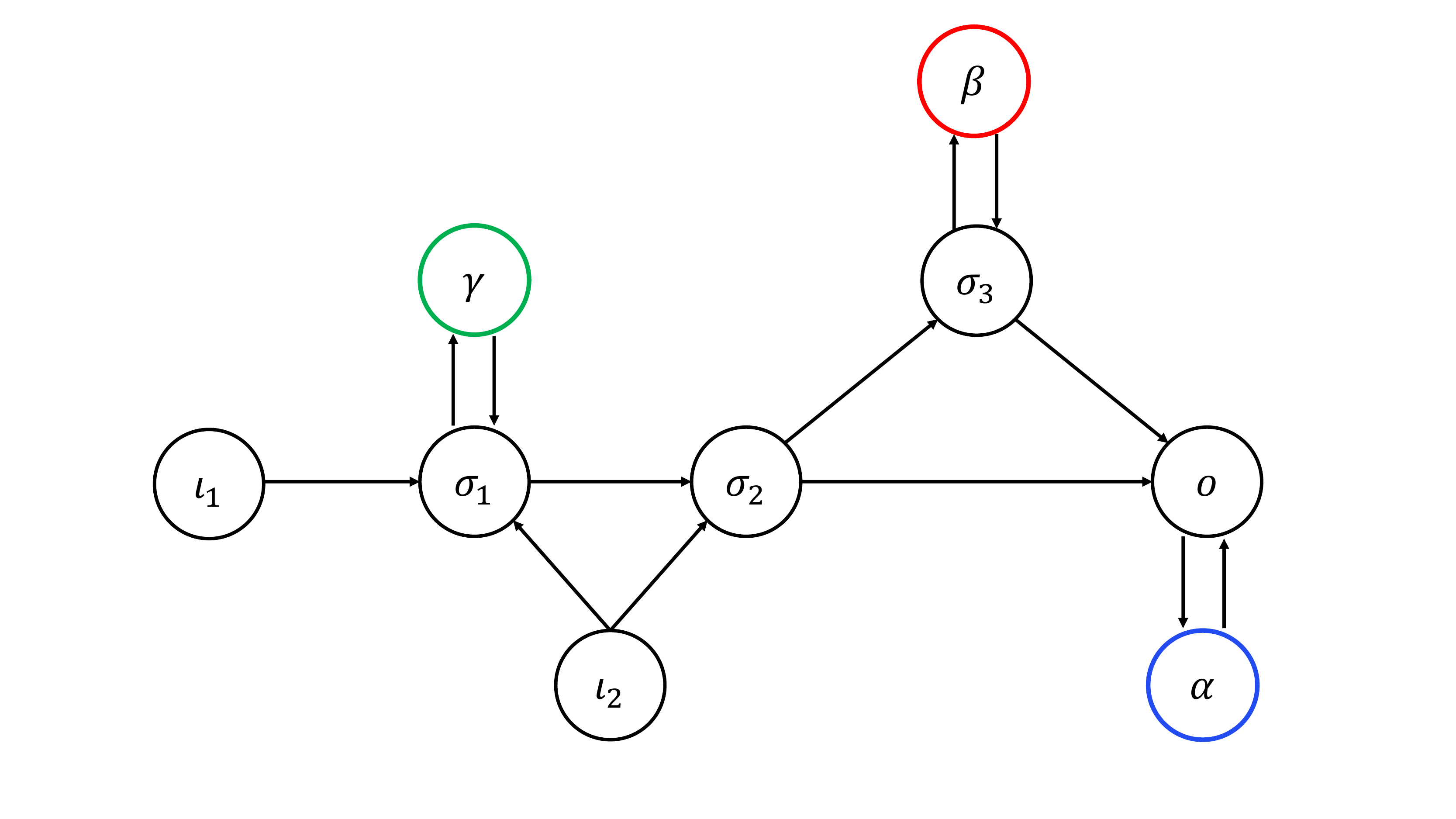}
\renewcommand{\figurename}{Figure}
\caption{A core network with input nodes $\iota_{1}$ and $\iota_{2}$ and output node $o$. Nodes $\sigma_{1}$ and $\sigma_{3}$ are absolutely simple while $\sigma_{2}$ and $o$ are absolutely super-simple. The absolutely super-simple subnetwork $\mathcal{L}(\sigma_{2}, o)$ is composed by nodes $\sigma_{2}, \sigma_{3}$ and $o$. The appendage subnetwork $\mathcal{A}_{\mathcal{G}}$ is composed by $\alpha, \beta$ and $\gamma$, and each of these nodes corresponds to a distinct $\mathcal{A}_{\mathcal{G}}$-path component. Following the nomenclature proposed above, node $\alpha$ (in blue) corresponds to an $\mathcal{A}_{\mathcal{G}}$-path component $\mathcal{A}$, $\beta$ (in red) corresponds to an $\mathcal{A}_{\mathcal{G}}$-path component $\mathcal{B}$ and $\gamma$ (in green) corresponds to an $\mathcal{A}_{\mathcal{G}}$-path component $\mathcal{C}$. According to definition \ref{absolutely_super_simple_structural_subnetwork}, the subnetwork composed by nodes $\sigma_{2}, \sigma_{3}, \beta$ and $o$ is an absolutely super-simple structural subnetwork $\mathcal{L}'(\sigma_{2}, o)$.}
\label{partition_A_G_path_components}
\end{figure}

We now generalise Lemma \ref{A_G_m_path_component_and_i_m_structural_subnetwork} to the $\mathcal{A}_{\mathcal{G}}$-path components $\mathcal{B}_{i}$.

\begin{lemma} \label{unique_correspondence_between_B_i_and_absolutely_super_simple_subnetwork}
Consider the $\mathcal{A}_{\mathcal{G}}$-path component $\mathcal{B}_{i}$ and suppose there is an $\iota_{m}o$-simple path $S_{m}$ such that nodes in $\mathcal{B}_{i}$ are $CS_{m}$-path equivalent to an absolutely simple node $\rho$ in $CS_{m} \setminus \mathcal{B}_{i}$ which belongs to an absolutely super-simple subnetwork $\mathcal{L}(\rho_{k}, \rho_{k+1})$. Then the following statements are valid:
\begin{enumerate}[(a)]
\item Nodes in $CS_{m} \setminus \mathcal{B}_{i}$ that are $CS_{m}$-path equivalent to $\mathcal{B}_{i}$ which are not absolutely appendage, including $\rho$, are absolutely simple nodes contained in $\mathcal{L}(\rho_{k}, \rho_{k+1})$. Furthermore, these nodes are not absolutely super-simple.
\item For every $j=1,\ldots,n$, there is an $\iota_{j}o$-simple path $S_{j}$ such that nodes in $\mathcal{B}_{i}$ are $CS_{j}$-path equivalent to $\rho$ in $CS_{j} \setminus \mathcal{B}_{i}$.
\item Suppose there is another $\iota_{m}o$-simple path $\tilde{S}_{m}$ such that nodes in $\mathcal{B}_{i}$ are $C\tilde{S}_{m}$-path equivalent to nodes in $C\tilde{S}_{m} \setminus \mathcal{B}_{i}$. Then, there is an absolutely simple node $\tau$ between $\rho_{k}$ and $\rho_{k+1}$ which is $C\tilde{S}_{m}$-path equivalent to $\mathcal{B}_{i}$.
\end{enumerate}
\end{lemma}

\begin{proof}
$(a)$ By hypothesis, $\rho$ is an absolutely simple node which belongs to $\mathcal{L}(\rho_{k}, \rho_{k+1})$. Moreover, $\rho$ cannot be an absolutely super-simple node, as $S_{m}$ passes by all absolutely super-simple nodes and $\rho$ belongs to $CS_{m}$. Consider so the $\iota_{m}o$-simple path $S_{m}^{*}$
$$\iota_{m} \rightarrow \cdots \rightarrow \rho_{k} \rightarrow \cdots \rightarrow \rho \rightarrow \cdots \rightarrow \rho_{k+1}  \rightarrow \cdots \rightarrow  o$$
Take now a node $\tau$ of $CS_{m} \setminus \mathcal{B}_{i}$ such that nodes in $\mathcal{B}_{i}$ are $CS_{m}$-path equivalent to $\tau$ and suppose that $\tau$ is not absolutely appendage. As $\tau$ is downstream nodes in $\mathcal{B}_{i}$, then we conclude that $\tau$ is downstream $\iota_{j}$, for every $j=1,\ldots,n$. That means that there is $j \in\{1,\ldots,n\}$ such that $\tau$ is $\iota_{j}$-simple. Consider the $\iota_{j}o$-simple path $S_{j}$ that passes by $\tau$. As $\rho_{k}$ and $\rho_{k+1}$ are absolutely super-simple nodes, $S_{j}$ must pass by $\rho_{k}$ and $\rho_{k+1}$. We shall see that $\tau$ is between $\rho_{k}$ and $\rho_{k+1}$. In fact, suppose that $S_{j}$ passes by these nodes in the following order
$$
\iota_{j} \rightarrow \cdots \rightarrow \tau \rightarrow \cdots \rightarrow \rho_{k} \rightarrow \cdots \rightarrow \rho_{k+1} \rightarrow \cdots \rightarrow o
$$
Then, as $\tau$ and $\rho$ belong to the same $CS_{m}$-path component and, as $\rho_{k}$ and $\rho_{k+1}$ do not belong to $CS_{m}$ (because $\rho_{k}$ and $\rho_{k+1}$ are absolutely super-simple nodes), then there is a path from $\tau \rightarrow \rho$ which does not pass by $\rho_{k}$ or $\rho_{k+1}$. Taking this path together with $S_{m}^{*}$ and $S_{j}$, we can obtain an $\iota_{j}o$-simple path that does not pass by $\rho_{k}$, which is a contradiction. On the other hand, if $S_{j}$ is of the form
$$
\iota_{j} \rightarrow \cdots \rightarrow \rho_{k} \rightarrow \cdots \rightarrow \rho_{k+1}  \rightarrow \cdots \rightarrow \tau \rightarrow \cdots \rightarrow o
$$
then, in analogous manner, we can obtain an $\iota_{j}o$-simple path that does not pass by $\rho_{k+1}$, which is also a contradiction. Therefore, $\tau$ must be between $\rho_{k}$ and $\rho_{k + 1}$, and consequently, by lemma \ref{basic_properties_super_simple_subnetwork}, $\tau$ is an absolutely simple node of $\mathcal{L}(\rho_{k}, \rho_{k+1})$. Moreover, as $\tau$ lies in $CS_{m}$, then $\tau$ must not be absolutely super-simple. \\
$(b)$ For every $j=1,\ldots,n$, we can obtain an $\iota_{j}o$-simple path $S_{j}$ such that the path of $S_{j}$ from $\rho_{k}$ to $o$ coincides with the path of $S_{m}$ from $\rho_{k}$ to $o$. Consider now $CS_{j}$. Clearly $\mathcal{B}_{i} \subset CS_{j}$. On the other hand, the absolutely appendage nodes that are $CS_{m}$-path equivalent to nodes in $\mathcal{B}_{i}$ also belong to $CS_{j}$. Consider now the node $\tau$ which is absolutely simple and $CS_{m}$-path equivalent to $\mathcal{B}_{i}$. We need to verify that $\tau$ also belongs to $CS_{j}$. In fact, by statement $a)$, $\tau$ is an absolutely simple node between $\rho_{k}$ and $\rho_{k+1}$, and by the initial hypothesis, $S_{m}$ and $S_{j}$ share the same path between $\rho_{k}$ and $\rho_{k+1}$, meaning that $\tau$ does not lie in $S_{j}$, and therefore $\tau$ belongs to $CS_{j}$. By item $(a)$ above, nodes that belong to the same $CS_{m}$-path component as $\mathcal{B}_{i}$ are absolutely appendage or absolutely simple, and therefore nodes that are $CS_{m}$-path equivalent to $\mathcal{B}_{i}$ are also $CS_{j}$-path equivalent to $\mathcal{B}_{i}$, including, in particular, $\rho$. \\
$(c)$ Suppose there is an $\iota_{m}o$-simple path $\tilde{S}_{m}$ such that nodes in $\mathcal{B}_{i}$ are $C\tilde{S}_{m}$-path equivalent to nodes in $C\tilde{S}_{m} \setminus \mathcal{B}_{i}$. As $\mathcal{B}_{i}$ is an $\mathcal{A}_{\mathcal{G}}$-path component, there is at least one node $\tau \in C\tilde{S}_{m} \setminus \mathcal{B}_{i}$ such that $\tau$ is not absolutely appendage and $\tau$ is $C\tilde{S}_{m}$-path equivalent to nodes in $\mathcal{B}_{i}$. In fact, if that was not the case, we would find an $\mathcal{A}_{\mathcal{G}}$-path component that contains $\mathcal{B}_{i}$ and which is different from $\mathcal{B}_{i}$, which is a contradiction. Therefore, there is $j=1,\ldots,n$ such that $\tau$ is $\iota_{j}$-simple. By item $(b)$ above, there is an $\iota_{j}o$-simple path $S_{j}$ such that $\rho \not\in S_{j}$ and nodes in $\mathcal{B}_{i}$ are $CS_{j}$-path equivalent to $\rho$. Moreover, as $\tau$ is $\iota_{j}$-simple, there is an $\iota_{j}o$-simple path $S_{j}^{*}$ that passes by $\tau$. As $\rho_{k}$ and $\rho_{k+1}$ are super-simple nodes, they are present in $\tilde{S}_{m}$, $S_{j}$ and $S_{j}^{*}$. Suppose that $S_{j}^{*}$ follows the order
$$
\iota_{j} \rightarrow \cdots \rightarrow \tau \rightarrow \cdots \rightarrow \rho_{k} \rightarrow \cdots \rightarrow \rho_{k+1} \rightarrow \cdots \rightarrow o
$$
In this case, we have the following:
\begin{enumerate}[(1)]
\item By definition of  $S_{j}^{*}$, there an $\iota_{j}\tau$-simple path that does not pass by $\rho_{k}$,
\item As $\tau$ is $C\tilde{S}_{m}$-path equivalent to nodes in $\mathcal{B}_{i}$ and $\tilde{S}_{m}$ contains $\rho_{k}$ and $\rho_{k+1}$, for every node from $\mathcal{B}_{i}$, there is a path from $\tau$ to this node which does not pass by $\rho_{k}$ or $\rho_{k+1}$,
\item By a similar argument, for every node in $\mathcal{B}_{i}$ there is a path from this node to $\rho$ which does not pass by $\rho_{k}$ or $\rho_{k+1}$,
\item As $\rho$ is between $\rho_{k}$ and $\rho_{k+1}$, there is a $\rho o$-simple path which does not pass by $\rho_{k}$,
\end{enumerate}
Now it follows from (1)-(4) above that we can obtain an $\iota_{j}o$-simple path that does not pass by $\rho_{k}$, which is a contradiction. On the other hand, if $S_{j}^{*}$ follows the order
$$\iota_{j} \rightarrow \cdots \rightarrow \rho_{k} \rightarrow \cdots \rightarrow \rho_{k+1} \rightarrow \cdots \rightarrow \tau \rightarrow \cdots \rightarrow o,$$
then, in an analogous manner, we can obtain an $\iota_{j}o$-simple path that does not pass by $\rho_{k+1}$, which is also a contradiction. Therefore, $\tau$ must belong to $\mathcal{L}_{j}(\rho_{k}, \rho_{k+1})$, which means that, by Lemma \ref{basic_properties_super_simple_subnetwork}, $\tau$ is an absolutely simple node that is between $\rho_{k}$ and $\rho_{k+1}$.
\end{proof}

Lemma \ref{unique_correspondence_between_B_i_and_absolutely_super_simple_subnetwork} implies that the correspondence between the $\mathcal{A}_{\mathcal{G}}$-path components $\mathcal{B}_{i}$ and the absolutely super-simple subnetworks is unique, and therefore the absolutely super-simple structural subnetworks are well defined (see Definition \ref{absolutely_super_simple_structural_subnetwork}).

\begin{lemma} \label{basic_properties_absolutely_super_simple_structural_subnetwork}
Let $\rho_{k}>\rho_{k+1}$ be two adjacent absolutely super-simple nodes. Then $\mathcal{L}'_{m}(\rho_{k}, \rho_{k+1}) = \mathcal{L}'(\rho_{k}, \rho_{k+1})$, for every $m=1,\ldots,n$.
\end{lemma}

\begin{proof}
By Lemma \ref{basic_properties_super_simple_subnetwork}, we already know that $\mathcal{L}_{m}(\rho_{k}, \rho_{k+1}) = \mathcal{L}(\rho_{k}, \rho_{k+1})$, for every $m=1,\ldots,n$. By definition, we also know that $\mathcal{L}'(\rho_{k}, \rho_{k+1}) = \mathcal{L}(\rho_{k}, \rho_{k+1}) \cup \mathcal{B}$, where $\mathcal{B}$ consists of all absolutely appendage nodes that are $CS_{m}$-path equivalent to nodes in $\mathcal{L}(\rho_{k}, \rho_{k+1})$ for some $\iota_{m}o$-simple path $S_{m}$, for some $m \in \{1,\ldots,n\}$. Consider an $\mathcal{A}_{\mathcal{G}}$-path component $\mathcal{B}_{i} \subset \mathcal{L}'(\rho_{k}, \rho_{k+1})$. By item $(b)$ of Lemma \ref{unique_correspondence_between_B_i_and_absolutely_super_simple_subnetwork}, for every $m=1,\ldots,n$, there is an $\iota_{m}o$-simple path $S_{m}$ such that nodes in $\mathcal{B}_{i}$ are $CS_{m}$-path equivalent to nodes in $\mathcal{L}(\rho_{k}, \rho_{k+1})$. Moreover, as nodes in $\mathcal{B}_{i}$ are downstream from $\iota_{m}$, then all nodes that are $CS_{m}$-path equivalent to nodes in $\mathcal{B}_{i}$ are downstream from $\iota_{m}$, and therefore nodes in $\mathcal{B}_{i}$ are $C_{m}S_{m}$-path equivalent to nodes in $\mathcal{L}_{m}(\rho_{k}, \rho_{k+1})$, for every $m=1,\ldots,n$. As this property is valid for every $\mathcal{A}_{\mathcal{G}}$-path component $\mathcal{B}_{i} \subset \mathcal{L}'(\rho_{k}, \rho_{k+1})$, we conclude that $\mathcal{L}'(\rho_{k}, \rho_{k+1}) \subseteq \mathcal{L}'_{m}(\rho_{k}, \rho_{k+1})$, for every $m=1,\ldots,n$. On the other hand, consider an $\mathcal{A}_{\mathcal{G}_{m}}$-path component $\mathcal{B}_{m, i} \subset \mathcal{L}'_{m}(\rho_{k}, \rho_{k+1})$. As there is a path between nodes in $\mathcal{L}_{m}(\rho_{k}, \rho_{k+1}) = \mathcal{L}(\rho_{k}, \rho_{k+1})$ and nodes in $\mathcal{B}_{m, i}$, then nodes in $\mathcal{B}_{m, i}$ are downstream from $\iota_{j}$, $j = 1, \ldots, n$. Suppose a node $\tau$ of $\mathcal{B}_{m, i}$ is not an absolutely appendage node, i.e., that there is $j \in\{1,\ldots,n\}$ such that $\tau$ is $\iota_{j}$-simple. Consider the $\iota_{j}o$-simple path $S_{j}$ that passes by $\tau$. If $\tau$ is between $\rho_{k}$ and $\rho_{k+1}$, then, by Lemma \ref{basic_properties_super_simple_subnetwork}, $\tau$ is an absolutely simple node, which is contradiction considering $\tau$ is $\iota_{m}$-appendage. Suppose then that $S_{j}$ follows the order
$$
\iota_{j} \rightarrow \cdots \rightarrow \tau \rightarrow \cdots \rightarrow \rho_{k} \rightarrow \cdots \rightarrow \rho_{k+1} \rightarrow \cdots \rightarrow o
$$
In that case, as there is a path between $\tau \in \mathcal{B}_{m,i}$ and nodes in $\mathcal{L}_{m}(\rho_{k}, \rho_{k+1}) = \mathcal{L}(\rho_{k}, \rho_{k+1})$ which does not pass by $\rho_{k}$ and there is a path between $\iota_{j}$ and $\tau$ which does not pass by $\rho_{k}$, then, we can obtain an $\iota_{j}o$-simple path that does not pass by $\rho_{k}$, contradicting the fact that $\rho_{k}$ is absolutely super-simple. By a similar argument, if $S_{j}$ follows the order
$$
\iota_{j} \rightarrow \cdots \rightarrow \rho_{k} \rightarrow \cdots \rightarrow \rho_{k+1} \rightarrow \cdots \rightarrow \tau \rightarrow \cdots \rightarrow o
$$
then we can obtain an $\iota_{j}o$-simple path that does not pass by $\rho_{k+1}$, which is also a contradiction. Therefore, we conclude that every node in $\mathcal{B}_{m, i}$ is absolutely appendage, and consequently $\mathcal{B}_{m, i}$ is an $\mathcal{A}_{\mathcal{G}}$-path component. We know that there is an $\iota_{m}o$-simple path $S_{m}$ such that $\mathcal{B}_{m, i}$ is $C_{m}S_{m}$-path equivalent to nodes in $\mathcal{L}_{m}(\rho_{k}, \rho_{k+1}) = \mathcal{L}(\rho_{k}, \rho_{k+1})$, and, consequently, as $C_{m}S_{m} \subset CS_{m}$, nodes in $\mathcal{B}_{m, i}$ are $CS_{m}$-path equivalent to nodes in $\mathcal{L}_{m}(\rho_{k}, \rho_{k+1}) = \mathcal{L}(\rho_{k}, \rho_{k+1})$, i.e., $\mathcal{B}_{m, i} \subset \mathcal{L}'(\rho_{k}, \rho_{k+1})$. As this is valid for every $\mathcal{A}_{\mathcal{G}_{m}}$-path component of $\mathcal{L}'_{m}(\rho_{k}, \rho_{k+1})$, we conclude that, for every $m=1,\ldots,n$, one has
$\mathcal{L}'_{m}(\rho_{k}, \rho_{k+1}) \subseteq \mathcal{L}'(\rho_{k}, \rho_{k+1}) \Rightarrow \mathcal{L}'_{m}(\rho_{k}, \rho_{k+1}) = \mathcal{L}'(\rho_{k}, \rho_{k+1})
$.
\end{proof}

Wang \etal \cite{wang20} proved that in networks with only one input node, every irreducible structural homeostasis block corresponds to the homeostasis determinant of a super-simple structural subnetwork $\mathcal{L}'(\tau, \sigma)$, where $\tau$ is the input node and $\sigma$ is the output node of this subnetwork. On the other hand, the homeostasis determinant of each super-simple structural subnetwork $\mathcal{L}'(\tau, \sigma)$ uniquely corresponds to an irreducible structural homeostasis block.

\begin{theorem} \label{structural_homeostasis_G}
Consider the core network $\mathcal{G}$ with multiple input nodes. If there is an irreducible structural homeostasis block $B_{s}$ such that $\det(B_{s})$ is an irreducible factor of $\det\!\big(\langle H \rangle\big)$, then $\mathcal{G}$ has adjacent absolutely super-simple nodes $\rho_{k}$ and $\rho_{k+1}$ such that $\det(B_{s}) = \det\!\big( H(\mathcal{L}'(\rho_{k}, \rho_{k+1}))\big)$
\end{theorem}

\begin{proof}
By the results above, we know that $\det(B_{s})$ is an irreducible factor of $\det \langle H \rangle$ if and only if $\det(B_{s})$ is an irreducible factor of each homeostasis determinant $\det(H_{\iota_{m}}^{c})$. That means, by the results of \cite[Thm 6.11]{wang20}, that for every $m=1,\ldots,n$, there are $\iota_{m}$-super-simple nodes $\rho_{k_{m}}, \rho_{k_{m}+1}$ such that $\det(B_{s}) = \det\!\big(H(\mathcal{L}_{m}'(\rho_{k_{m}}, \rho_{k_{m}+1}))\big)$. In particular, these implies that all the subnetworks $\mathcal{L}'(\rho_{k_{m}}, \rho_{k_{m}+1})$ must share the same input and output nodes, i.e., there are absolutely super-simple nodes $\rho_{k}, \rho_{k+1}$ such that $\det(B_{s}) = \det\!\big(H(\mathcal{L}_{m}'(\rho_{k}, \rho_{k+1}))\big)$, for every $m = 1, \ldots, n$. By Lemma \ref{basic_properties_absolutely_super_simple_structural_subnetwork}, $\det(B_{s}) = \det\!\big(H(\mathcal{L}'(\rho_{k}, \rho_{k+1}))\big)$.
\end{proof}

\begin{corollary} \label{structural_homeostasis_G_reverse}
Consider a core network $\mathcal{G}$ with multiple input nodes. If $\mathcal{G}$ have absolutely super-simple nodes other than the output node, then the homeostasis matrix of each absolutely super-simple structural subnetwork corresponds to an irreducible structural homeostasis block.
\end{corollary}

\begin{proof}
Consider the adjacent absolutely super-simple nodes $\rho_{k}, \rho_{k+1}$ in $\mathcal{G}$. By Lemma \ref{basic_properties_absolutely_super_simple_structural_subnetwork}, for every $m=1,\ldots,n$, we have $\mathcal{L}'_{m}(\rho_{k}, \rho_{k+1}) = \mathcal{L}'(\rho_{k}, \rho_{k+1})$. As proved in \cite[Thm 7.2]{wang20}, this means that the homeostasis matrix of $\mathcal{L}'(\rho_{k}, \rho_{k+1})$ is an irreducible structural homeostasis block of each subnetwork $\mathcal{G}_{m}$, and therefore the homeostasis matrix of $\mathcal{L}'(\rho_{k}, \rho_{k+1})$ is an irreducible structural homeostasis block of $\mathcal{G}$.
\end{proof}

\subsubsection{Input Counterweight Homeostasis}

Recall that for each core network $\mathcal{G}$, the determinant of the input counterweight homeostasis block $C$ is unique up to signal and its explicit formula is given by \eqref{term_input_counter_weight_homeostasis}.

We have already proved that $\det(C)$ is an irreducible factor of $\det\!\big(\langle H \rangle\big)$. It remains to show that $C$ is associated with the input counterweight subnetwork $\mathcal{W}_{\mathcal{G}}$ (see Definition \ref{definition_input_counterweight_subnetwork}). For this purpose, we first verify that in certain sense $\mathcal{G}$ can be divided in subnetworks associated to appendage homeostasis, absolutely super-simple structural subnetworks (structural homeostasis) and $\mathcal{W}_{\mathcal{G}}$. 

\begin{lemma} \label{partition_G_W_G_appendage_subnetworks_structuralk_subnetwork} Consider a core network $\mathcal{G}$ with multiple input nodes, with absolutely super-simple nodes $\rho_{1} > \cdots > \rho_{s} > \rho_{s+1} = o$, and its associated input counterweight subnetwork $\mathcal{W}_{\mathcal{G}}$. Then, the following statements are valid
\begin{enumerate}[(a)]
    \item $\mathcal{W}_{\mathcal{G}}$ does not share nodes with any of the subnetworks $\mathcal{A}_{i}$, where $\mathcal{A}_{i}$ is an $\mathcal{A}_{\mathcal{G}}$- path component that satisfy the following condition: for all $m = 1, \cdots, n$, for every $\iota_{m}o$-simple path $S_{m}$, nodes in $\mathcal{A}_{i}$ are not $CS_{m}$-path equivalent to any node in $CS_{m} \setminus \mathcal{A}_{i}$.
    \item For $k \neq 1$, $\mathcal{W}_{\mathcal{G}}$ does not share nodes with $\mathcal{L}'(\rho_{k}, \rho_{k+1})$. Moreover, the only common node between $\mathcal{W}_{\mathcal{G}}$ and $\mathcal{L}'(\rho_{1}, \rho_{2})$ is $\rho_{1}$.
    \item If a node $\sigma$ in $\mathcal{G}$ is such that it does not belong to any subnetwork $\mathcal{A}_{i}$ defined in item $(a)$, neither to any absolutely super-simple structural subnetwork, then $\sigma$ belongs to $\mathcal{W}_{\mathcal{G}}$.
\end{enumerate}
\end{lemma}

\begin{proof}
$(a)$ The proof is straightforward, as none of the nodes in $\mathcal{W}_{\mathcal{G}}$ are absolutely appendage nodes that satisfy the condition of item $(a)$ above. \\
$(b)$ Consider the subnetwork $\mathcal{S}$ composed by the union of all absolutely structural super-simple subnetworks of $\mathcal{G}$:
$$
\mathcal{S} = \mathcal{L}'(\rho_{1}, \rho_{2}) \cup \mathcal{L}'(\rho_{2}, \rho_{3}) \cup \cdots \cup \mathcal{L}'(\rho_{s}, o)
$$
Therefore, with the exception of $\rho_{1}$, by Lemmas \ref{basic_properties_super_simple_subnetwork} and \ref{unique_correspondence_between_B_i_and_absolutely_super_simple_subnetwork}, every node $\sigma$ in $\mathcal{S}$ satisfies one of the two conditions: $\sigma$ is absolutely simple and for every $m=1,\ldots,n$ there is a simple path $S_{m}$ that passes by $\iota_{m}$, $\rho_{1}$ and $\sigma$ in this order; or $\sigma$ is an absolutely appendage node which belongs to an $\mathcal{A}_{\mathcal{G}}$-path component $\mathcal{B}_{i}$ that $CS_{m}$-path equivalent to an absolutely simple node $\rho$ in $CS_{m} \setminus \mathcal{B}_{i}$ which belongs to an absolutely super-simple subnetwork $\mathcal{L}(\rho_{k}, \rho_{k+1})$. In both cases, by definition, $\sigma$ does not belong to $\mathcal{W}_{\mathcal{G}}$. Whereas $\rho_{1}$ belongs to both $\mathcal{W}_{\mathcal{G}}$ and $\mathcal{S}$. \\
$(c)$ If $\sigma$ is absolutely simple, then for every $\iota_{m}o$-simple path $S_{m}$ that passes by $\sigma$, $\sigma$ must be upstream from $\rho_{1}$ (if this does not happen, then, by Definition \ref{definition_adjacent_absolutely_super_simple_}, $\sigma$ belongs to one absolutely structural super-simple subnetwork, which is a contradiction). This means that there is an $\iota_{m}o$-simple path $S_{m}$ which follows the order $\iota_{m} \rightarrow \cdots \rightarrow \sigma \rightarrow \cdots \rightarrow \rho_{1}$, and therefore $\sigma$ belongs to $\mathcal{W}_{\mathcal{G}}$. On the other hand, if $\sigma$ is absolutely appendage, then, by the partition of $\mathcal{A}_{\mathcal{G}}$ explained above, $\sigma$ must belong to an $\mathcal{A}_{\mathcal{G}}$-path component $\mathcal{C}_{i}$ for which there is $CS_{m}$-path equivalent to nodes that are not absolutely appendage and that are not between two absolutely super-simple nodes, for some $\iota_{m}o$-simple path $S_{m}$, for some $m \in\{1,\ldots,n\}$, i.e., in this case $\sigma$ also belongs to $\mathcal{W}_{\mathcal{G}}$. Finally, if $\sigma$ is not absolutely simple nor absolutely appendage, then $\sigma$ belongs to $\mathcal{W}_{\mathcal{G}}$.
\end{proof}

Lemma \ref{partition_G_W_G_appendage_subnetworks_structuralk_subnetwork} implies that the network $\mathcal{G}$ is basically the union between $\mathcal{W}_{\mathcal{G}}$, the subnetworks associated to appendage homeostasis and the absolutely structural super-simple subnetworks of $\mathcal{G}$. Moreover, these different subnetworks do not share common paths. It is interesting to note that $\mathcal{W}_{\mathcal{G}}$ contains all the vestigial subnetworks, as these subnetworks are composed by nodes which are not either absolutely simple nor absolutely appendage.

\begin{lemma} \label{properties_input_counterweight_subnetwork} Consider a core network $\mathcal{G}$ with multiple input nodes and its associated input counterweight subnetwork $\mathcal{W}_{\mathcal{G}}$. Then, the following statements are valid
\begin{enumerate} [(a)]
\item $\mathcal{W}_{\mathcal{G}}$ is a core network with input nodes $\iota_{1}, \ldots, \iota_{n}$ and output node $\rho_{1}$.
\item $\mathcal{W}_{\mathcal{G}}$ does not support either appendage nor structural homeostasis.
\item $\det\!\big(\langle H \rangle (\mathcal{W}_{\mathcal{G}})\big)$ is a factor of $\det\!\big(\langle H \rangle\big)$.
\end{enumerate}
\end{lemma}

\begin{proof}
$(a)$ As $\mathcal{W}_{\mathcal{G}}$ is a subnetwork of $\mathcal{G}$, every node in $\mathcal{W}_{\mathcal{G}}$ is downstream from at least one of the input nodes. We must now verify that every node in $\mathcal{W}_{\mathcal{G}}$ is upstream from $\rho_{1}$. In fact, this is true for the input nodes and for $\rho_{1}$. This is also true for every node $\tau$ for which there is $m \in\{1,\ldots,n\}$ such that there is an $\iota_{m}o$-simple path that passes at $\iota_{m}$, $\tau$ and $\rho_{1}$ in that order. On the other hand, take a node $\sigma$ which cannot be classified as absolutely appendage nor absolutely simple. There are two possibilities for $\sigma$: (i) $\sigma$ is not downstream every input node, or (ii) $\sigma$ is downstream from every input node, but there are $m, j \in \{1,\ldots,n\}$ such that $\sigma$ is $\iota_m$-simple and $\iota_j$-appendage. In the first case, $\sigma$ must be downstream at least one input node $\iota_{m}$. If $\sigma$ is $\iota_{m}$-simple, then if by the $\iota_{m}o$-simple path that passes by $\sigma$, $\sigma$ is downstream $\rho_{1}$, then $\sigma$ would be downstream every input node, which is a contradiction, and therefore $\sigma$ must be upstream $\rho_{1}$ by this $\iota_{m}o$-simple path. If $\sigma$ is $\iota_{m}$-appendage, then $\sigma$ must be $\mathcal{G}_{m}$-path connected to some $\iota_{m}$-simple node which, by a similar argument, must be upstream $\rho_{1}$, and, again, $\sigma$ is upstream $\rho_{1}$. Considering now the second case ($\sigma$ is downstream every input node, but there are $m, j \in\{1,\ldots,n\}$ such that  $\sigma$ is $\iota_m$-simple and $\iota_j$-appendage), then $\sigma$ must be upstream $\rho_{1}$ by the $\iota_{m}o$-simple path that passes by $\sigma$, as if this does not happen, $\sigma$ would be absolutely simple. Finally, take an absolutely appendage node $\sigma$ in $\mathcal{C}$ described in Definition \ref{definition_input_counterweight_subnetwork}. Then, there is a path between $\sigma$ and a node $\tau$ in $\mathcal{W}_{\mathcal{G}}$ such that there is an $\iota_{m}o$-simple path $S_{m}$ passing by $\tau$ and $\tau$ is not between two absolutely super-simple nodes. If by $S_{m}$, $\tau$ is downstream $\rho_{1}$, then that means that $\tau$ must be between two absolutely super-simple nodes, which is a contradiction. Therefore, $S_{m}$ must pass by $\iota_{m}$, $\tau$ and $\rho_{1}$ in that order, and, as there is a path from $\sigma$ to $\tau$, then $\sigma$ is also upstream $\rho_{1}$. As all nodes in $\mathcal{W}_{\mathcal{G}}$, we can see this subnetwork as a core network between with input nodes $\iota_{1}, \ldots, \iota_{n}$ and output node $\rho_{1}$. \\
$(b)$ The only absolutely appendage nodes in $\mathcal{W}_{\mathcal{G}}$ are the absolutely appendage nodes in $\mathcal{C}$. By Lemma \ref{unique_correspondence_between_B_i_and_absolutely_super_simple_subnetwork}, for every $\iota_{m}o$-simple path $S_{m}$ for which an $\mathcal{A}_{\mathcal{G}}$-path component $\mathcal{C}_{i} \subseteq \mathcal{C}$ is $CS_{m}$-path equivalent to nodes in $CS_{m} \setminus \mathcal{C}_{i}$, then not absolutely appendage nodes $CS_{m}$-path equivalent to $\mathcal{C}_{i}$ are not between two absolutely super-simple nodes, meaning that these nodes are also present in $\mathcal{W}_{\mathcal{G}}$. Therefore, all $\mathcal{A}_{\mathcal{W}_{\mathcal{G}}}$-path components do not follow the necessary conditions to present appendage homeostasis, i.e., this kind of homeostasis is not supported by $\mathcal{W}_{\mathcal{G}}$. On the other hand, as the only absolutely super-simple node in $\mathcal{W}_{\mathcal{G}}$ is $\rho_{1}$, then $\mathcal{W}_{\mathcal{G}}$ does not support structural homeostasis neither. \\
$(c)$ We verify this looking at the factorisation of $\det\!\big(\langle H \rangle\big)$. Consider permutation matrices $P$ and $Q$ such that $P \langle H \rangle Q$ is the Frobenius-K\"onig normal form of the matrix $\langle H \rangle$. Recall that each row of $P \langle H \rangle Q$ represents the partial derivatives of a function $f_{j}$ with respect to all other nodes of $\mathcal{G}$, and each column of $P \langle H \rangle Q$ represents the partial derivatives of all the functions that describe the dynamics of node with respect to same node $j$ (with the exception of the column composed by zeros and by the terms $f_{\iota_{m}, \mathcal{I}}$). Consider the $\mathcal{A}_{\mathcal{G}}$-path components $\mathcal{A}_{i}$ which satisfy the following condition: for all $m = 1, \ldots, n$, for every $\iota_{m}o$-simple path $S_{m}$, nodes in $\mathcal{A}_{i}$ are not $CS_{m}$-path equivalent to any node in $CS_{m} \setminus \mathcal{A}_{i}$, and the absolutely super-simple structural subnetworks $\mathcal{L}'(\rho_{1}, \rho_{2}), \mathcal{L}'(\rho_{2}, \rho_{3}), \ldots, \mathcal{L}'(\rho_{p}, o)$. We have already verified that the Jacobian $J_{\mathcal{A}_{i}}$ of each $\mathcal{A}_{\mathcal{G}}$-path components $\mathcal{A}_{i}$ and the homeostasis matrix $H(\mathcal{L}'(\rho_{k}, \rho_{k+1}))$ of each absolutely super-simple structural subnetwork appear as independent irreducible blocks of $P \langle H \rangle Q$. By equation \eqref{factoring_det_weighted_H} and by our results on the characterization of appendage and structural blocks, we get
\begin{equation} \label{product_determinant_blocks_normal_form}
\begin{split}
    \det\!\big( \langle H \rangle\big) = & \det(J_{\mathcal{A}_{1}}) \cdots \det(J_{\mathcal{A}_{r}}) \cdot \\ 
    & \det\!\big( H(\mathcal{L}'(\rho_{1}, \rho_{2}))\big) \cdots \det\!\big( H(\mathcal{L}'(\rho_{p}, o))\big) \cdot \det(C)
\end{split}
\end{equation}
We need to determine which rows and columns of $\langle H \rangle$ appear in $C$. Recall that the matrices that appear in the right-handed side of equation \eqref{product_determinant_blocks_normal_form} are the blocks that appear in the normal form of $\langle H \rangle$. Therefore, $C$ consists of rows and columns that are not present in such matrices. This means, in particular, that $C$ consists of:
\begin{enumerate}[(1)]
\item rows that contain the partial derivatives of the functions $f_{\iota_{m}}$, for every $m=1,\ldots,n$,
\item rows that contain the partial derivatives of $f_{\rho_{1}}$,
\item rows that contain the partial derivatives of the functions $f_{\sigma}$, where $\sigma$ represents nodes in $\mathcal{G}$ which do not belong to any subnetwork $\mathcal{A}_{i}$ neither to any absolutely super-simple structural subnetwork,
\item columns composed by zeros and by the terms $f_{\iota_{m}, \mathcal{I}}$,
\item columns containing the partial derivatives with respect to nodes $\iota_{m}$, for every $m=1,\ldots,n$ and nodes $\sigma$ (described in item (3)) above,
\item columns containing the partial derivatives with respect to $\rho_{1}$ is not a columns of $C$, as this column is present in $H(\mathcal{L}'(\rho_{1}, \rho_{2}))$.
\end{enumerate}
As shown in item $(a)$ above, $\mathcal{W}_{\mathcal{G}}$ is a core network with input nodes $\iota_{1}, \cdots, \iota_{n}$ and output node $\rho_{1}$, then, by Lemma \ref{partition_G_W_G_appendage_subnetworks_structuralk_subnetwork}, we conclude that $C$ is equivalent, up to permutation of rows and/or columns, to the matrix $\langle H \rangle (\mathcal{W}_{\mathcal{G}})$, which means that $\det\!\big(\langle H \rangle (\mathcal{W}_{\mathcal{G}})\big)$ is a factor of $\det\!\big( \langle H \rangle\big)$.
\end{proof}

Now we can finally characterize input counterweight homeostasis in terms of network topology.

\begin{theorem} \label{fully_characterization_input_counterweight_homeostasis}
Consider a core network $\mathcal{G}$ with multiple input nodes and its associated input counterweight subnetwork $\mathcal{W}_{\mathcal{G}}$. Then, $\langle H \rangle (\mathcal{W}_{\mathcal{G}})$ is, up to permutation of rows or columns, the irreducible input counterweight homeostasis block $C$ of $\langle H \rangle$.
\end{theorem}

\begin{proof}
By Lemma \ref{properties_input_counterweight_subnetwork}, $\det\!\big( \langle H \rangle (\mathcal{W}_{\mathcal{G}})\big)$ is a factor of $\det\!\big(\langle H \rangle\big)$. Moreover, as $\mathcal{W}_{\mathcal{G}}$ does not support neither appendage or structural homeostasis, $\det\!\big( \langle H \rangle (\mathcal{W}_{\mathcal{G}})\big)$ is an irreducible homogeneous polynomial of degree $1$ on variables $f_{\iota_{1}, \mathcal{I}}, \ldots, f_{\iota_{n}, \mathcal{I}}$. By Frobenius-K\"onig theory, we conclude that $\langle H \rangle (\mathcal{W}_{\mathcal{G}})$ must be, up to permutation of rows or columns, the irreducible input counterweight homeostasis block $C$ of $\langle H \rangle$.
\end{proof}

\section{Analysis of \textit{Escherichia coli} Chemotaxis}
\label{sec:e_coli}

In order to apply the theory developed in this paper to the model for the \textit{Escherichia coli} chemotaxis of \cite{edgington18}, we first must rewrite the system \eqref{original_e_coli} in the standard form \eqref{admissible_systems_ODE_multiple_input_nodes}.
We make the following correspondence between variables:
$m \leftrightarrow x_{\iota_1}$, $a_p \leftrightarrow x_{\iota_2}$, $b_p \leftrightarrow x_{\sigma}$, $y_p \leftrightarrow x_{o}$, $L \leftrightarrow \mathcal{I}$.
This gives the following system of ODEs
\begin{equation} \label{e_coli_network_notation}
\begin{aligned}
\dot{x}_{\iota_1} & = \gamma_{R}  \, (1 - \phi(x_{\iota_1}, \mathcal{I})) - \gamma_{B} \, x_{\sigma}^{2} \, \phi(x_{\iota_1}, \mathcal{I}) \\
\dot{x}_{\iota_{2}} & = \phi(x_{\iota_1}, \mathcal{I}) \, k_{1} \, (1 - x_{\iota_2}) - k_{2} \, (1 - x_{o}) \, x_{\iota_2} - k_{3} \, (1 - x_{\sigma}) \, x_{\iota_2} \\
\dot{x}_{\sigma} & = \alpha_{2} \, k_{3} \, (1 - x_{\sigma}) \, x_{\iota_2} - k_{5} \, x_{\sigma} \\
\dot{x}_{o} & = \alpha_{1} \, k_{2} \, (1 - x_{o}) \, x_{\iota_2} - k_{4} \, x_{o}
\end{aligned}
\end{equation}
with the function $\phi$ given by \eqref{definition_phi} and the input parameter $\mathcal{I}$.
Note that the ODE system \eqref{e_coli_network_notation} is an admissible system for the abstract input-output network shown in Figure \ref{e_coli_network}(a). In fact, the general form of an admissible system for this network is
\begin{equation} \label{e_coli_general_admissible}
\begin{aligned}
\dot{x}_{\iota_1} & = f_{\iota_1}(x_{\iota_1},x_{\sigma},\mathcal{I}) \\
\dot{x}_{\iota_2} & = f_{\iota_2}(x_{\iota_1}, x_{\iota_2},x_{\sigma},x_{o},\mathcal{I}) \\
\dot{x}_{\sigma} & = f_{\sigma}(x_{\iota_2},x_{\sigma}) \\
\dot{x}_{o} & = f_{o}(x_{\iota_2},x_{o})
\end{aligned}
\end{equation}

\ignore{
\begin{figure}[!ht]
\centering
\includegraphics[scale = 0.33]{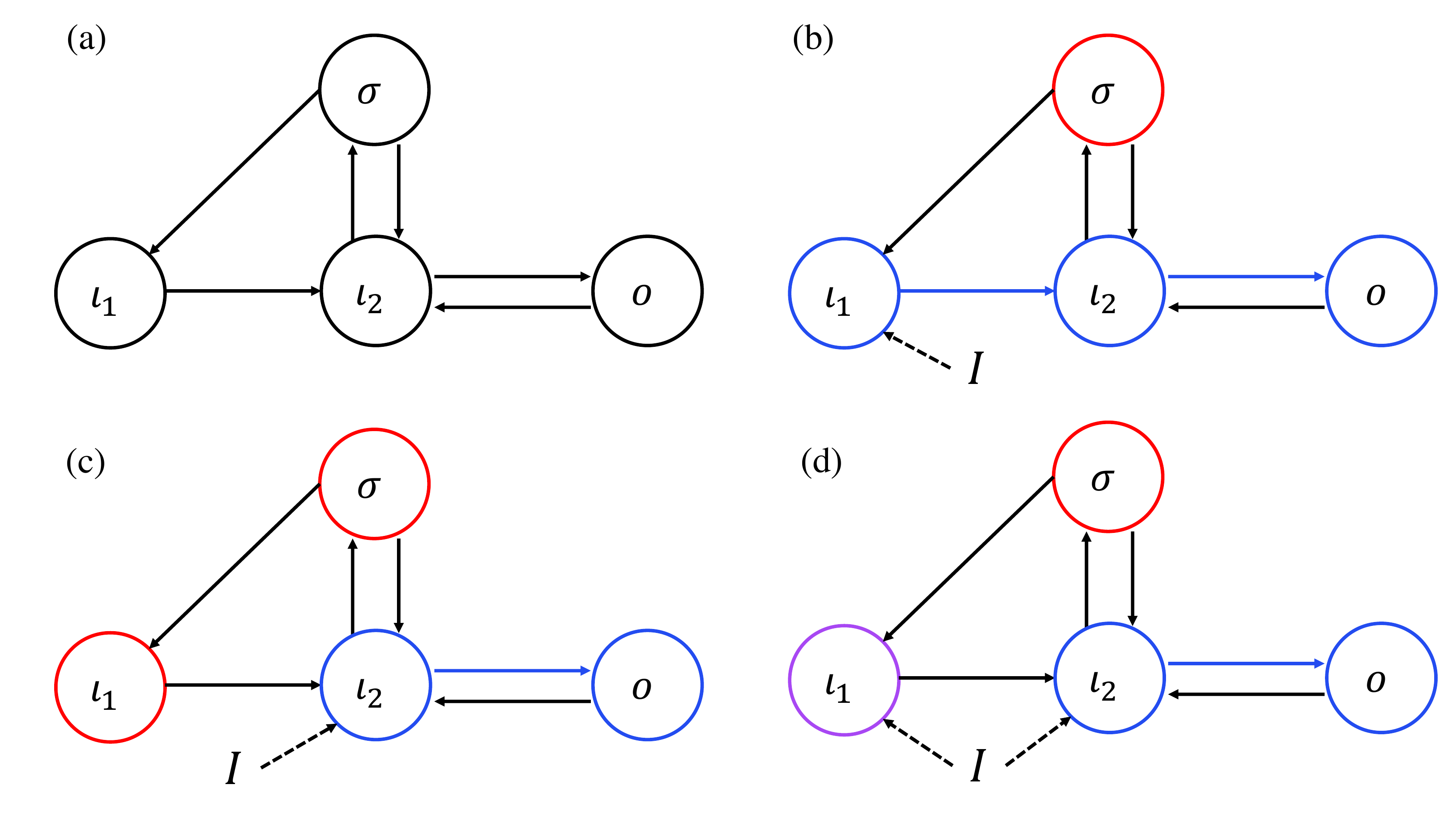}
\caption{\label{e_coli_network}a) Abstract network corresponding to the admissible ODE system in \eqref{e_coli_network_notation}. This abstract network is a core network with input nodes $\iota_{1}$ and $\iota_{2}$ and output node $o$. b) Analysis of the core subnetwork $\mathcal{G}_{1}$ between the input node $\iota_{1}$ and the output node $o$. $\iota_{1}$-simple nodes and paths are highlighted in blue and $\iota_{1}$-appendage nodes are highlighted in red. c) Analysis of the core subnetwork $\mathcal{G}_{2}$ between the input node $\iota_{2}$ and the output node $o$. $\iota_{2}$-simple nodes and paths are highlighted in blue and $\iota_{2}$-appendage nodes are highlighted in red. d) The core network $\mathcal{G}$ is the union between $\mathcal{G}_{1}$ and $\mathcal{G}_{2}$. $\iota_{2}$ and $o$ (in blue) are super-simple nodes, and $\sigma$ (in red) is an absolute appendage node. On the other hand, $\iota_{1}$ (in purple) is neither a simple nor an appendage node.}
\end{figure}
}

\begin{figure}[!ht]
\centering
\includegraphics[width=\linewidth]%
{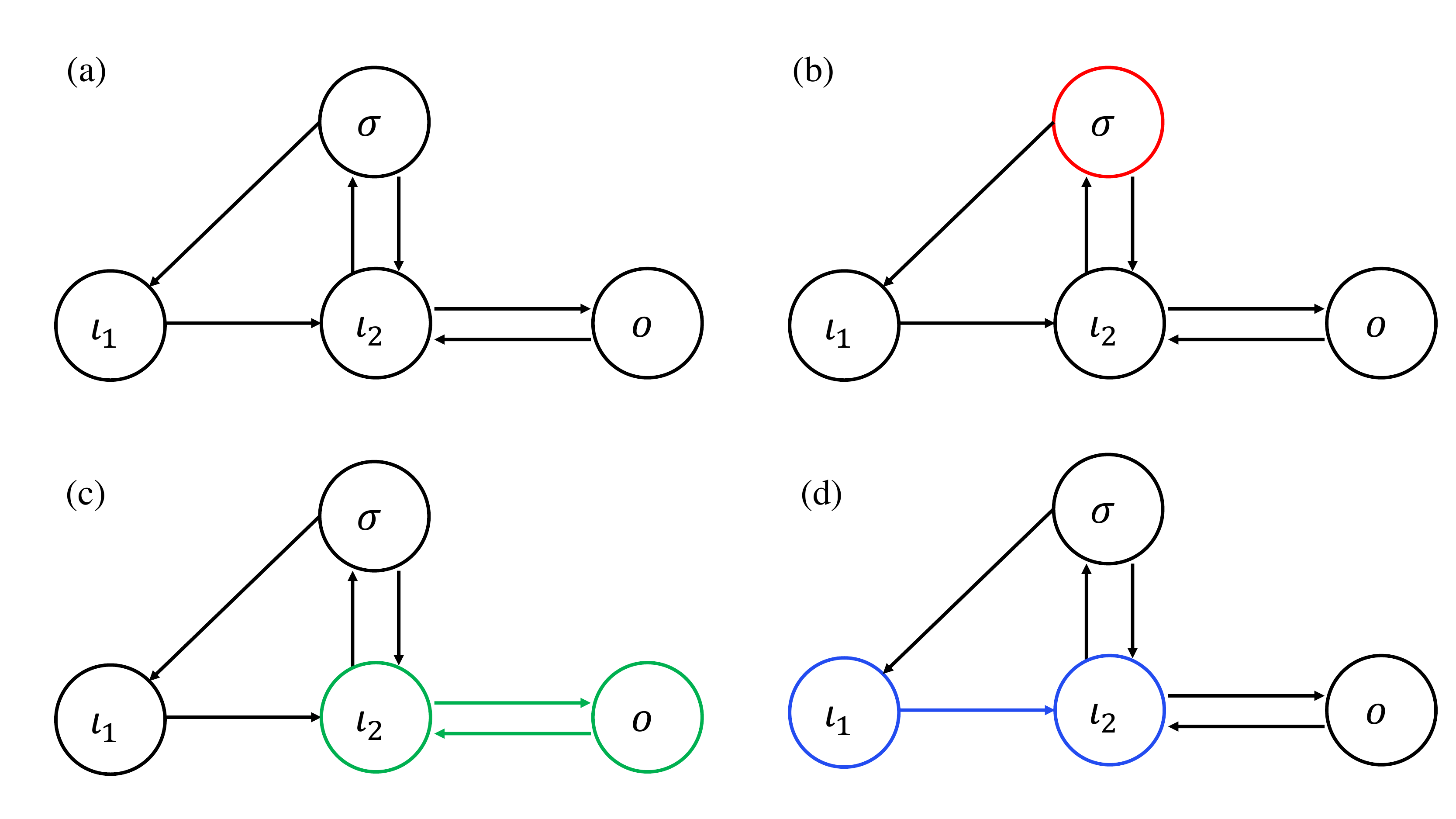}
\caption{\label{e_coli_network} Abstract network and the `subnetwork motifs' associated to the types of homeostasis that this network supports. (a) Abstract network $\mathcal{G}$ corresponding to the admissible ODE system in \eqref{e_coli_network_notation}. It is a core network with input nodes $\iota_{1}$ and $\iota_{2}$ and output node $o$. (b) The only absolutely appendage node in $\mathcal{G}$ is $\sigma$ (in red), and for every $\iota_{m}o$-simple path $S_{m}$, $\sigma$ is not $CS_{m}$-path equivalent to any other node, which means that $\sigma$ is an $\mathcal{A}_{\mathcal{G}}$-path component. (c) The absolutely super-simple nodes of $\mathcal{G}$ are $\iota_{2} > o$. The corresponding absolutely super-simple structural subnetwork is $\mathcal{L}'(\iota_{2},o)$ (in green).  (d) The greatest absolutely super-simple node of $\mathcal{G}$ (according to the natural ordering) is $\iota_{2}$. The input counterweight subnetwork $\mathcal{W}_{\mathcal{G}}$ is composed by nodes $\iota_{1}$ and $\iota_{2}$ (in blue).}
\label{abstract_network_ecoli}
\end{figure}

From the theory developed in this paper, it is clear that the network $\mathcal{G}$ is a core network with input nodes $\iota_{1}$ and $\iota_{2}$ and output node $o$. 
This core network is the union of networks $\mathcal{G}_{1}$ (the core subnetwork between $\iota_{1}$ and $o$) and $\mathcal{G}_{2}$ (the core subnetwork between $\iota_{2}$ and $o$).

The Jacobian matrix of network $\mathcal{G}$ is:
\begin{equation} \label{jacobian_matrix_e_coli}
J = 
\begin{pmatrix} 
f_{\iota_1, x_{\iota_1}} & 0 & f_{\iota_1, x_{\sigma}} & 0 \\
f_{\iota_2, x_{\iota_1}} & f_{\iota_2, x_{\iota_2}} & f_{\iota_2, x_{\sigma}} & f_{\iota_2, x_o} \\
0 & f_{\sigma, x_{\iota_2}} & f_{\sigma, x_{\sigma}} & 0 \\
0 & f_{o, x_{\iota_2}} & 0 & f_{o,x_o}
\end{pmatrix}
\end{equation}
The generalised homeostasis matrix of network $\mathcal{G}$ is:
\begin{equation} \label{homeostasis_matrix_e_coli}
\langle H \rangle = 
\begin{pmatrix} 
f_{\iota_1, x_{\iota_1}} & 0 & f_{\iota_1, x_{\sigma}} & -f_{\iota_1, \mathcal{I}} \\
f_{\iota_2, x_{\iota_1}} & f_{\iota_2, x_{\iota_2}} & f_{\iota_2, x_{\sigma}} & -f_{\iota_2, \mathcal{I}} \\
0 & f_{\sigma, x_{\iota_2}} & f_{\sigma, x_{\sigma}} & 0 \\
0 & f_{o, x_{\iota_2}} & 0 & 0
\end{pmatrix}
\end{equation}
Even though one can easily find the factorisation of matrix $\langle H \rangle$ above, by direct calculation, it is very instructive to apply the algorithm described in Subsection \ref{algorithm_core_network_multiple_input_nodes} to factorise $\det\!\big(\langle H \rangle\big)$ to reveal the structure of the underlying network motifs associated with homeostasis types supported by network $\mathcal{G}$

\ignore{
Its determinant can be decomposed as
\begin{equation} \label{analysis_e_coli_part_i}
\det\big(\langle H \rangle\big) = -f_{\iota_1, \mathcal{I}} \, \det(H_{\iota_1}) + f_{\iota_2, \mathcal{I}} \, \det(H_{\iota_2})
\end{equation}
where
\begin{equation} \label{analysis_e_coli_part_ii}
H_{\iota_1} = \begin{pmatrix} 
f_{\iota_2, x_{\iota_1}} & f_{\iota_2, x_{\iota_2}} & f_{\iota_2, x_{\sigma}} \\
0 & f_{\sigma, x_{\iota_{2}}} & f_{\sigma, x_{\sigma}} \\
0 & f_{o, x_{\iota_2}} & 0 
\end{pmatrix} 
\quad\text{and}\quad
H_{\iota_2} = \begin{pmatrix} 
f_{\iota_1, x_{\iota_1}} & 0 & f_{\iota_1, x_{\sigma}} \\
0 & f_{\sigma, x_{\iota_2}} & f_{\sigma, x_{\sigma}} \\
0 & f_{o, x_{\iota_2}} & 0
\end{pmatrix}
\end{equation}
}

We start by observing that $\sigma$ is the only absolutely appendage node of $\mathcal{G}$. Moreover, for every $\iota_{m}o$-simple path $S_{m}$, $\sigma$ is not $CS_{m}$-path equivalent to any other node. Therefore, $\mathcal{G}$ supports appendage homeostasis at $\mathcal{A}=\{\sigma\}$, and the corresponding irreducible factor is $\det(J_{\mathcal{A}})=f_{\sigma, x_{\sigma}}$ (see Figure \ref{e_coli_network}(b)).

The absolutely super-simple nodes of $\mathcal{G}$ are $\iota_{2}>o$, that define the absolutely super-simple structural subnetwork $\mathcal{L}'(\iota_{2}, o)=\{\iota_{2} \,\smash{{}^\leftarrow_\rightarrow}\, o\}$. As this subnetwork is composed by only $2$ nodes, its homeostasis matrix is a degree $1$ homeostasis block and the corresponding irreducible factor is  $\det\!\big(H(\mathcal{L}'(\iota_{2}, o))\big) = f_{o, x_{\iota_{2}}}$ (see Figure \ref{e_coli_network}(c)).

We have already determined the appendage and structural blocks of $\langle H \rangle$. Now we describe the input counterweight block. In fact, as the absolutely super-simple nodes of $\mathcal{G}$ are $\iota_{2} > o$ then $\mathcal{W}_{\mathcal{G}}=\{\iota_{1}\rightarrow\iota_{2}\}$ (see Figure \ref{e_coli_network}(d)). Furthermore, by Lemma \ref{properties_input_counterweight_subnetwork}, $\mathcal{W}_{\mathcal{G}}$ is a core network with input nodes $\iota_{1}$ and $\iota_{2}$, and output node $\iota_{2}$ and thus
\begin{equation}
    \langle H \rangle \left( \mathcal{W}_{\mathcal{G}}\right) = \begin{pmatrix}
    f_{\iota_{1}, x_{\iota_{1}}} & - f_{\iota_{1}, \mathcal{I}} \\
    f_{\iota_{2}, x_{\iota_{1}}} & - f_{\iota_{2}, \mathcal{I}}
    \end{pmatrix} 
\end{equation}
Hence, the complete factorisation of $\det\!\big(\langle H \rangle \big)$ is
\begin{equation} \label{analysis_e_coli_part_iii}
\begin{split}
\det\!\big(\langle H \rangle \big)
& = f_{\sigma, x_{\sigma}} \, f_{o, x_{\iota_{2}}} \, (-f_{\iota_{1}, \mathcal{I}} \, f_{\iota_{2}, x_{\iota_{1}}} + f_{\iota_{2}, \mathcal{I}} \, f_{\iota_{1}, x_{\iota_{1}}})
\end{split}
\end{equation}
Summarising, network $\mathcal{G}$ (Figure \ref{e_coli_network}) generically supports three types of homeostasis: (1) appendage (null-degradation) homeostasis associated with $\{\sigma\}$, (2) structural (Haldane) homeostasis associated with $\{\iota_{2} \,\smash{{}^\leftarrow_\rightarrow}\, o\}$ and (3) input counterweight homeostasis associated with $\{\iota_{1}\rightarrow\iota_{2}\}$

Now, specializing to the model equations \eqref{e_coli_network_notation}, we 
observe that, although the abstract network supports appendage and structural homeostasis, the model equations \eqref{e_coli_network_notation} cannot exhibit these types of homeostasis.
In fact, at equilibrium we have that
\begin{equation}
f_{\sigma, x_{\sigma}} = -\alpha_{2}k_{3}x_{\iota_{2}} - k_{5} < 0
\end{equation}
\begin{equation}
f_{o, x_{\iota_{2}}} = \alpha_{1}k_{2}(1 - x_{o}) \neq 0
\end{equation}
The first inequality follows from the fact that all parameters are positive and $x_{\iota_{2}}$ is positive at equilibrium.
The last inequality follows from that fact that if $\alpha_{1}k_{2}(1 - x_{o}) = 0 \Rightarrow x_{o} = 1$, at equilibrium, then one would have $\dot{x}_{o} = - k_{4} \neq 0$.

\ignore{
\begin{equation}
\frac{\partial \phi}{\partial x_{\iota_{1}}} = \frac{Ne^{F}}{2(1 + e^{F})^{2}} 
\quad\text{and}\quad
\frac{\partial \phi}{\partial \mathcal{I}} = \frac{-Ne^{F}\left( \frac{1}{K^{\textrm{off}}_{a}} - \frac{1}{K^{\textrm{on}}_{a}}\right)}{(1 + e^{F})^{2}\left(1 + \frac{\mathcal{I}}{K^{\textrm{off}}_{a}} \right)\left(1 + \frac{\mathcal{I}}{K^{\textrm{on}}_{a}} \right)}
\end{equation}
Recall that $K^{\textrm{off}}_{a} \neq K^{\textrm{on}}_{a}$.
}

This leaves the only remaining possibility: input counterweight homeostasis. To verify that the model equations \eqref{e_coli_network_notation} indeed exhibit input counterweight homeostasis we compute, assuming that both $\phi_{x_{\iota_1}}$ and $\phi_{\mathcal{I}}$ (see Remark \ref{rmk:non-zero}) are non-zero:
\begin{equation} \label{analysis_e_coli_part_iiii}
\begin{aligned}
f_{\iota_{1}, \mathcal{I}} & = - \phi_{\mathcal{I}} (x_{\iota_1},\mathcal{I})  \, (\gamma_{R} + \gamma_{B}x_{\sigma}^{2}) \\
f_{\iota_{2}, x_{\iota_1}} & = \phi_{x_{\iota_1}} (x_{\iota_1},\mathcal{I}) \, k_{1} \, (1 - x_{\iota_2}) \\
f_{\iota_{2}, \mathcal{I}} & = \phi_{\mathcal{I}} (x_{\iota_1},\mathcal{I})  \, k_{1} \, (1 - x_{\iota_2}) \\
f_{\iota_{1}, x_{\iota_1}} & = -\phi_{x_{\iota_1}} (x_{\iota_1},\mathcal{I}) \, (\gamma_{R} + \gamma_{B}x_{\sigma}^{2})
\end{aligned}
\end{equation}
Thus, for any $C^1$ function $\phi$, we have
\begin{equation} \label{analysis_e_coli_part_v}
-f_{\iota_{1}, \mathcal{I}} \, f_{\iota_{2}, x_{\iota_{1}}} + f_{\iota_{2}, \mathcal{I}} \, f_{\iota_{1}, x_{\iota_{1}}} \equiv 0
\end{equation}
In particular, $x_o'(\mathcal{I})=0$ for all $\mathcal{I}$ and hence, the model equations \eqref{e_coli_network_notation} exhibits, not only infinitesimal homeostasis (of the input counterweight type), but perfect homeostasis.

A consequence of the above calculation is that the property of perfect homeostasis in the model \eqref{e_coli_network_notation} for \textit{E. coli} chemotaxis is \emph{robust}, in the sense that it does not depend on the values of the parameters of the model. Even more, it holds for a much larger set of perturbations than just parameter change. Consider the space of all sufficiently regular vector fields $\mathcal{X}_{\mathcal{G}}$ given by the right-handed side of \eqref{e_coli_general_admissible}, with an appropriate topology (for example, the $C^1$ vector fields with the $C^1$ topology). Now consider the subspace $\mathcal{Y}\subset\mathcal{X}_{\mathcal{G}}$ of vector fields whose first two components are the same as in \eqref{e_coli_network_notation}, with an arbitrary function $\phi$ (of class $C^1$), and arbitrary last two components. Then it is clear that $\mathcal{Y}$ is a closed infinite dimensional subspace of $\mathcal{X}_{\mathcal{G}}$ that contain the model equations \eqref{e_coli_network_notation} and satisfy \eqref{analysis_e_coli_part_v}. In other words, any perturbation of \eqref{e_coli_network_notation} contained in the infinite dimensional space $\mathcal{Y}$ displays perfect homeostasis.

\begin{remark} \rm \label{rmk:non-zero}
To show that in the model equations \eqref{e_coli_network_notation}, both $\phi_{x_{\iota_1}}$ and $\phi_{\mathcal{I}}$ are non-zero, recall the correspondence between the variables $m \leftrightarrow x_{\iota_{1}}$ and $L \leftrightarrow \mathcal{I}$ and differentiate the equations in \eqref{definition_phi}, obtaining
\begin{equation} \label{derivates_phi}
\begin{aligned}
    \phi_{x_{\iota_1}}(x_{\iota_1}, \mathcal{I}) & = \frac{Ne^{F(x_{\iota_1}, \mathcal{I})}}{2(1 + e^{F(x_{\iota_1}, \mathcal{I})})^{2}} \\
    \phi_{\mathcal{I}}(x_{\iota_1}, \mathcal{I}) & =\frac{Ne^{F(x_{\iota_1}, \mathcal{I})}\left(\frac{1}{K^{\textrm{on}}_{a}} - \frac{1}{K^{\textrm{off}}_{a}} \right)}{(1 + e^{F(x_{\iota_1}, \mathcal{I})})^{2}\left(1 + \frac{\mathcal{I}}{K^{\textrm{on}}_{a}}\right)\left(1 + \frac{\mathcal{I}}{K^{\textrm{off}}_{a}}\right)}
\end{aligned}
\end{equation}
As $K^{\textrm{on}}_{a} \neq K^{\textrm{off}}_{a}$, we conclude by \eqref{derivates_phi} that in the studied model $\phi_{x_{\iota_1}} \neq 0$ and $\phi_{\mathcal{I}} \neq 0$.
\end{remark}

\begin{remark} \rm
As can be seen in Figure \ref{e_coli_plot1} it looks like the time series of all three variables $a_p$, $b_p$ and $y_p$ exhibit homeostatic behavior. We can use our results to show that this is indeed the case. By interchanging the roles of the nodes $\iota_{2}$ and $\sigma$ with the output node $o$, we can compute new generalized homeostasis matrices $H_{(\iota_{2}\leftrightarrow o)}$ and $H_{(\sigma\leftrightarrow o)}$ and show that the corresponding input-output functions have the same irreducible factor associated with input counterweight homeostasis (see \eqref{analysis_e_coli_part_v}) which, in turn, causes all of them to vanish simultaneously. The matrices $H_{(\iota_{2}\leftrightarrow o)}$ and $H_{(\sigma\leftrightarrow o)}$ can be obtained from the Jacobian matrix  \eqref{jacobian_matrix_e_coli} by replacing the second and third columns, respectively, by $(-f_{\iota,\mathcal{I}},0,0)^t$. The corresponding determinants are
\begin{equation*}
\begin{aligned}
    & \det(H_{(\iota_{2}\leftrightarrow o)}) = f_{\sigma, x_{\sigma}}f_{o, x_{o}}(-f_{\iota_{1}, \mathcal{I}}f_{\iota_{2}, x_{\iota_{1}}} + f_{\iota_{2}, \mathcal{I}}f_{\iota_{1}, x_{\iota_{1}}})  \\ 
    & \det(H_{(\sigma\leftrightarrow o)}) = f_{o, x_{o}}f_{o, x_{\iota_{2}}} (-f_{\iota_{1}, \mathcal{I}}f_{\iota_{2}, x_{\iota_{1}}} + f_{\iota_{2}, \mathcal{I}}f_{\iota_{1}, x_{\iota_{1}}}) 
\end{aligned}
\end{equation*}
\end{remark}

We have shown the occurrence of perfect homeostasis for an infinite dimensional `subspace of vector fields' containing \eqref{e_coli_network_notation}. In what follows we show that there is another infinite dimensional `subspace of vector fields' containing \eqref{e_coli_network_notation} where perfect homeostasis does not occur.

Consider the following $1$-parameter family of perturbations of \eqref{e_coli_network_notation}, for $\epsilon \geqslant 0$,
\begin{equation} \label{e_coli_general_perturbation}
\begin{aligned}
\dot{x}_{\iota_1} & = \gamma_{R}  \, (1 - \phi(x_{\iota_1}, \mathcal{I})) - \gamma_{B} \, x_{\sigma}^{2} \, \phi(x_{\iota_1}, \mathcal{I}) -\epsilon \, \psi(x_{\iota_{1}}) \\
\dot{x}_{\iota_{2}} & = \phi(x_{\iota_1}, \mathcal{I}) \, k_{1} \, (1 - x_{\iota_2}) - k_{2} \, (1 - x_{o}) \, x_{\iota_2} - k_{3} \, (1 - x_{\sigma}) \, x_{\iota_2} \\
\dot{x}_{\sigma} & = \alpha_{2} \, k_{3} \, (1 - x_{\sigma}) \, x_{\iota_2} - k_{5} \, x_{\sigma} \\
\dot{x}_{o} & = \alpha_{1} \, k_{2} \, (1 - x_{o}) \, x_{\iota_2} - k_{4} \, x_{o}
\end{aligned}
\end{equation}
where $\psi$ is a $C^1$ function. 
It is clear that \eqref{e_coli_general_perturbation} is admissible for the network $\mathcal{G}$ for all $\epsilon \geqslant 0$ and coincides with \eqref{e_coli_network_notation} when $\epsilon=0$.
The expressions for $f_{\iota_{1}, \mathcal{I}}$, $f_{\iota_{2}, \mathcal{I}}$ and $f_{\iota_{2}, x_{\iota{1}}}$ are independent of $\epsilon$ and thus are the same as in the original system \eqref{analysis_e_coli_part_iiii}. 
The expression for $f_{\iota_{1}, x_{\iota{1}}}$ is
\begin{equation} \label{new_formula_derivative}
    f_{\iota_{1}, x_{\iota_1}} = -\phi_{x_{\iota_1}} (x_{\iota_1},\mathcal{I}) \, (\gamma_{R} + \gamma_{B}x_{\sigma}^{2}) - \epsilon \, \psi_{x_{\iota_{1}}}(x_{\iota_{1}})
\end{equation}
Since the equations of $f_{\sigma, x_{\sigma}}$ and of $f_{o, x_{\iota_{2}}}$ are independent of $\epsilon$, the same argument as before shows that \eqref{e_coli_general_perturbation} do not exhibit appendage or structural homeostasis, for all $\epsilon \geqslant 0$. As for the input counterweight homeostasis factor, we get
\begin{equation} \label{final_analysis_e_coli_not_generic_part_i}
 -f_{\iota_{1}, \mathcal{I}} \, f_{\iota_{2}, x_{\iota_{1}}} + f_{\iota_{2}, \mathcal{I}} \, f_{\iota_{1}, x_{\iota_{1}}} = - \epsilon \, k_{1} \, \phi_{\mathcal{I}} (x_{\iota_1},\mathcal{I}) \, (1 - x_{\iota_2}) 
 \, \psi_{x_{\iota_{1}}}(x_{\iota_{1}})
\end{equation}
Recall that $\phi_{\mathcal{I}} (x_{\iota_1},\mathcal{I}) \neq 0$. Generically, $(1 - x_{\iota_2}) \neq 0$, as well. Otherwise, $x_{\iota_{2}} = 1$ at equilibrium, and so $x_{o}$ and $x_{\sigma}$ must satisfy the over-determined linear system
\begin{equation} \label{final_analysis_e_coli_not_generic_part_ii}
\begin{aligned}
0 & = - k_{2} \, (1 - x_{o}) - k_{3} \, (1 - x_{\sigma}) \\
0 & = \alpha_{2} \, k_{3} \, (1 - x_{\sigma}) - k_{5} \, x_{\sigma} \\
0 & = \alpha_{1} \, k_{2} \, (1 - x_{o}) - k_{4} \, x_{o}
\end{aligned}
\end{equation}
Therefore, if the function $\psi_z(z)$ does not vanish on an interval (say, $\psi(z)=z$), the model \eqref{e_coli_general_perturbation} do not generically, exhibit perfect homeostasis for all $\epsilon > 0$. It is clear that this construction can be carried out for an arbitrary number of independent parameters and functions $(\epsilon_n,\psi_n)_{n}$, showing that the set of perturbations that destroy perfect homeostasis is not contained in any finitely parametric family of vector fields. This is a manifestation of the well-known phenomenon in singularity theory, that an exactly flat function has `infinite codimension' (see Section \ref{sec:discussion}). 

Finally, numerical simulations suggests that, at least when $\epsilon>0$ is small, the model \eqref{e_coli_general_perturbation} displays near perfect homeostasis, see Figure \ref{e_coli_plot2}. Finally, it seems to be possible to show that if $\psi_z(z)$ vanishes at some point $z_0$ (say, $\psi(z)=z^2$) then infinitesimal homeostasis occurs for an open set of parameters in the model \eqref{e_coli_general_perturbation}. 

\begin{figure}[!htp]
\centering
\includegraphics[width=\linewidth,trim=0cm 1cm 0cm 0.5cm,clip=true]{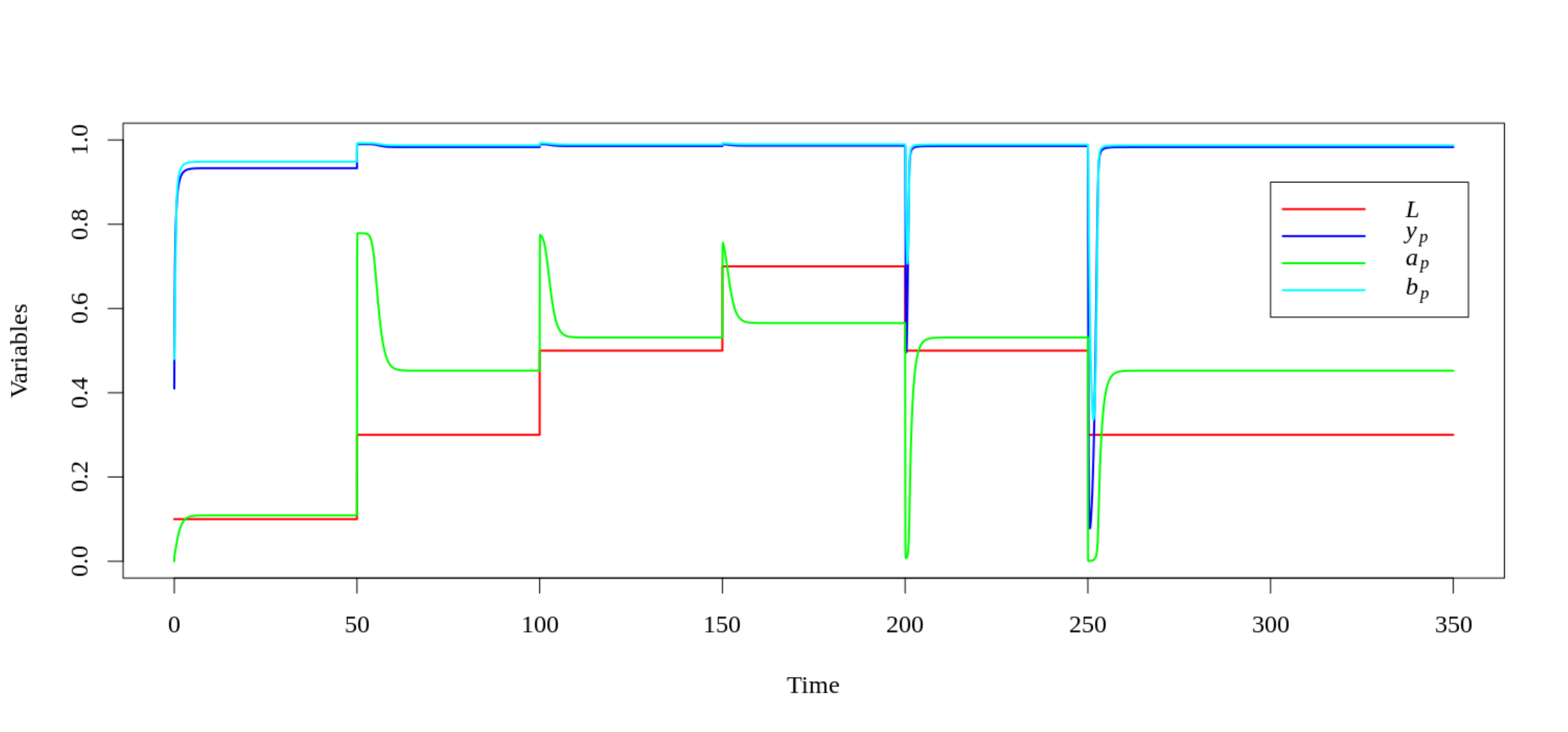} \\ [1ex]
\includegraphics[width=\linewidth,trim=0cm 0.5cm 0cm 1cm,clip=true]{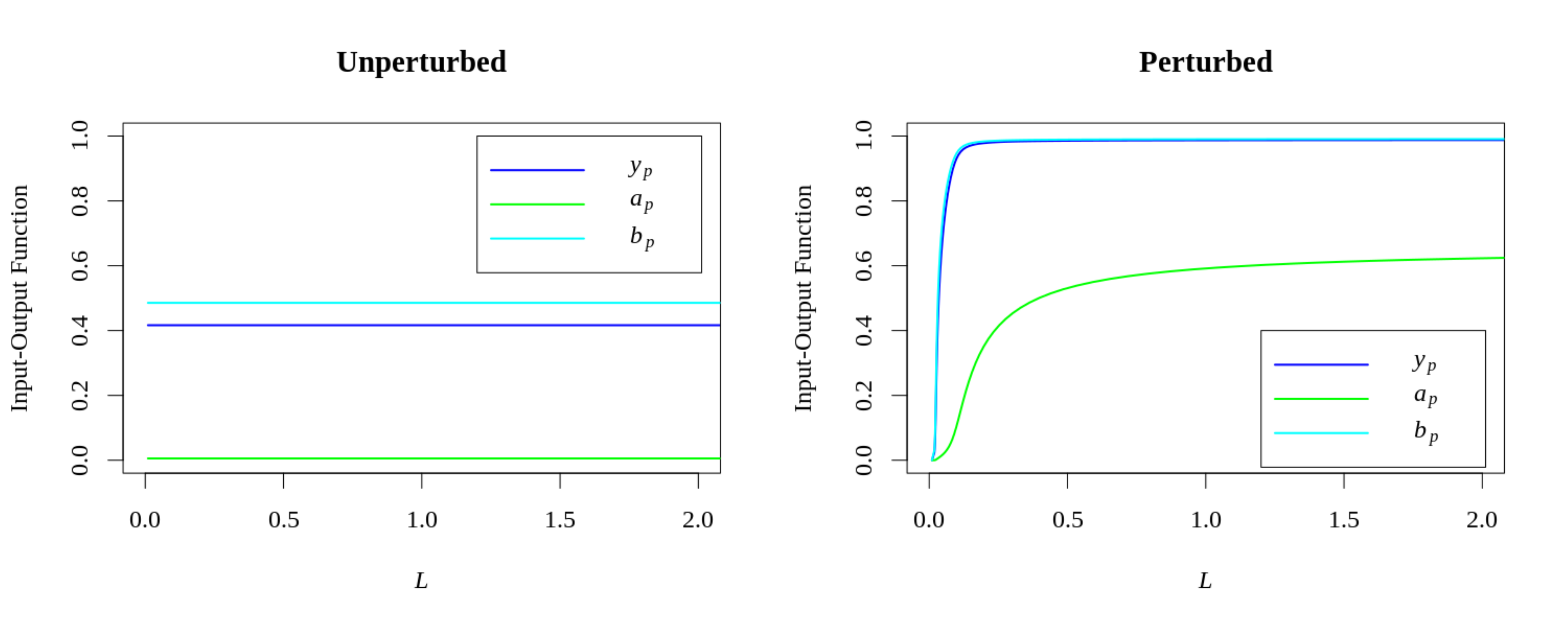}
\caption{\label{e_coli_plot2} (Upper) Time series of the model \eqref{e_coli_general_perturbation}, with $\psi(z)=z$ and $\epsilon=0.05$, showing near perfect homeostasis of the original variables $a_p \leftrightarrow x_{\iota_2}$ (green), $b_p \leftrightarrow x_{\sigma}$ (cyan), $y_p \leftrightarrow x_{o}$ (blue) at the corresponding non-dimensional equilibrium. Input parameter $L \leftrightarrow \mathcal{I}$ is given by a step function (red curve). Other parameters were set to non-dimensional values of \cite[Table~2]{edgington18}.
(Lower) Input-output functions of the original variables $a_p \leftrightarrow x_{\iota_2}$ (green), $b_p\leftrightarrow x_{\sigma}$ (cyan), $y_p \leftrightarrow x_{o}$ (blue), as functions of input parameter $L \leftrightarrow \mathcal{I}$, for the
model \eqref{e_coli_general_perturbation}.
(Left) For $\epsilon=0$, which reduces to the original model \eqref{original_e_coli}, we have perfect homeostasis (constant input-output functions) and (Right) for $\epsilon=0.05$ we have near perfect homeostasis. Time series were computed using the software \textsc{XPPAut} and input-output functions were computed by numerical continuation of an equilibrium point using \textsc{Auto} from \textsc{XPPAut} \cite{bard02}.
}
\end{figure}

\section{Discussion}
\label{sec:discussion}

%{\color{red}
In this paper, we aimed to generalize the classification results of \cite{wang20} and establish the topological characterization of infinitesimal homeostasis in networks with one input parameter and multiple input nodes.
Using this extended theory, we were able to characterize the homeostasis in a model of bacterial chemotaxis. We notice that our results are only applicable to classify infinitesimal homeostasis in input-output networks with multiple input nodes and a unique input parameter. The classification of infinitesimal homeostasis in networks with multiple input parameters (as defined in \cite{golubitsky2018homeostasis}) and the study of near perfect homeostasis are left for future work.

Wang et al. \cite{wang20} identified that two types of homeostasis may be present in input-output networks with only one input node: appendage and structural homeostasis. We were able to generalize such concepts for networks with multiple input nodes. In particular, we notice that applying theorems \ref{characterization_appendage_homeostasis_general_G} (regarding appendage homeostasis) and \ref{structural_homeostasis_G} (regarding structural homeostasis) to networks with only one input node, we recover the findings of \cite{wang20}. This observation is nice as it shows the theory proposed by Wang \etal. may be generalized to other abstract networks.

Interestingly, we also identified that networks with multiple input nodes support a new class of homeostasis which occurs by the balance of the input nodes. We called it \emph{input counterweight homeostasis}. The input counterweight homeostasis block corresponds to the generalized homeostasis matrix of the subnetwork `between' the input nodes and the first absolutely super-simple node. Unlike the multiple input node case, the determinant of the corresponding block in the single input node is always non-zero by definition. Hence, this class of homeostasis does not occur in networks with a unique input node.

The biological relevance of input counterweight homeostasis may be exemplified by the model of chemotaxis studied in this paper. Indeed, although the abstract network corresponding to this model (see Figure \ref{abstract_network_ecoli}) supports all classes of homeostasis, the model equations \eqref{e_coli_network_notation} exhibit input counterweight rather than appendage or structural homeostasis. Noticeably, this fact (and hence the presence of perfect homeostasis) does not depend on the definition of the function $\phi$. As this function represents the receptor signaling activity \cite{edgington18}, this means from the biological point of view that if one incorporates more (or less) details in the expression of $\phi$ or $F$ (see \eqref{definition_phi}), the property of perfect homeostasis will remain preserved.

Mathematically, this means that there is an infinite dimensional `space of admissible perturbations' of \eqref{original_e_coli} that \emph{exhibit} perfect homeostasis. An `admissible perturbation' is  a vector field that is compatible with the network structure of \eqref{original_e_coli}. This goes beyond `robustness' in the sense of preservation by changes in the parameters.

On the other hand, if we perturb the system in a way to disrupt the symmetry that leads to input counterweight homeostasis (as done in \eqref{e_coli_general_perturbation}), then the perturbed model will not present infinitesimal homeostasis. In other words, there is an infinite dimensional `space of admissible perturbations' of \eqref{original_e_coli} that destroy perfect homeostasis. Again, by `admissible perturbation' we mean a vector field that is compatible with the same network structure of \eqref{original_e_coli}. This is illustrated in Figure \ref{e_coli_plot2}. From the singularity theory point of view, we say that the singular points of the ODE describing the chemotaxis have \emph{infinite codimension}.

In the model equations \eqref{original_e_coli}, it is clear that perfect homeostasis is not a generic property. However, it seems that the less strict properties `infinitesimal homeostasis' and `near perfect homeostasis' are persistent for a wide class of perturbations that destroy perfect homeostasis, as mentioned above (see Figure \ref{e_coli_plot2}).

It seems reasonable to expect that other biological systems with perfect homeostasis harbor similar features. Indeed, considering that homeostasis occurs as an emergent dynamical property of the system, rather than due to the presetting of a target value, it is expected that in such systems the input-output function is only approximately flat. This means that we expect "typical" singularities associated to biological systems to have finite codimension. In such a context, the approach based on singularity theory is potentially useful. In fact, it is shown in \cite{gs17,gs18} that `infinitesimal homeostasis' is a finite codimension property. Based on the above, we formulate the following hypothesis.

\paragraph{Conjecture.} Fix a multiple input nodes input-output network $\mathcal{G}$ and an appropriate functional space of vector fields $\mathcal{X}_{\mathcal{G}}$ compatible with the network $\mathcal{G}$. Consider we define in $\mathcal{X}_{\mathcal{G}}$ a topology $\tau$ that reflects the derivatives of the vector field with respect to the input (e.g., the equivalent of Whitney topology). Then, we would expect the following to occur:
\begin{enumerate}[(1)]
\item Near perfect homeostasis occurs generically in $\mathcal{X}_{\mathcal{G}}$, in the sense that it holds on an open dense subset of $\mathcal{X}_{\mathcal{G}}$.
\item Infinitesimal homeostasis occurs generically in $\mathcal{X}_{\mathcal{G}}$, in the sense that it holds on an open dense subset of finite-parameter families in $\mathcal{X}_{\mathcal{G}}$.
\end{enumerate}

The discussion of this conjecture is relevant, as a precise mathematical description of a biological system through an exact known system of (differential) equations is very rarely plausible. Indeed, in the majority of cases the model equations are known only through approximate empirical formulas (as for example Michaelis-Menten kinetics). This means that in biology there are approximation errors not only on the parameters used to reproduce the behavior of the system, but also on the equations employed to model the system \cite{thom69}. Therefore, to say a property presented by the studied system is robust, it would not be enough to evaluate the preservation of the property under a appropriate (i.e., admissible) perturbation on the parameters' values, but also under a perturbation on the vector field.  The precise definition of such perturbations and the study of their properties is beyond the scope of this paper. Similar remarks have been made in \cite{kitano2004,kitano2007} but in a different context.

Taking altogether, in this paper we extended the theory of homeostasis topological classification to input-output networks with multiple input nodes. Applying our results, we were able to demonstrate that the perfect homeostasis property presented by the studied model of \textit{E. coli} chemotaxis is due to input counterweight homeostasis. 
Further work is necessary to extend our theory to networks with multiple inputs and to study the genericity and robustness of singularities in biological systems.

%END IGNORE

\paragraph{Acknowledgments.} We thank Martin Golubitsky, Yangyang Wang, Zhengyuan Huang, Pedro P. A. C. Andrade and Misaki Yamada for helpful conversations.
The research of JLOM and FA was supported by Funda\c{c}\~ao de Amparo \`a Pesquisa do Estado de S\~ao Paulo (FAPESP) grant 2019/12247-7. JLOM is also supported by a scholarship from the EPSRC Centre for Doctoral Training in Statistical Applied Mathematics at Bath (SAMBa), under the project EP/S022945/1.

\section*{Appendices}

\appendix

\section{Irreducibility of Homogeneous Polynomials}
\label{ap:AppA}

\setcounter{equation}{0}

In this appendix we prove general results about irreducibility of certain homogeneous polynomials that are used in Section \ref{mathematical_arguments_and_proofs} and Appendix~\ref{ap:AppB}.

\begin{lemma} \label{lemma_commum_factors_det_H}
Let $P(y_{1}, y_{2}, \ldots, y_{n})$ be an homogeneous polynomial of degree 1, i.e.,
\begin{equation}
    P(y_{1}, y_{2}, \ldots, y_{n}) \equiv a_{1}y_{1} + \cdots + a_{n}y_{n}
\end{equation}
then, a polynomial $Q(y_{1}, y_{2}, \ldots, y_{n})$ divides $P$ iff either $Q(y_{1}, y_{2}, \ldots, y_{n}) \equiv b_{0}$ and $b_{0}$ divides each term $a_{i}$, for all $i = 1,\ldots,n$, or there is a coefficient $c$ such that
\begin{equation}
    Q(y_{1}, y_{2}, \ldots, y_{n}) \equiv b_{1}y_{1} + \cdots + b_{n}y_{n}
\end{equation}
and $a_{i} = c \cdot b_{i}$, for all $i = 1,\ldots,n$.
\end{lemma}

\begin{proof}
($\Leftarrow$) It is straightforward to see that in both cases $Q$ divides $P$.

\noindent
($\Rightarrow$) Suppose that a polynomial $Q(y_{1}, y_{2}, \ldots, y_{n})$ divides $P$. As $P$ is an homogeneous polynomial of degree 1, then $Q$ must have at most degree $1$, i.e., we can explicitly write $Q$ as
\begin{equation} \label{explicit_formula_Q}
    Q(y_{1}, y_{2}, \ldots, y_{n}) = b_{0} + b_{1}y_{1} + \cdots + b_{n}y_{n}
\end{equation}
Moreover, there is a polynomial $R(y_{1}, y_{2}, \ldots, y_{n})$ such that $P \equiv QR$. By the same argument, we can explicitly write $R$ as
\begin{equation} \label{explicit_formula_R}
    R(y_{1}, y_{2}, \ldots, y_{n}) = c_{0} + c_{1}y_{1} + \cdots + c_{n}y_{n}
\end{equation}
By \eqref{explicit_formula_Q} and \eqref{explicit_formula_R}, we have
\begin{equation} \label{lemma_homogeneous_polynomials_part_i}
\begin{aligned}
    P(y_{1}, y_{2}, \ldots, y_{n}) & \equiv Q(y_{1}, y_{2}, \ldots, y_{n})R(y_{1}, y_{2}, \ldots, y_{n}) \\[2mm]
    a_{1}y_{1} + \cdots + a_{n}y_{n} & \equiv (b_{0} + b_{1}y_{1} + \cdots + b_{n}y_{n}) \cdot (c_{0} + c_{1}y_{1} + \cdots + c_{n}y_{n}) \\
    a_{1}y_{1} + \cdots + a_{n}y_{n} & \equiv b_{0}c_{0} 
    + \sum_{i = 1}^{n}(b_{i}c_{0} + b_{0}c_{i})y_{i}
    + \sum_{1 \leq i < j \leq n}(b_{i}c_{j} + b_{j}c_{i})y_{i}y_{j} \\
    & \qquad + \sum_{i = 1}^{n}b_{i}c_{i}y_{i}^{2}
\end{aligned}
\end{equation}
From \eqref{lemma_homogeneous_polynomials_part_i}, we conclude:
\begin{equation}
\begin{aligned} \label{lemma_homogeneous_polynomials_part_ii}
    & b_{0}c_{0} = 0 \\
    & b_{i}c_{0} + b_{0}c_{i} = a_{i}, \quad\text{for all}\; i = 1,\ldots,n \\
    & b_{i}c_{j} + b_{j}c_{i} = 0, \quad 1 \leq i < j \leq n \\
    & b_{i}c_{i} = 0, \qquad\qquad\text{for all}\; i = 1,\ldots,n
\end{aligned}
\end{equation}
From \eqref{lemma_homogeneous_polynomials_part_ii}, we have $b_{0} = 0$ or $c_{0} = 0$. Suppose that $b_{0} = 0$. Therefore, by \eqref{explicit_formula_Q}, we get
\begin{equation}
    Q(y_{1}, y_{2}, \ldots, y_{n}) = b_{1}y_{1} + \cdots + b_{n}y_{n}
\end{equation}
Taking $b_{0} = 0$ in \eqref{lemma_homogeneous_polynomials_part_ii}, we get
\begin{equation}
    b_{i}c_{0} = a_{i}, \quad\text{for all}\; i = 1,\ldots,n
\end{equation}
Considering $c_{0} \equiv c$, one of the cases of the lemma is obtained. Suppose now that $c_{0} = 0$. In that case, by \eqref{lemma_homogeneous_polynomials_part_ii}, we get
\begin{equation} \label{lemma_homogeneous_polynomials_part_iii}
    b_{0}c_{i} = a_{i}, \quad\text{for all}\; i = 1,\ldots,n
\end{equation}
and therefore, as $a_{i} \neq 0$, then $c_{i} \neq 0$. However, we must satisfy $b_{i}c_{i} = 0$, for all $i = 1,\ldots,n$. Therefore, we conclude that $b_{i} = 0$, for all $i = 1,\ldots,n$, which means that
\begin{equation*}
    Q(y_{1}, y_{2}, \ldots, y_{n}) \equiv b_{0}
\end{equation*}
From \eqref{lemma_homogeneous_polynomials_part_iii}, we have that $b_{0}$ divides each term $a_{i}$, for all $i = 1,\ldots,n$.
\end{proof}

\begin{corollary} \label{corollary_irreductible_coefficients_irreductible_polynomial}
Let $P(y_{1}, y_{2}, \ldots, y_{n})$ be an homogeneous polynomial of degree 1, i.e.,
\begin{equation}
    P(y_{1}, y_{2}, \ldots, y_{n}) \equiv a_{1}y_{1} + \cdots + a_{n}y_{n}.
\end{equation}
such that $a_{1}, \cdots, a_{n}$ do not share common factors. Then $P$ is irreducible.
\end{corollary}
\begin{proof}
Suppose there is a polynomial $Q(y_{1}, y_{2}, \ldots, y_{n})$ which is a divisor of $P$. By Lemma \ref{lemma_commum_factors_det_H}, there is a coefficient $c$ that divides every term $a_{i}$. As $a_{1}, \ldots, a_{n}$ do not share common factors, $c = \pm 1$, and hence, either $Q(y_{1}, y_{2}, \ldots, y_{n}) \equiv \pm 1$ or $Q(y_{1}, y_{2}, \ldots, y_{n}) \equiv \pm P(y_{1}, y_{2}, \ldots, y_{n})$. That is, $P$ is irreducible.
\end{proof}

\section{Non-triviality of the Homeostasis Determinant}
\label{ap:AppB}

In this appendix we prove the claim that if the self couplings are not identically zero and the output node is downstream from all input nodes (see Remark~\ref{rmk:nontriviality}), then $\det \langle H \rangle$ is not identically zero.

Let $f$ be a smooth admissible vector field associated to a network $\mathcal{G}$ and denote by $J$ its Jacobian matrix. 
Let $\mathcal{G}$ be composed by nodes $\sigma_{1}, \sigma_{2}, \cdots, \sigma_{n}$. Then
\[
    J = \left[\begin{array}{cccc}
         f_{\sigma_{1}, \sigma_{1}} & f_{\sigma_{1}, \sigma_{2}} & \cdots & f_{\sigma_{1}, \sigma_{n}} \\
         f_{\sigma_{2}, \sigma_{1}} & f_{\sigma_{2}, \sigma_{2}} & \cdots & f_{\sigma_{2}, \sigma_{n}} \\
         \vdots & \vdots & \ddots & \vdots \\
         f_{\sigma_{n}, \sigma_{1}} & f_{\sigma_{n}, \sigma_{2}} & \cdots & f_{\sigma_{n}, \sigma_{n}}
    \end{array}\right]
\]
where $f_{\sigma_{i}, \sigma_{j}}$ denote the partial derivatives of $f$.
We say that $f$ or $J$ is \emph{generic} if the following conditions are satisfied:
\begin{enumerate}[(a)]
\item $f_{\sigma_{i}, \sigma_{j}} \equiv 0$ if and only if there is no arrow $\sigma_{j} \rightarrow \sigma_{i}$.
\item there is no polynomial relation among $f_{\sigma_{i}, \sigma_{j}}$.
\end{enumerate}
This means that $f_{\sigma_{i}, \sigma_{j}}$ may be seen as \emph{algebraically independent variables} and the $J$ may be seen as a `decorated' adjacency matrix of $\mathcal{G}$. 
In particular, $\det J$ (or the determinant of any sub-matrix of $J$) may be seen as a polynomial on the algebraically independent variables $f_{\sigma_{i}, \sigma_{j}}$ and thus it vanishes identically if and only if all its summands are identically zero (see~\cite{br91}).
Note that the set of generic admissible vector fields $f$ is an open set in any appropriate topology on the space of all smooth admissible vector fields.

\begin{lemma}\label{jacobian_not_null}
If all the self-couplings of a generic Jacobian matrix $J$ are not identically zero (i.e., $f_{\sigma_{i}, \sigma_{i}} \not\equiv 0$) then $\det J$ is not identically zero.
\end{lemma}
\begin{proof}
Let $\mathcal{G}$ be a network with $n$ nodes $\sigma_{1}, \sigma_{2}, \cdots, \sigma_{n}$.
Consider the product of diagonal elements of $J$:
\[
 f_{\sigma_{1}, \sigma_{1}} \ldots f_{\sigma_{n}, \sigma_{n}} \not\equiv 0
\]
First of all, it is a summand of $\det J$. 
Moreover, because $f$ (and hence $J$) is generic and all the self-couplings are not identically zero, the product is not identically zero.
Therefore, we conclude that $\det J$ cannot be identically zero.
\end{proof}

\begin{proposition} \label{homeostasis_matrix_one_inpu_node_not_null}
Let $\mathcal{G}$ be an input-output network with input node $\iota$ and output node $o$ such that $o$ is downstream from $\iota$. Suppose that $J$ is generic and all its self-couplings are not identically zero. Then the determinant of the corresponding homeostasis matrix $H$ is not identically zero.
\end{proposition}
\begin{proof} 
By Lemma \ref{jacobian_not_null} and \cite[Thm 2.4]{wang20}, it is enough to consider the case when $\mathcal{G}$ is a core input-output network with $n$ nodes.
Since the output node $o$ is downstream from the input node $\iota$ there is a directed path $S$ between them
\[
 \iota \to \sigma_1 \to \cdots \to \sigma_k \to o
\]
We can assume, without loss of generality, that the directed path $S$ is $\iota o$-simple.
Let $\sigma_{k+1},\ldots,\sigma_{n-2}$ be the remaining (regulatory) nodes of the network that are not in $S$.
Consider the following product of partial derivatives of a generic admissible vector field $f$:
\[
 \underbrace{f_{\sigma_{1},\iota} \ldots f_{o, \sigma_{k}}}_{k+1} \,
 \underbrace{f_{\sigma_{k+1}, \sigma_{k+1}} \ldots f_{\sigma_{n-2}, \sigma_{n-2}}}_{n-2-k} \not \equiv 0
\]
This product has $n-1$ factors, each one belonging to a unique row and column of the Jacobian matrix $J$. We claim that it is one of the summands of $\det H$. 
Indeed, the homeostasis matrix $H$ is obtained from the Jacobian matrix $J$ by dropping the first row and the last column, that is, elements of the form $f_{\iota,\cdot}$ and $f_{\cdot,o}$, respectively. Since none of these elements are in the product above, it follows that the claim is true.
Because $f$ (and hence $J$) is generic and all the self-couplings are not identically zero, the product is not identically zero.
Therefore, we conclude that $\det H$ cannot be identically zero.
\end{proof}

\begin{proposition} \label{weighted_det_H_not_null}
Let $\mathcal{G}$ be a network with input nodes $\iota_{1}, \cdots, \iota_{n}$ and output node $o$ such that $o$ is downstream from every input node. Suppose that $J$ is generic and all its self-couplings are not identically zero. Then the determinant of the corresponding weighted homeostasis matrix $\langle H \rangle$ is not identically zero.
\end{proposition}
\begin{proof}
By Lemma \ref{jacobian_not_null} and equation \eqref{relation_homeostasis_matrix_original_core}, it is enough to prove the claim when $\mathcal{G}$ is a core network. As shown in Subsection \ref{topological_characterization}, $\det \langle H \rangle$ can be factorized as a product of common factors between each of the determinants $\det H^{c}_{\iota_{m}}$ and the input counterweight term. By Proposition \ref{homeostasis_matrix_one_inpu_node_not_null}, none of these common factors can be identically zero. This means that $\det \langle H \rangle$ is identically zero if, and only if, the input counterweight term is identically zero. Recall that the input counterweight term is irreducible and its expression is given by eq.~\eqref{term_input_counter_weight_homeostasis}.
By hypothesis, $f_{\iota_{m}, \mathcal{I}}$ is always non-zero and Lemma \ref{jacobian_not_null} implies that the determinants $\det (J_{\mathcal{D}_{m}})$ must be all non-zero. Moreover, as each determinant $\det(C_{\iota_{m}})$ is a factor of $\det H^{c}_{\iota_{m}}$, then, by Proposition \ref{homeostasis_matrix_one_inpu_node_not_null}, we conclude that $\det(C_{\iota_{m}})$ cannot be identically zero for all $1 \leq m \leq n$. Hence, as $\det C$ is irreducible, we conclude that it cannot be identically zero. Therefore, $\det \langle H \rangle$ is not identically zero.
\end{proof}

\begin{remark} \normalfont
In the context of biological systems, it is usually expected that all nodes are self-coupled and that the output node is downstream all input nodes. Therefore, the assumptions made in this section are usually satisfied by almost all models. 
\end{remark}

\begin{remark} \normalfont
The results of this appendix cannot be improved. To see this, first consider the case when $\mathcal{G}$ is a network with only one input node $\iota$ and output node $o$. Then, if we drop the requirement that all nodes are self-coupled, there exists a network with identically zero homeostasis determinants (see Figure~\ref{non_null_condition} (a)). Moreover, if $o$ is not downstream from the input node $\iota$, then $\det H $ may also be identically zero (see Figure~\ref{non_null_condition} (b)). 
Now consider the case where $\mathcal{G}$ has multiple input nodes. Then $\det \langle H \rangle$ may be identically zero if we allow some node to be not self-coupled (see Figure~\ref{non_null_condition} (c)). Furthermore, if we require that $o$ be downstream from some, but not all input nodes, then, although $\det \langle H \rangle$ mat be not identically zero, the output node will not depend on the dynamics of all input nodes (see Figure~\ref{non_null_condition} (d)).  
\end{remark}

\begin{remark} \normalfont
The results of this appendix do not contradict the fact that the model for bacterial chemotaxis exhibits perfect homeostasis, that is, $\det \langle H \rangle \equiv 0$. 
This is because the model equations \eqref{original_e_coli} are not given by a generic vector field in the sense defined above.
Indeed, as it is shown in \eqref{analysis_e_coli_part_iiii} and
\eqref{analysis_e_coli_part_v}, there is a polynomial relation between certain partial derivatives of the vector field.
This is another manifestation of the fact that robust perfect homeostasis is a non-generic phenomenon (see Section \ref{sec:discussion}).
\end{remark}

\begin{figure}[!ht]
\centering
\includegraphics[width=\linewidth,trim=0cm 0cm 5cm 0cm,clip=true]%
{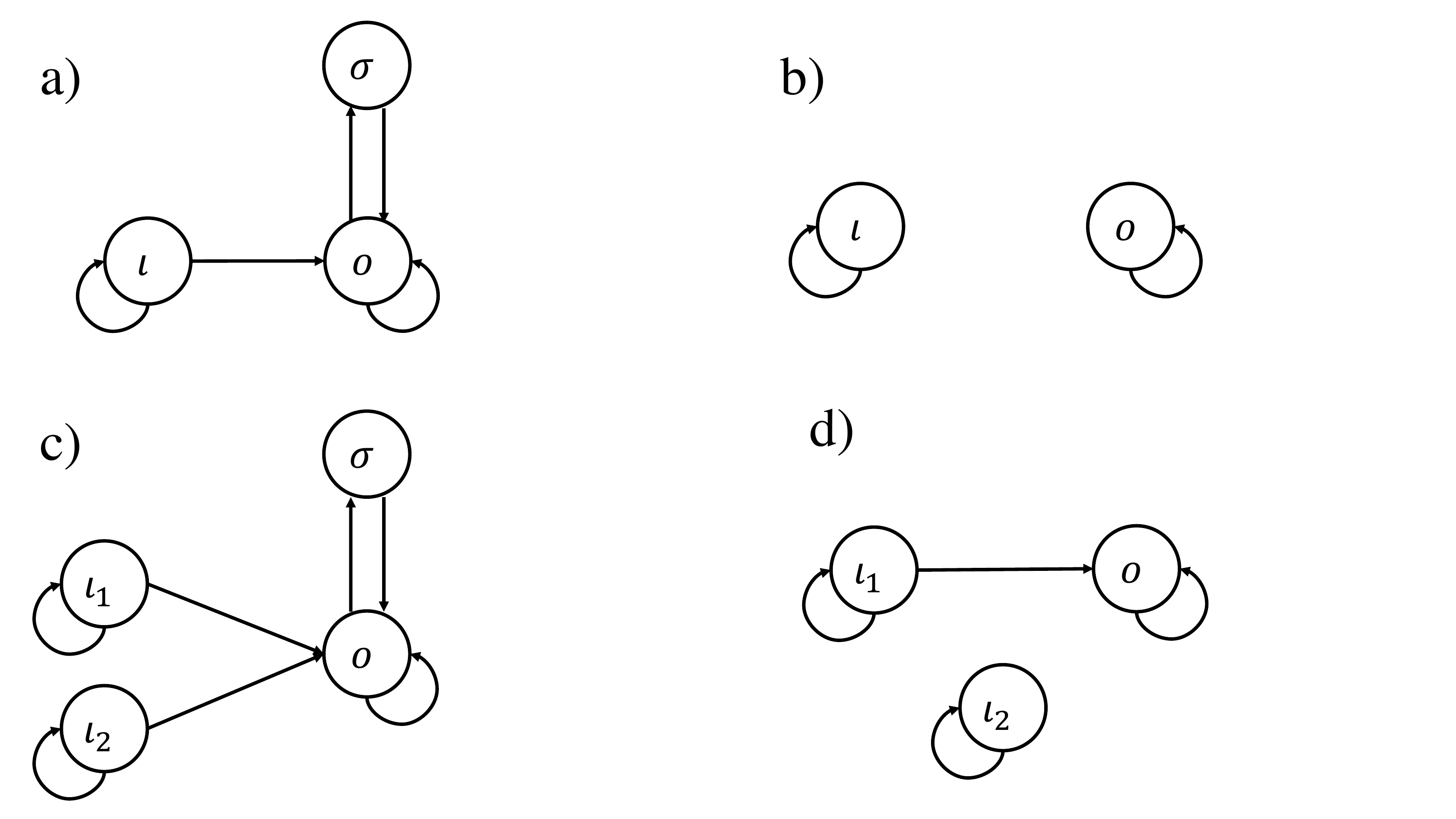}
\renewcommand{\figurename}{Figure}
\caption{\label{non_null_condition}
Examples of networks that do not satisfy the assumptions of non-zero self-couplings and/or the output node is not downstream from the input node(s). With the existence of nodes that are not self-coupled (cases (a) and (c)) the determinant of the (weighted) homeostasis matrix may vanish identically either in the case of only one input node (a) or when there exist multiple input nodes (c). In the case the output node is not downstream from the input node(s), in (b) it is trivial to see that $\det H 0$. In (d), the network has input nodes $\iota_{1}$ and $\iota_{2}$, and the output node $o$ is not downstream from $\iota_{2}$. In that case $\det H \not \equiv 0$, but the dynamics of $o$ does not depend on the dynamics of $\iota_{2}$. Finally, in all these examples the eigenvalues of the Jacobian are non-trivial, and therefore the network may display a linearly stable equilibrium, depending on the specific model that it represents. 
}
\end{figure}

%---------------END--------------

\bibliographystyle{abbrv}
%\nocite{*}
\bibliography{refs}

\end{document}